\documentclass[11pt]{article}

\usepackage{authblk} 
\usepackage{hyperref} 
\usepackage{abstract} 

\usepackage{fullpage} 

\usepackage{xcolor}

\usepackage{amsfonts,amsmath,amsthm,mathtools,amssymb} 
\usepackage{physics} 
\usepackage{bbm} 
\usepackage{appendix} 

\usepackage[utf8]{inputenc}
\usepackage[T1]{fontenc}

\renewcommand{\braket}[2]{\left\langle #1 , #2 \right\rangle}
\newcommand{\sbra}[1]{\langle #1 |}
\newcommand{\sket}[1]{| #1 \rangle}

\newcommand{\HNR}{H_{N,R}}

\newcommand{\hc}{\mathrm{h.c.}}

\renewcommand{\d}[1]{\mathrm{d}#1}

\newcommand{\one}{\mathbbm{1}}

\def\ii{\mathrm{i}}
\def\im{\mathrm{i}}

\def\R{\mathbb{R}}

\theoremstyle{plain} 
\newtheorem{theorem}{Theorem}[section]
\newtheorem{assumption}[theorem]{Assumption}
\newtheorem{proposition}[theorem]{Proposition}
\newtheorem{lemma}[theorem]{Lemma}
\newtheorem{corollary}[theorem]{Corollary}

\theoremstyle{remark} 
\newtheorem{remark}[theorem]{Remark}

\theoremstyle{definition} 
\newtheorem{definition}{Definition}

\numberwithin{equation}{section}

\def\b{|}

\title{Derivation of the Hartree dynamics for a $N$-particles fermionic $2D$ system interacting via long-range electric and magnetic two-body singular potentials in a dilute regime}

\author{
	Th\'eotime Girardot%
	\thanks{GSSI, Via comunale, L’Aquila, Italy
	\href{mailto:theotime.girardot@gssi.it}{theotime.girardot@gssi.it}}
}
\date{}
\begin{document}
\maketitle
\begin{abstract}
We show that a $2$-dimensional system of $N$ fermions interacting through a pairwise electric and magnetic singular interactions with Slater initial data preserves its Slater structure over time when $N$ gets large. In other words, the wave function of the system can be approximated by a Slater determinant whose orbitals evolves according to a coupled system of $N$ Chern-Simons-Schr\"odinger type effective equations. The result holds in a dilute regime where the density is of order one on a large volume proportional to the number of particles. The pairwise magnetic feature of the system implies to deal with the diagonalization of three-body potentials which is the main mathematical innovation of the paper.
\end{abstract}

\tableofcontents







\section{Introduction}
Human sized physical systems involve a large number of particles. Whenever their microscopic description takes into account their interactions in a non negligible way, the time evolution of these systems, even for simple initial data, can become so complicated that it even escapes the scope of numerical analysis. One way to study them then consists in adopting a statistical description of the systems. The identification of a larger scale structure such as Bose-Einstein condensation for bosons or Slater product for fermions often allows to reduce the many degrees of freedom of the system in a, usually non-linear, effective equation with few degrees of freedom. In the case of quantum fermionic many-body systems, Hartree, Hartree-Fock, Vlasov or Thomas-Fermi are well-known effective models \cite{BarGolGotMau-03,Elg_Erd_Sch_Tze_04,FroKno-11,Nar_Hei_Sew_81,Bar_Gol_Got_Nor_03,Bar_Mau_04,BacBrePetPicTza-16}. The above effective equations have been mostly studied from the point of view of their ground state energy as well as for their dynamics in semi-classical type scaling limit. We refer to \cite{Fou_Lew_Sol_18,Lew_Mad_Pet_Tri_19,Fou_Mad_Pet_20,Gir_Rou_21} for results about the ground state energy and to \cite{BenPorSch-14,Ben_Por_Sch_14bis,Ben_Jak_Por_Saf_Sch_16,BenPorSch-15,Por_Rad_Saf_Sch_17, Hai_Por_Rex_20} for the dynamics in the semi-classical setting. The present paper also means to derive a Hartree type dynamical effective equation for a $N$-body system of interacting fermions but in a regime where the particles occupy a volume of order $N$. In this dilute regime, the system stays more quantum and the Vlasov and Thomas-Fermi equation fail to approximate the system. At the level of the ground state energy, this scaling is close to the dilute thermodynamic regime studied in \cite{Fal_Gia_Hain_Por_25,Lie_Sei_Sol_05}. At the level of the dynamics, this regime was introduced in  \cite{Pet_Pic_16} and stays, as far as we know, the only occurrence of such a scaling in the literature.  This scaling in relevant in the case of large molecules, where the size of the system is of order $N$ and where more quantum effects, as interferences of wave packets, can take place compared to a semi-classical regime. In  \cite{Pet_Pic_16}, the authors study a system of $N$ interacting fermions in three dimensional space. The potential is a singular electric interaction $|x|^{-s}$. The approximation by the effective equation then holds for some positive $s\leq s_*$ where $s_*$ depends on the constraint they put on the initial data. In this paper we study the magnetic version of this system in dimension two. Its magnetic features implies to deal with a three-body interaction and to control the kinetic energy of the system along time. 
\subsection{The Hamiltonian}
We study the time evolution of an $L_{asym}^2(\R^{2N})$ product form initial data $\Psi_N(0)=\bigwedge_{j=1}^N u_j(0)$ where $\{u_j(0)\}_{j=1}^N$ is a orthonormal family of functions belonging to $L^2(\R^2)$ and whose dynamics is driven by the Schr\"odinger equation 
\begin{equation}\label{def:schro1}
	\ii \partial_t \Psi_N(t) = \HNR \Psi_N(t).
\end{equation}  The Hamiltonian $H_{R,N}$ is defined in \eqref{HRN1}, below. Note that, in the above equation, we scale units such that $\hbar=1$. This choice and its implications will be discussed in Section \ref{sec:dis}. The aim of the paper is to demonstrate that, for $N$ large enough, the product structure of the initial data is conserved over time, i.e, $\Psi_N(t)\simeq\bigwedge_{j=1}^N u_j(t)$, where the $u_j(t)$'s are solution of a system of $N$ effective equations defined in Definition \ref{def:css} below with initial data $u_j(0)$. 
This effective equation is a Euler-Lagrange mean-field version of our microscopic Hamiltonian that we now precisely define. 
\begin{definition} [\textbf{Microscopic Hamiltonian}]\mbox{}\\
For any \textit{coupling constants} $g,\beta\in\R_+$, \textit{singularity of the electro-magnetic interaction} $0\leq s\leq 1$ and \textit{cut-off} $0\leq R\leq 1$, we define
\begin{equation}
 H_{N,R}=\sum_{j=1}^{N}\left(-\im\nabla_j+\beta N^{s/2-1}\sum_{\substack{k=1,\\ k\neq j}}^N\frac{(x_j -x_k)^{\perp}}{|x_j -x_k|_R^{1+s}}\right)^{2}+\frac{g}{2}\sum_{j=1}^N\sum_{\substack{k=1,\\ k\neq j}}^N\frac{N^{s-1}}{|x_j-x_k|_R^{2s}}
 \label{HRN1}
\end{equation}
where the cut-off distance is defined as $|x|_R:=\max\{|x|,R\}$.
\end{definition} The above Hamiltonian acts on $L_{asym}^2(\R^{2N})$ and describes $N$ fermions interacting through a pair magnetic and electric interactions. The origin of the scaling $N^{s/2 -1}$ is explained in the discussion of Section \ref{sec:dis}.
\begin{remark}
In the case $s=1$ and $g=0$ the Hamiltonian $H_{R,N}$ describes a system of $N$ almost-fermionic anyons of radius $R$ with statistical parameter $\alpha=\beta /\sqrt{N}$, see \cite[Remark 1.2]{Gir_Rou_21}. Note that the present scaling differs from the semi-classical one studied in \cite{Gir_Rou_21}. In this sub-cited study, the system is trapped in a volume of order one. In the present paper, we allow the gas to occupy a volume of order $N$, see Section \ref{sec:dis}. The best version (with the weakest assumptions) of the main Theorem \ref{thm:main} is only valid for $s<1/4$. Nevertheless, the freedom in the parameter $s$ clearly highlights (see Remarks \ref{rem:cor}, \ref{rem:s},\ref{rem:g} and \ref{rem:om}) what problems must be solved to reach $s=1$.   
\end{remark}
We will denote the potential vector of the model by
\begin{equation}
\mathbf{a}_R(x):=\frac{x^{\perp}}{|x|_R^{1+s}}\quad\text{and}\quad \mathbf{a}_N(x):= N^{s/2 -1}\mathbf{a}_R(x)
\end{equation}
creating the magnetic field
\begin{align}
\mathbf{B}(x)=
\begin{cases}
\frac{2}{ R^{1+s}}\quad &\text{for}\quad x\in B(0,R)\\
\frac{1-s}{|x|^{1+s}} \quad &\text{for}\quad x\in B^c(0,R)
\end{cases}
\end{align}
through which the particles pairwise interact. With the above conventions, the Hamiltonian takes the form
\begin{equation}
 H_{N,R}=\sum_{j=1}^{N}\left(-\im\nabla_j+\beta\sum_{\substack{k=1,\\ k\neq j}}^N \mathbf{a}_N(x_j-x_k)\right)^{2} +\frac{g}{2}N\sum_{j=1}^N\sum_{\substack{k=1,\\ k\neq j}}^N|\mathbf{a}_N(x_j -x_k)|^2.
 \label{HRN}
\end{equation}
 For any fixed $R>0$ and $0\leq s\leq 1$, the Hamiltonian is essentially self-adjoint on $L^{2}_{\mathrm{asym}}(\mathbb{R}^{2N})$ (see \cite[Theorem X.17]{ReeSim2} and \cite{AvrHerSim}).
However, this operator does not have a unique limit as $R\to 0$ and the Hamiltonian at $R=0$, discussed in \cite{LunSol-14} for $s=1$, is \emph{not} essentially self-adjoint (see for instance \cite{CorOdd,DabSto,AdaTet}). Nevertheless, in the joint limit $R\to 0$ and $N\to \infty $ we recover a unique well-defined, non-linear model that depends on how the limits  $R\to 0$ and $N\to \infty $ are precisely taken.

\subsection{Discussion about the regime}\label{sec:dis}
The Schr\"odinger equation \eqref{def:schro1} has been defined with $\hbar=1$. In a box of size one, Lieb-Thirring inequalities \cite{Lieb-76,Lie_Thi_75} would imply to deal with a kinetic energy of magnitude $N^2$. In our case, the regime is dilute as in the thermodynamic limit of \cite{Fal_Gia_Hain_Por_25,Lie_Sei_Sol_05}. In this case, $N,L\to \infty$ at the fix rate $\rho=NL^{-d}$. The volume $L^d$ is then proportional to $N$ and we have to think to our system as being in a box  $\Lambda=[0,\sqrt{N}]^2$. The mean-field description is then expected to hold due to the long-range feature of the interaction. The following paragraph provides a more detailed heuristics of this fact.
Let's consider our gas as being trapped in a ball $B(0,\sqrt{N})$ and released at $t=0$. The ground state of this system is then expected to have a kinetic energy of order
\begin{equation}
\sum_{\substack{k\in\Lambda^*\\|k|\lesssim 1}}k^2\simeq N
\end{equation}
due to the scaling of the box $\Lambda=[0,\sqrt{N}]^2$. We can then think of $\nabla$ to be of order $1$. The density $\rho$ being of mass $N$, Lieb-Thirring's inequality suggests that $\rho$ must be bounded by $\sim 1$ and spread on a volume of order $\sim N$. 
To understand the order of magnitude of the other terms, it is convenient to expand $\eqref{HRN}$ and to treat the summands separately
\begin{align}
 H_{N,R}&=\sum_{j=1}^{N}-\Delta_j 
+(gN+\beta^2)\sum_{j\neq k}\left|\mathbf{a}_N(x_{j}-x_{k})\right|^{2}\;\text{: Kinetic + Diagonal term }\nonumber \\
 &\quad+ \beta\sum_{j\neq k}\left(-\im\nabla_{j}\right).\mathbf{a}_N(x_{j}-x_{k})+\mathbf{a}_N(x_{j}-x_{k}).\left(-\im\nabla_{j}\right)\;\text{: Mixed two-body term }\nonumber\\
 &\quad+ \beta^2\sum_{j\neq k\neq l}\mathbf{a}_N(x_{j}-x_{k}).\mathbf{a}_N(x_{j}-x_{l})\;\text{: Three-body term }.
 \label{expanded_H}
 \end{align}
 We expect that, for large $N$, particles behave independently, only interacting with the average field created by the others. For $s<1$, we can compute that
\begin{align}
\sum_{k=1}^N|\mathbf{a}_N(x-x_k)|^2&\sim\int_{\mathbb{R}^{2}}|\mathbf{a}_N(x-y)|^2\rho(y)\mathrm{d}y\simeq \int_{B(0,\sqrt{N})}|\mathbf{a}_N(x)|^2\mathrm{d}x\sim N^{-1}\nonumber\\
\sum_{k=1}^N\mathbf{a}_N(x-x_k)&\sim\int_{\mathbb{R}^{2}}\mathbf{a}_N(x-y)\rho(y)\mathrm{d}y\simeq \int_{B(0,\sqrt{N})}\mathbf{a}_N(x)\mathrm{d}x\sim 1.\nonumber
\end{align}
The mixed two-body term is a bit more complicated to estimate and reads
\begin{align}
\sum_{k=1}^N\mathbf{a}_R(x-x_k)\cdot \nabla_{x_k} &\sim\int_{\mathbb{R}^{2}}\mathbf{a}_R(x-y)\cdot \sum_{k=1}^N u_k(y)\nabla_y u_k(y)\mathrm{d}y\nonumber\\
&\sim\int_{\mathbb{R}^{2}}\mathbf{a}_N(x-y) \sqrt{\rho(y)} \sqrt{ \sum_{k=1}^N |\nabla_y u_k(y)|^2}\mathrm{d}y\nonumber\\
&\sim \left(\int_{B(0,\sqrt{N})}|\mathbf{a}_N(x)|^2\mathrm{d}x\right)^{1/2}\left(\int_{\R^2} \sum_{k=1}^N |\nabla_y u_k(y)|^2\mathrm{d}x\right)^{1/2}\sim  1\label{diss:regime}.
\end{align}
All terms in \eqref{expanded_H} are consequently of order $N$ except for the $\beta^2$ part of the diagonal one that will disappear in the limit. The above heuristics stays morally valid for $s=1$ up to logarithmic divergences. Even if the present paper will only allow $0\leq s<1/2$ in the best case scenario, it is reasonable to think that the case $s=1$ is reachable. Indeed, the main obstructions discussed in Remarks \ref{rem:cor}, \ref{rem:s},\ref{rem:g} and \ref{rem:om} have nothing to do with these logarithmic divergences and most of the lemmas used in the proof could accommodate it at the cost of slowest rates of convergence.  

\subsection{Hartree-Fock trial state}
In this section we introduce the product state $\Psi_N(0)=\bigwedge_{j=1}^N u_j(0)$ and compute its energy. We then show how the effective equation can be inferred from there. 
For a given $\Psi_N\in L^2_{\mathrm{as}}(\R^{2N})$ we define the density matrix $\gamma_{\Psi_N}$ as the operator having the kernel
\begin{equation}
\gamma_{\Psi_N}(x_1,\dots,x_N,y_1\dots,y_N)=\overline{\Psi_N(x_1,\dots,x_N)}\Psi_N(y_1,\dots,y_N)
\end{equation}
and the $k$-reduced density as 
\begin{equation}
\gamma_{\Psi_N}^{(k)}:=\tr_{[k+1,\dots,N]}\left[\gamma_{\Psi_N}\right]
\end{equation}
with kernel
\begin{align*}
\gamma_{\Psi_N}^{(k)}&(x_1,\dots,x_k,y_1\dots,y_k)=\\
&\int_{\R^{2(N-k)}}\overline{\Psi_N(x_1,\dots,x_k, u_{k+1},\dots,u_N)}\Psi_N(y_1,\dots,y_k, u_{k+1},\dots,u_N)\mathrm{d}u_{k+1}\dots\mathrm{d}u_{N}.
\end{align*}
We consider an orthonormal family $\{u_j\}_{j=1}^N$, where each $u_j\in L^2(\R^2)$ and define a Slater state to be
\begin{equation}
\Psi_N^{\mathrm{SL}}:=\bigwedge_{j=1}^N u_j(0)=\frac{1}{\sqrt{N!}}\sum_{\sigma\in S_N}{\mathrm{sgn}(\sigma)}\bigotimes_{j=1}^Nu_{\sigma(j)}(x_j)
\end{equation}
where the sign of the permutation is $\mathrm{sgn}(\sigma):=(-1)^m$ with $m$ the number of transpositions composing $\sigma$ and $\sigma=(\sigma(1),\dots,\sigma(N))$ an element of the symmetric group of $N$ objects, denoted $S_N$. For the $k$-reduced density matrix of a Slater, we use the notation
\begin{equation}
\omega_N^{(k)}:=\gamma_{\Psi^{\mathrm{SL}}_N}^{(k)}.
\end{equation}
\begin{proposition}[\textbf{Reduced density matrices}]\mbox{}\\\label{pro:sla}The three first reduced densities of a Slater determinant have the kernels
\begin{align}\label{def:matdenslater}
\omega_N^{(1)}(x_1;y_1)&=\frac{1}{N}\sum_{j=1}^N\overline{\psi_j(x_1)}\psi_j(y_1)\\
\omega_N^{(2)}(x_1,x_2;y_1,y_2)&=\frac{1}{N(N-1)}\sum_{j=1}^N\sum_{k=1}^N\overline{\psi_j(x_1)\psi_k(x_2)}(\one -U_{jk})\psi_j(y_1)\psi_k(y_2)\\
\omega_N^{(3)}(x_1,x_2,x_3;y_1,y_2,y_3)&=\frac{1}{N(N-1)(N-2)}\sum_{j=1}^N\sum_{k=1}^N\sum_{m=1}^N\overline{\psi_j(x_1)\psi_k(x_2)\psi_m(x_3)}\nonumber\\
\times(\one -&U_{jk}-U_{jm}-U_{km}+U_{jk}U_{km}+U_{km}U_{jm})\psi_j(y_1)\psi_k(y_2)\psi_m(y_3)
\end{align}
where the operator $U_{ij}$ swaps the indices $i$ and $j$ such that $U_{ij}\psi_j(x)\psi_i(y)=\psi_i(x)\psi_j(y)$, see the proof \ref{proof:slater} in Appendix for more detail.
\end{proposition}

\begin{lemma}[\textbf{Hartree energy}]\mbox{}\\\label{lem:Ha} 
Assume that $\tr[-\Delta \omega_N^{(1)}]\lesssim N$. Take $R=N^{-r}$ and assume that $r,s\leq 1/2$, we have that
\begin{equation}
\big <\Psi_N^{\mathrm{SL}} ,H_{N,R}\Psi_N^{\mathrm{SL}}\big >=\mathcal{E}^{\mathrm{af}}_{R}[\omega_N^{(1)}]+O(N^{s(1+2r)})
\end{equation}
where, with the notation $ \mathbf{a}_N[\rho_{\omega}]:= \mathbf{a}_N\ast\rho_{\omega}$ and $ \mathbf{a}^2_N[\rho_{\omega}]:= |\mathbf{a}_N|^2\ast\rho_{\omega}$,
\begin{equation}
\mathcal{E}^{\mathrm{af}}_{R}[\omega]=\Tr\left [\left (-\im\nabla_1+\beta \mathbf{a}_N[\rho_{\omega}]\right )^{2}\omega\right ]+\frac{g}{2}N\Tr\left[\mathbf{a}^2_N[\rho_{\omega}]\omega\right]
\end{equation}
is called the Hartree one-particle density functional.
See the proof \ref{proof:Hartree} in Appendix.
\end{lemma}
\begin{remark}
It is worth to note that, in the case $g=0$, the above lemma stays true with a way better error of order $O(N^{\frac{s}{2}(1+2r)})$ that would allow to consider $s=1$. 
\end{remark}

\subsection{The effective equation $\mathrm{EQ}$}
Starting from the Hamiltonian $H_{R,N}$, the Hartree effective equation describing the dynamics of the orbitals whose the product state is made of, can be inferred by writing the BBGKY hierarchy satisfied by the reduced density matrices and by computing its formal limit.
We obtain a system of $N$ Chern--Simons--Schr\"odinger types equations with potential vector $\mathbf{a}_N$. These equations depend on the radius $R$, the singularity $s$, the number of particles $N$ and on the coupling constants $\beta$ and $g$ of the microscopic system.
\begin{definition}[\textbf{The effective equation}]\mbox{}\\
\label{def:css}
	We denote by $\text{EQ}(N,R,s,m,\omega(0))$,  the differential problem whose the unknowns are the functions  $\{u_j(t)\}_{j=1}^N$ with each $u_j:(\R_+\times \R^2)\mapsto \mathbb{C}$ satisfying the system of equations
	\begin{align}\label{eq:pilotu}
		\im \partial_{t}u_j 
		=\left (-\im\nabla + \beta\mathbf{a}_N\left [\rho\right ]\right )^{2}u_j +gN\mathbf{a}_N^2[\rho]u_j- \left [\beta\mathbf{a}_N\ast\left (2\beta \mathbf{a}_N\left [\rho\right ]\rho+ J \right )\right ]u_j,
		\end{align}
		with the notations
		\begin{align}
		 \omega(t)=\sum_{j=1}^N\ket{u_j(t)}\bra{u_j(t)},\,\,\rho(t)=\sum_{j=1}^N|u_j(t)|^2\quad \text{and}\quad J(t)=\sum_{j=1}^N\im(u_j(t)\nabla \overline{u}_j(t)-\overline{u}_j(t)\nabla u_j(t))\nonumber
	\end{align}
	and with the initial conditions $\{u_j(0)\}_{j=1}^N$ forming an orthonormal set of $L^2(\R^2)$ where each $u_j(0,x)\in H^m(\R^{2})$ for all $j\in[1,\dots,N]$.
	We define $\text{EQ}(s,\omega(0)):=\text{EQ}(N,R=0,s,m=2,\omega(0))$.
\end{definition}
We shall emphasize that there does not exist any literature about the above equation.  However, its bosonic version with $s=1$, $g=0$ and in which $\mathbf{a}_N=x^{\perp}/|x|^2$ has been studied \cite{Ber_Bou_Sau_95,Liu_Smi_Tat_14}  and shares many of its important features with $\text{EQ}$. In particular, the technics developed in \cite{Lim_17} are likely to apply in the case of $\text{EQ}$.
Except for basic results as Energy Conservation \ref{Thm:cons} and the a priori bounds of Lemma  \ref{lem:aprioriK}, the present article will not focus on the properties of the effective equation $\text{EQ}$ and postpone them to a possible later work. For that reason, the main results will be presented according to three sets of assumptions on the solutions of $\text{EQ}$. Each set of assumption will be larger, allowing higher values of the singularity $s$ of the potential $\mathbf{a}_N$ in Theorem \ref{thm:main}. We will need the notations
\begin{equation}
\rho_{\nabla}:=\sum_{j=1}^N|\nabla u_j|^2\quad \text{and}\quad\rho_{\Delta}:=\sum_{j=1}^N|\Delta u_j|^2
\end{equation}
that will be useful through the whole paper.
\begin{assumption}\label{ass:H2}
For any $0\leq s<1/2$ and $0\leq R\leq 1$, the system  $\mathrm{EQ}(N,R,s,\omega(0))$ admits a solution $\omega(t)$, where $u_j(t)\in C\left([0,T],H^2(\R^{2})\right)$. In this case $\norm{\rho(t)}^2_2\simeq \norm{\rho_{\nabla}(t)}_1\simeq \norm{\rho_{\Delta}(t)}_1\leq C N$ where $C$ does not depend on $R$.
\end{assumption}
\begin{remark}
This is our most general assumption. The control $\norm{\rho(t)}^2_2\simeq \norm{\rho_{\nabla}(t)}_1\simeq N$  is directly obtained from Energy Conservation \ref{Thm:cons} and Lieb-Thirring inequalities provided that the initial data satisfies $\braket{\Psi_N(0)}{H_{R,N}\Psi_N(0)}\leq CN$, see Lemma \ref{lem:aprioriK}. The counter part of being too general is the need of the regularization $R=N^{-r}$ with the constraint  $0\leq s(1+2r)+\frac{1}{2}r < \frac{1}{4}$. We leave the control $\norm{\rho_{\Delta}(t)}_1\leq C N$ as an assumption but notice that it could be obtained by adapting the techniques of \cite[Theorem 3.2]{Gir_Lee_25}.
\end{remark}
\begin{remark}
Assumption \ref{ass:H2} is not the minimal assumption for which we expect  $\mathrm{EQ}(N,s,m,\omega(0))$ to have a solution. Well-posedness in $H^1(\R^{2})$ seems likely if we extrapolate the results of \cite[Theorem 1.1]{Lim_17}. Actually, we expect global well-posedness in $H^m(\R^{2})$, $m\geq 1$ due to the presence of the positive coupling $g$ \cite{Lim_17}[Last statement of Theorem 1.1]. This motivates the content of our second assumption.
\end{remark}
\begin{assumption}\label{ass:Hm}
For any $0\leq s<1/2$ and $0\leq R\leq 1$, the system of equations $\mathrm{EQ}(N,R,s,m,\omega(0))$ admits a solution $\omega(t)$ where $u_j(t)\in C\left([0,T],H^m(\R^{2})\right)$, $m\geq 2$ with the property that $\norm{\rho_{\nabla}}_p\simeq \norm{\rho}_p\lesssim N^{1/p}$ for all $1\leq p\leq m+1$ and where $\norm{u_j}_{H^m}\simeq 1$.
\end{assumption}
\begin{remark}
This assumption offers new useful control on critical quantities gathered in Lemma \ref{lem:LRtrick}. The property $\norm{\rho}_p\lesssim N^{1/p}$ is not a strong assumption in the sense that it can be obtained via generalized Lieb-Thirring inequalities \cite[Theorem 2]{Ghi_88}, using that $u_j\in H^m$. The control $\norm{\rho_{\nabla}}_p\lesssim N^{1/p}$ is however a strong assumption. Its relevance is not guaranteed in any regime but would enlarge the validity of the theorem to $0\leq s\leq \frac{1}{2}(1-\frac{1}{m})$. Note that is would be true if $\{\nabla u_j(t)\}_{j=1}^N$ was an orthogonal family. 
\end{remark}
\begin{assumption}\label{ass:Hinfty}
For any $0\leq s<1/2$ and $0\leq R\leq 1$, the system  $\mathrm{EQ}(N,R,s,\omega(0))$ admits a solution $\omega(t)$ where $u_j(t)\in C\left([0,T],H^2(\R^{2})\right)$. Moreover, $\norm{\rho(t)}_{\infty}\simeq \norm{\rho_{\nabla}(t)}_{\infty}\lesssim C_{\infty}(t)$ where $C_{\infty}$ does not depend on $N$ nor $R$.
\end{assumption}
\begin{remark}
The largest range of validity $s=1/2-\varepsilon$, is obtained under this assumption. Note that Assumption \ref{ass:H2} is always contained in Assumptions  \ref{ass:Hm} and  \ref{ass:Hinfty}.  For $g=0$, our result could be extended to $0\leq s\leq 1$ under Assumption \ref{ass:Hinfty}. Remark \ref{rem:g} explains why we have to consider $g>0$.
\end{remark}

\subsection{Results and Remarks}

\begin{theorem}\label{thm:main}
Let $\Psi_N(t)$ be the solution of the Schr\"odinger equation \eqref{def:schro1} with initial data $\Psi_N(0)=\bigwedge_{j=1}^N u_j(0)$ where $\{u_j(0)\}_{j=1}^N$ is an orthonormal family of $L^2(\R^{2N})$. Let  $\omega_R(t)$ be (whenever it exists) the solution of $\mathrm{EQ}(N,R,s,m,\omega(0))$. Then there exist a  $C>0$ and a $\beta_*$ such that for any $R=N^{-r}$, any $0\leq \beta\leq \beta_*$ and any $g$ such that $g\geq 12\beta^2$, we have
\begin{equation}
\norm{\gamma^{(1)}_{\Psi_N}(t)-\omega_R(t)}_{\mathrm{tr}}\leq C\left(\frac{1}{\sqrt{N}}+\frac{1}{N^{1-s(1+2r)}}\right)e^{Ct}
\end{equation}
\begin{align}
\text{for any}\quad\begin{cases}
 0\leq s(1+2r)+1/2r<1/4, \text{ in the case of Assumption} \ref{ass:H2}\\
 0\leq s<1/2(1-1/m),\, \,\,\,\quad\quad\text{in the case of Assumption} \ref{ass:Hm}\\
  0\leq s<1/2,\,\, \,\,\,\,\quad\quad\quad\quad\quad\quad\text{in the case of Assumption} \ref{ass:Hinfty}\nonumber
\end{cases}
\end{align}
and all $0\leq t\leq T$.
\end{theorem}
\begin{proof}
Let $\alpha_m(t)$ be the square root of the number of excited particles "$\sqrt{\mathcal{N}_+}$" precisely defined in Equation \eqref{def:alpham}. We know, by \cite[Lemma 3.3]{Pet_Pic_16}, that 
\begin{equation}
\norm{\gamma^{(1)}_{\Psi_N}(t)-\omega_R(t)}_{\mathrm{tr}}\leq C\alpha_m(t).
\end{equation}
In view of a Gr\"onwall argument, we apply Lemma \ref{lem:dtam} and obtain the result of Equation \eqref{eq:mainGron} providing that
\begin{equation}
\left|\partial_t \alpha_m(t)\right|\leq C\norm{\nabla_2 q_2\Psi_N}^2+C\left(\alpha_m(t)+\frac{1}{\sqrt{N}}\right)
\end{equation}
where the projector $q_2$ is defined in Section \ref{sec:proj}.
We use Lemma \ref{lem:kinetic} to bound the kinetic energy of the excited particles
\begin{equation}
 \norm{\nabla_2 q_2\Psi_N}^2\leq C\frac{\left| E_N(t)-\mathcal{E}_R^{\mathrm{af}}[\omega(t)] \right|}{N} +C\left(\alpha_m(t)+\frac{1}{\sqrt{N}}\right)
\end{equation}
under the constraint $g>12\beta^2$. For $\beta\leq \beta_*$, we use the conservation of energy of Lemma \ref{Thm:cons} and the estimate of the Hartree energy proved in Lemma \ref{lem:Ha} to compute the difference of energies in the above to be of order $N^{s(1+2r)-1}$. We finally obtain
\begin{equation}
\left|\partial_t \alpha_m(t)\right|\leq C\left(\alpha_m(t) + N^{s(1+2r)-1}+N^{-1/2}\right)
\end{equation}
and the application of Gr\"onwall's Lemma \ref{lem:Gron} concludes the proof of the theorem.
\end{proof}
We can also take the cut-off $R\to 0$ in $\mathrm{EQ}$ and approximate $\gamma^{(1)}_{\Psi_N}(t)$ by the solution of $\mathrm{EQ}(s,\omega(0))$ as shown in the following corollary.
\begin{corollary}\label{cor:main} Let  $\omega(t)$ be the solution of $\mathrm{EQ}(N,s,m,\omega(0))$. 
Under the assumptions of Theorem \ref{thm:main} in the sub-case of Assumption \ref{ass:H2}, there exists a constant $C>0$ such that for all time $0\leq t\leq T$, we have
\begin{equation}
\norm{\gamma^{(1)}_{\Psi_N}(t)-\omega(t)}_{\mathrm{tr}}\leq C\left(\frac{1}{\sqrt{N}}+\frac{1}{N^{1-s(1+2r)}}+\frac{1}{N^{\left(1/2+ r-s(1+2r)\right)}}\right)e^{Ct}
\end{equation}
for any $s(1+2r)+1/2r<1/4$ and all $0\leq t\leq T$.
\end{corollary}

\begin{proof}
We apply Theorem \ref{thm:main} and Lemma \ref{lem:Conv_pilot} which is valid for any $0< R\leq 1$ and $0\leq s<1/4$, to get
\begin{align*}
\norm{\gamma^{(1)}_{\Psi_N}(t)-\omega_R(t)}_{\mathrm{tr}}&=\norm{\gamma^{(1)}_{\Psi_N}(t)-\omega(t)+\omega(t)-\omega_R(t)}_{\mathrm{tr}}\\
&\leq \norm{\gamma^{(1)}_{\Psi_N}(t)-\omega(t)}_{\mathrm{tr}}+\norm{\omega(t)-\omega_R(t)}_{\mathrm{tr}}\\
&\leq C\left(\frac{1}{\sqrt{N}}+\frac{1}{N^{1-s(1+2r)}}\right)e^{Ct}+CN^{s-1/2}R^{1-2s}e^{Ct}.
\end{align*}
Taking $R=N^{-r}$ concludes the proof of the corollary.
\end{proof}
\begin{remark}[\textbf{Novelty of the result}]\mbox{}\\\label{rem:nov}
We shall emphasize that the dilute regime we treat here is rare in the literature and has only been studied for a two-body electric interaction \cite{Pet_Pic_16}.  The case of a two-body magnetic interaction is usual in the field of anyons \cite{Lun_24,CorLunRou-16,Gir_20,LunRou,Gir_Rou_23,LarLun-16} but is new in the field of fermionic dynamics. The magnetic feature of the system creates two problems that our proof has to solve. First, the mixed two-body term is not a positive operator and mixes momentum and space operators. As in \cite{Gir_Lee_25}, the kinetic part of this operator implies to be able to control the kinetic energy of excited particles, see Lemma \ref{lem:kinetic}. There is then to adapt the $2$-body diagonalization procedure of \cite{Pet_Pic_16}. Second, the $3$-body term of $H_{R,N}$ implies to establish a \textit{general diagonalization procedure} (Lemma \ref{lem:diag3b}) for  $3$-body operators, see Section \ref{sec:strpr} for a more detailed explanation. The Lemma \ref{lem:diag3b} and its application in Section \ref{sec:MW} are the main mathematical novelty of the paper.
\end{remark}
\begin{remark}[\textbf{Validity of the Corollary}]\mbox{}\\\label{rem:cor}
Corollary \ref{cor:main} could be extended to larger $s$ in the case of Assumptions \ref{ass:Hm} and \ref{ass:Hinfty}. To this end, it would be needed to derive Lemmas \ref{lem:Rto0} and \ref{lem:Conv_pilot} using the control on $\norm{\rho}_p$ and $\norm{\rho_{\nabla}}_p$, $1\leq p\leq \infty$ the same way we used them in Lemma \ref{lem:LRtrick}. This would lead to a simplification of the sub-cited lemmas and would allow to reach the same ranges of $s$ as in Theorem \ref{thm:main}. In particular, the constraint on $r$ would disappear allowing to consider $R=0$.\end{remark}
\begin{remark}[\textbf{The cut-off $R$}]\mbox{}\\\label{rem:R}
The convergence rate of Theorem \ref{thm:main} limits the range of $R=N^{-r}$ we can choose to $s(1+2r)<1$. In the case of Assumption  \ref{ass:Hm} and \ref{ass:Hinfty}, the constraint only originates from the Hartree approximation of Lemma \ref{lem:Ha}. In the case of  Assumption  \ref{ass:H2}, the constraint $0\leq s(1+2r)+1/2r<1/4$ comes from the control of \eqref{ine:LR1} in Lemma \ref{lem:LRtrick}.
\end{remark}
\begin{remark}[\textbf{The general constraint $s<1/2$}]\mbox{}\\\label{rem:s}
The main limitation of Theorem \ref{thm:main} lies in the constraint $s<1/2$ which has many origins. In the case of Assumption \ref{ass:H2}, this control is (among other things discussed right after) needed to treat the current type terms $-\ii\nabla\cdot \mathbf{a}_N$ in $H_{R,N}$ as well as $J$ in $\mathrm{EQ}$. Under the two other assumptions,  the current terms can be handled more easily and this source of constraint disappears. Nevertheless, under any of our assumption, the control of $\norm{|\mathbf{a}_N|^4\ast \rho}_{\infty}$ is crucial to handle the two-body electric interaction of the Hamiltonian and requires $s<1/2$ through the whole proof.   
\end{remark}
\begin{remark}[\textbf{Necessity of the two-body interaction $g>0$}]\mbox{}\\
\label{rem:g}
Our proof relies on the fact that we are able to control the kinetic energy of the excited particles living outside the orbitals. Lemma \ref{lem:mf} allows to extract this quantity by splitting $H_{R,N}$ into two regions, one only acting on the $N$ orbitals, and another one acting outside. As explained with more detail in Remark \ref{rem:om} below and discussed in paragraph \ref{diss:regime}, our regime makes impossible to a priori control the operators of $H_{R,N}$ acting on the excited particles, see \eqref{ine:pos}. The key to deal with the excited part of $H_{R,N}$ is a positivity argument used in \ref{lem:kinetic} to simply drop this part. The $g$ two-body term is then here to compensate the possible negativity of the mixed-two-body term $(-\ii\nabla)\cdot \mathbf{a}_N+\mathbf{a}_N\cdot (-\ii\nabla)$. \end{remark}
\begin{remark}[\textbf{Orders of magnitude of the regime}]\mbox{}\\\label{rem:om}
The relevance of our regime is based on the fact that $\norm{\rho}_1=N\gg\norm{\rho}_2=\sqrt{N}$. This asymmetry is used to make the heuristic of \ref{diss:regime} concrete through the computations of Lemma \ref{lem:LRtrick}. All these estimates highly depend on the convolution structure of the quantities controlled in Lemma \ref{lem:LRtrick}, needed to take advantage of $\norm{\rho}_2=\sqrt{N}$ and provides a two and three-body interactions of order $\sim N$. This makes of the excited part of $H_{R,N}$, an operator difficult to control because of its a priori higher order of magnitude $\sim N^2$.
\end{remark}
\begin{remark}[\textbf{The anyons case $s=1$}]\mbox{}\\\label{rem:om}
The present inquiry suggests that, in the case of Assumption \ref{ass:H2}, the regime $s=1$ is out of reach. The main reason lies in the central mean-field cancelations of Lemma \ref{lem:dtam} on which we base our main Gr\"onwall argument. Because of the diagonalization procedure, the estimates of Lemma \ref{lem:dtam} implies to square all operators. We can however imagine that, in the case of Assumption \ref{ass:Hinfty}, if we could replace the positivity argument discussed in Remark \ref{rem:g} by another one allowing to take $g=0$, the whole proof would hold up to $s<1$ and we would have to deal with logarithmic divergences in the case $s=1$.   
\end{remark}
\begin{remark}[\textbf{Possible use of a cut-off to count the number of excited particles}]\mbox{}\\\label{rem:cutoff}
A method to obtain larger $s$ would be to use a cut-off in the number of particles operator $\alpha_m(t)$ as it was introduced in \cite[Lemma 7.1]{Pet_Pic_16}. In the present case, it would allow some improvement in the case of Lemma \ref{lem:dtam} but not for Lemma \ref{lem:mf} needed to control the kinetic energy of excited particles because none commutation is involved. We would then end with the same constraints in the final theorem.  
\end{remark}
\medskip\noindent\textbf{Acknowledgments.} 
I am grateful to the Gran Sasso Science Institute's math area for providing the financial support that makes this work possible. I also thank Jinyeop Lee and Daniele Ferretti for many discussions we had about fermions and their dynamics.

\section{Structure of the paper}

\subsection{Notations and basic lemmas}
The main Hilbert space for our computations is $L^2_{\mathrm{asym}}(\R^{2N})$. For that reason, we will always mean $\norm{\Psi}_{L^2_{\mathrm{asym}}(\R^{2N})}$ when writing $\norm{\Psi}$ and the different $L^p(\R^{2})$ norms will be denoted $\norm{u}_p$ with the $p$ subscrit. For the Sobolev norms $H^p(\R^{2})$, we write $H^p$ in the subscript. The symbol $C$ will denote any strictly positive constant possibly depending on fixed parameters as $\beta ,s$ or on $\rho(0)$ related quantities but never depending on $R,t$ or $N$. The same $C$ can denote different constants from a line to the other. We will use the symbol $A\lesssim B$ to say that there exists a constant $c>0$ such that  $A\leq c B$. We will also use the convention 
\begin{equation}
	\braket{\Psi_N}{A\Psi_N}:=\left\langle A\right\rangle_{\Psi_N}
\end{equation}
for self-adjoint operators $A$. 
We use the bracket representation of projection operator onto a function $f$ in the variable $x_j$, by writing $\sket{f}\sbra{f}_j$.
In the following computations we will make a constant use of the Cauchy-Schwarz type inequalities
 \begin{align}
\left|\sum_{j=1}^N|\braket{a_j}{b_j}|\right|^2\leq \sum_{j=1}^N\norm{a_j}^2\sum_{j=1}^N\norm{b_j}^2,\quad\text{and}\quad
\sum_{j=1}^N|\braket{\omega_j}{v\omega_j}|^2\leq \sum_{j=1}^N\braket{\omega_j}{v^2\omega_j}
\end{align} where $\{\omega_j\}_{j=1}^N$ is an orthonormal family and where $v$, $a_j$ and $b_j$ will take various values.
We will also apply Gr\"onwall's lemma that we state here for convenience of reading.
\begin{lemma}[\textbf{Gr\"onwall Lemma}]\mbox{}\\
\label{lem:Gron}
	Let a function $f(t)$ be such that $|\partial_t f(t)|\leq C(t)f(t) + \varepsilon(t)$, then
	\begin{align}\label{eq:Gron}
		f(t)\leq f(0)\,e^{\int_0^t C(\tau)\dd \tau} + \int_0^t \varepsilon(s) \, e^{\int_s^t C(\tau) \dd \tau} \dd s.
	\end{align}
\end{lemma}

\subsection{Structure of the proof}\label{sec:strpr}
In this section we first sketch the global strategy of the proof and then describe the structure of the paper. The present work extensively applies and adapts the techniques introduced in \cite{Pet_14,Pet_Pic_16,KnoPic-10,Gir_Lee_25}. To elude complicated definitions, we will introduce in this paragraph, the imprecise notions of excited number of particles, kinetic energy of excited particles and square root of excited particles that we denote $\mathcal{N}_+$, $\Delta\mathcal{N}_+$ and $\sqrt{\mathcal{N}_+}$. These objects are made precise in Section \ref{sec:proj}. 
When it comes to deriving Hartree dynamics, the most direct idea is, in general, to try to control the number of excited particles $\braket{\Psi_N(t)}{\mathcal{N}_+\Psi_N(t)}$ via a Gr\"onwall argument. This is not enough in our case. If we directly compute
\begin{equation}
\left|\partial_t \braket{\Psi_N(t)}{\mathcal{N}_+\Psi_N(t)}\right|\leq \braket{\Psi_N(t)}{\mathcal{N}_+\Psi_N(t)} +N^{-1/2} +\norm{\Delta \mathcal{N}_+\Psi_N }^2
\end{equation}
the kinetic energy of the excited particles $\norm{\Delta \mathcal{N}_+\Psi_N }^2$ appears. It originates from the magnetic nature of the Hamiltonian \eqref{HRN}. Using the approach of \cite{KnoPic-10}, we control this kinetic energy by Lemma \ref{lem:kinetic} which provides 
\begin{equation}\label{ine:deltaN}
\norm{\Delta \mathcal{N}_+\Psi_N }^2\leq \braket{\Psi_N(t)}{\sqrt{\mathcal{N}_+}\Psi_N(t)}+o_N(1)
\end{equation}
where the square root of the excitation number of particles, precisely defined in Equation \eqref{def:alpham}, appears. The technique is then to apply a Gr\"onwall's lemma (content of Lemma \ref{lem:dtam}) on this square root. It provides
\begin{equation}\label{ine:grosN}
\left|\partial_t \braket{\Psi_N(t)}{\sqrt{\mathcal{N}_+}\Psi_N(t)}\right|\leq \braket{\Psi_N(t)}{\sqrt{\mathcal{N}_+}\Psi_N(t)} +N^{-1/2} +\norm{\Delta \mathcal{N}_+\Psi_N }^2
\end{equation}
and we can close the Gr\"onwall thanks to \eqref{ine:deltaN} and $\mathcal{N}_+\leq \sqrt{\mathcal{N}_+}$. In order to derive \eqref{ine:deltaN} and \eqref{ine:grosN}, bounds of squares of the operators constituting $H_{R,N}$ are needed. The very basics estimates to derive them are obtained in Lemma \ref{lem:LRtrick} where different ranges of $s$ are exhibited depending on the regularity of the solution of the $R$-dependent effective equation $\mathrm{EQ}$ of Definition \eqref{def:css}. Once the result is obtained at fix $R$, Lemma \ref{lem:Conv_pilot} allows to take the limit $R\to 0$ in the case of $0\leq s<1/4$ only. To obtain the right-hand-side of \ref{ine:grosN}, we have to compute the difference between the $N$-body Hamiltonian projected onto the space of the $N$ orbitals and the $1$-body Hartree Hamiltonian. In this difference, the largest term, projection of $H_{R,N}$ on the span of the orbitals, is expected to cancel. We name it \textit{the mean-field cancelation}. The key to obtain this mean-field cancelation is a diagonalization procedure that we explain in more detail.\newline
\textbf{Diagonalization procedure:} given a two-body operator $v$ and a three-body operator $w$, the key point of the mean-field regime lies in making the equivalences 
\begin{align}
p_2v(x_1-x_2)p_2&\simeq v\ast \rho \nonumber\\
p_2p_3 w(x_1 , x_2 ,x_3)p_2p_3&\simeq  w\ast \rho \ast \rho\label{def:eq}
\end{align}
rigorous to recover the terms of Hartree energy. Here $p_i=\sum_{j=1}^N\ket{u_j}\bra{u_j}_i$ is a sum of projectors acting on the variable $x_i$ made of the orbitals $u_j$, solutions of the effective equation $\mathrm{EQ}$. For bosonic systems, the projector $p$ is rank-one. This makes of \eqref{def:eq} a direct identity and the mean-field cancelation follows. When it comes to fermions, we have to restrict the double sum
\begin{equation}
\sum_{i,j=1}^N\ket{u_j}\bra{u_j}_yv(x-y)\ket{u_i}\bra{u_i}_y
\end{equation}
to its diagonal terms only. This is done by seeing $\braket{u_j}{v u_i}$ as the elements of a $(N\times N)$ symmetric matrix that we can diagonalize. The treatment of the three-body term is more complicated. In the following proof, two cases appear. The first one is
\begin{align}\label{eq:pro}
p_2p_3 w(x_1 , x_2 ,x_3)p_2p_3=\left(p_2v(x_1 -x_2)p_2\right)\left(p_3v(x_1 -x_3)p_3\right)
\end{align}
and simply implies to apply the $2$-body operator diagonalization two times. The second one is
\begin{align}\label{eq:pro1}
p_2p_3 w(x_1 , x_2 ,x_3)p_2p_3=p_2p_3v(x_2 -x_1)v(x_2 -x_3)p_3p_2.
\end{align}
The above cannot be factorized as in \eqref{eq:pro} and requires a more involved procedure. The case of \eqref{eq:pro1} is treated in Corollary \ref{cor:diag3bprod} which is itself derived from Lemma \ref{lem:diag3b} treating the general case of a $3$-body operator with no a priori product structures as \eqref{eq:pro} and \eqref{eq:pro1}. This Lemma \ref{lem:diag3b} is the main mathematical novelty of the present article.\newline \textbf{Global structure of the paper:} 
 \begin{itemize}
 \item In Section \ref{sec:EQ}, we establish the minimal knowledge about the effective equation such as conservation of energy Lemma \ref{Thm:cons} and the a priori bounds of Lemma \ref{lem:aprioriK} used to control $\norm{\rho}_2$ and $\norm{\rho_{\nabla}}_1$. We also prove the convergence of $\mathrm{EQ}$ when $R\to 0$.
 \item In Section \ref{sec:proj} we provide a rigorous definition of the number of excited particles $\mathcal{N}_+^{\xi}$ to the power $\xi$ and derive its main properties. We also state the fundamental Pauli principle, Lemma \ref{lem:exclusion}, and useful technical identities used in the next sections.
 \item In Section \ref{sec:diag}, we describe the diagonalization procedure for two-body, Lemma \ref{lem:diag2b}, and three-body, Lemma \ref{lem:diag3b}, operators. The rest of the section is dedicated to technical estimates of our specific operators.
 \item Section \ref{sec:Gron} contains the main Gr\"onwall argument. Lemma \ref{lem:dtam} makes the bound precise and shows its dependence on the quantities of Lemma \ref{lem:LRtrick} and on the kinetic energy of excited particles. All the section is dedicated to the proof of  Lemma \ref{lem:dtam}. 
 \item Section \ref{sec:kin} is devoted to the control of the kinetic energy of excited particles. In a first Lemma \ref{lem:mf} we extract the kinetic energy of the excited particles by splitting the Hamiltonian $H_{R,N}$ into its low part acting on the span of the orbitals, and its hight part. We conclude with Lemma \ref{lem:kinetic} that uses a positivity argument to drop this hight part on which we don't have any a priori control.
 \end{itemize}

\subsection{Fundamental estimates}

In the next lemma we gather estimates on fundamental quantities appearing all along the paper.
We will state each of them with their own rate of validity in term of $s$ according to the possible Assumptions \eqref{ass:H2}, \eqref{ass:Hm} and \eqref{ass:Hinfty}. 
Recall the definition of the vector potential $$\mathbf{a}_N(x)=N^{s/2 -1}\frac{x^{\perp}}{|x|_R^{1+s}},\quad\text{for}\quad R=N^{-r}.$$
\begin{lemma}[\textbf{Long-range spliting}]\mbox{}\\
\label{lem:LRtrick}
Assume that $\norm{\rho}_1\simeq\norm{\rho_{\nabla}}_1\lesssim N$ and that $\norm{\rho}_p\lesssim N^{1/p}$ for a given $2\leq p<\infty$, we then have for $R\geq 0$, that
\begin{align}\label{ine:LRs} \begin{cases}
\norm{\mathbf{a}_N\ast \rho}_{\infty}&\lesssim 1,\\
\norm{|\mathbf{a}_N|\ast J}_{\infty}&\lesssim 1,\quad \text{ for any}\quad 0\leq s<\left(1-\frac{1}{p}\right),
\\
\norm{|\mathbf{a}_N|^2\ast\rho}_{\infty}&\lesssim N^{-1},\quad \end{cases}
\end{align}
and that
\begin{equation}\label{ine:a4} 
\norm{|\mathbf{a}_N|^4\ast\rho}_{\infty}\lesssim N^{-3},\quad \text{for any}\quad 0\leq s<\frac{1}{2}\left(1-\frac{1}{p}\right).
\end{equation}
 Assume moreover that $\norm{\rho_{\Delta}}_1\lesssim N$ and that $R=N^{-r}$, then
 \begin{align}\label{ine:LR1} 
\norm{\nabla |\mathbf{a}_N|^2\ast \sqrt{\rho}\sqrt{\rho_{\nabla}}}_{\infty}&\lesssim N^{-1},\quad \text{for any}\quad 0\leq s(1+2r)+\frac{1}{2}r \leq \frac{1}{4}\\
\norm{|\mathbf{a}_N|^2\ast \sqrt{\rho}\sqrt{\rho_{\Delta}}}_{\infty}&\lesssim N^{-1},\quad \text{for any}\quad 0\leq s<\left(1-\frac{1}{p}\right).\label{ine:LR2}
\end{align}
If we also assume that $\norm{\rho_{\nabla}}_p\lesssim N^{1/p}$, we derive for any $R\geq 0$,
\begin{align}\label{ine:LR3} 
\norm{|\mathbf{a}_N|^2\ast\rho_{\nabla}}_{\infty}&\lesssim N^{-1},\quad \text{for any}\quad 0\leq s<\left(1-\frac{1}{p}\right).
\end{align}
Finally, assume that $\norm{\rho_{\nabla}}_{\infty}\simeq \norm{\rho}_{\infty}\lesssim 1$, then all the above bounds hold for $0\leq s<1/2$.
\end{lemma}

\begin{proof}
Before entering the calculations, note that for any $1\leq q<+\infty$ we have that
\begin{equation}
\norm{\mathbf{a}_N\one_{B(0,\sqrt{N})}}_q\lesssim N^{s/2 -1}\left(\int_{0}^{\sqrt{N}}r^{1-qs}\mathrm{d}r\right)^{1/q}=N^{1/q -1}
\end{equation}
holding as long as $s<2/q$. In the rest of the present proof, we will always consider $p> 1$ and $1\leq q< \infty$ to be conjugated exponents $$1/q+1/p=1.$$ Note that $p$ will \textit{not always} refer to the $p$ of the assumptions $\norm{\rho_{\nabla}}_p\simeq\norm{\rho}_p\lesssim N^{1/p}$.
We decompose
\begin{align*}
\norm{\mathbf{a}_N\ast \rho}_{\infty}
&=\norm{\mathbf{a}_N\left(\one_{B(0,\sqrt{N})}+\one_{B(\sqrt{N},+\infty)}\right)\ast \rho}_{\infty}\\
&\leq  \norm{\mathbf{a}_N\one_{B(0,K)}}_{q}\norm{\rho}_p+\frac{N^{s/2 -1}}{K^{s}}\norm{\rho}_1\\
 &\lesssim N^{1/q -1}N^{1/p}+1 \\
 &\lesssim 1.
\end{align*}
This proves the first bound of the lemma. The second inequality of the lemma is   
\begin{align}
\norm{|\mathbf{a}_N|^2\ast \rho}_{\infty}
&\leq \norm{|\mathbf{a}_N|^2\one_{B(0,\sqrt{N})}}_{q}\norm{ \rho}_{p}+N^{s-2}N^{-s}\norm{\rho}_1\nonumber\\
&\lesssim  N^{1/p}\norm{\mathbf{a}_N\one_{B(0,\sqrt{N})}}_{2q}^2 +N^{-1}\nonumber\\
&\lesssim N^{-1}\label{ine:a2rho}
\end{align}
provided that $s<q$.
The same way, we prove that
\begin{align*}
\norm{|\mathbf{a}_N|^4\ast \rho}_{\infty}
&\leq \norm{|\mathbf{a}_N|^4\one_{B(0,\sqrt{N})}}_{q}\norm{ \rho}_{p}+N^{-2s}N^{2s-4}\norm{\rho}_1\\
&\lesssim  N^{1/p} \norm{\mathbf{a}_N\one_{B(0,\sqrt{N})}}^4_{4q}+ N^{-3}\\
&\lesssim N^{-3}
\end{align*}
for $s<\frac{1}{2q}$.
For the next one, notice that
\begin{equation}
\left|J\right|=\left|\sum_{j=1}^N\ii\left(u_j\nabla \overline{u}_i -\overline{u}_j\nabla u_j\right)\right|\lesssim \sum_{j=1}^N|u_j||\nabla \overline{u}_i| \leq \sqrt{\rho}\sqrt{\rho_{\nabla}}.
\end{equation}
We use it in the following computation, only valid for $1<p< 2$,
\begin{align*}
\norm{|\mathbf{a}_N|\ast J}_{\infty}
&\leq \norm{|\mathbf{a}_N|\one_{B(0,\sqrt{N})}\ast |J|}_{\infty}+ \norm{|\mathbf{a}_N|^2\one_{B^c(0,\sqrt{N})}\ast J}_{\infty}\\
&\leq \norm{|\mathbf{a}_R|\one_{B(0,\sqrt{N})}}_q\norm{J}_{p}+ N^{-1}\norm{ J}_{1}\\
&\lesssim N^{1/q -1}\norm{ \sqrt{\rho}\sqrt{\rho_{\nabla}}}_p+N^{-1}\norm{ \sqrt{\rho}\sqrt{\rho_{\nabla}}}_1\\
&\lesssim N^{1/q -1}\norm{ \rho}^{\frac{1}{2}}_{\frac{p}{2-p}}\norm{ \sqrt{\rho_{\nabla}}}_2+N^{-1}\norm{ \rho+\rho_{\nabla}}_1\\
&\lesssim N^{1/q -1}N^{\frac{2-p}{2p}}\norm{\rho_{\nabla}}^{\frac{1}{2}}_1+1\\
&\lesssim N^{1/q+1/p -1}+1\lesssim 1
\end{align*}
for any $s<2/q$. Note that, in the above $p$ did note denote the $\norm{\rho}_p$ used in the inequality and does not then directly appear in the constraint on $s$.
We similarly obtain that
\begin{align*}
\norm{|\mathbf{a}_N|^2\ast \sqrt{\rho}\sqrt{\rho_{\Delta}}}_{\infty}&\leq \norm{|\mathbf{a}_N|^2\one_{B(0,\sqrt{N})}\ast  \sqrt{\rho}\sqrt{\rho_{\Delta}}}_{\infty}+ \norm{|\mathbf{a}_N|^2\one_{B^c(0,\sqrt{N})}\ast  \sqrt{\rho}\sqrt{\rho_{\Delta}}}_{\infty}\\
&\leq \norm{|\mathbf{a}_N|^2\one_{B(0,\sqrt{N})}}_q\norm{ \sqrt{\rho}\sqrt{\rho_{\Delta}}}_{p}+ N^{-2}\norm{  \sqrt{\rho}\sqrt{\rho_{\Delta}}}_{1}\\
&\lesssim N^{1/q -2}\norm{ \rho}^{\frac{1}{2}}_{\frac{p}{2-p}}\norm{ \sqrt{\rho_{\Delta}}}_2+N^{-1}\\
&\lesssim N^{1/q+1/p -2}+N^{-1}\lesssim N^{-1}
\end{align*}
as long as $s<1/q$.
For proving the inequality \eqref{ine:LR1} we first compute 
 $$\nabla |\mathbf{a}_R|^2= \frac{2x \one_{B(0,R)} }{R^{2(1+s)}}+\frac{2x \one_{B^c(0,R)}}{|x|^{2(1+s)}}.$$
 We then obtain that
 \begin{align*}
\norm{\nabla |\mathbf{a}_N|^2\ast \sqrt{\rho}\sqrt{\rho_{\nabla}}}_{\infty}&\lesssim N^{s-2}\norm{R^{-2(1+s)}|x|\one_{B(0,R)}+|x|^{-1-2s}\one_{B(R,\sqrt{N})}}_4\norm{ \sqrt{\rho}\sqrt{\rho_{\nabla}}}_{4/3}\\
&\quad + N^{-5/2}\norm{  \sqrt{\rho}\sqrt{\rho_{\nabla}}}_{1}\\
&\lesssim N^{s-2}\left(\int_R^{\sqrt{N}}r^{-3-8s}\mathrm{d}r\right)^{1/4}\norm{ \rho}^{\frac{1}{2}}_{2}\norm{ \sqrt{\rho_{\nabla}}}_{2}\\
&\quad + N^{s-2}R^{-2(1+s)}\left(\int_0^{R}r^{5}\mathrm{d}r\right)^{1/4}\norm{ \rho}^{\frac{1}{2}}_{2}\norm{ \sqrt{\rho_{\nabla}}}_{2} +N^{-3/2}\\
&\lesssim N^{s-2}R^{-1/2 -2s}N^{3/4}+N^{-3/2}\lesssim N^{-1}
\end{align*}
provided that $N^{s+\frac{1}{2}r +2rs -\frac{1}{4}}\leq 1$, in other words, $s(1+2r)+\frac{1}{2}r \leq \frac{1}{4}$.
The last inequality \eqref{ine:LR3} is proven as \eqref{ine:a2rho} with $\rho_{\nabla}$ instead of $\rho$.
When $\norm{\rho_{\nabla}}_{\infty}\simeq \norm{\rho}_{\infty}\lesssim 1$ similar estimates to the above provide the last statement of the lemma.
\end{proof}

\section{The effective equation $\mathrm{EQ}$}\label{sec:EQ}
Recall the Definition \ref{def:css} of the effective equation $\mathrm{EQ}$. The energy of the system made of the $N$ orbitals $\{u_j(t)\}_{j=1}^N$ is the Hartree energy of Lemma \ref{lem:Ha} that can also be written as
\begin{equation}
\mathcal{E}_R^{\mathrm{af}}[\omega(t)]:=\sum_{j=1}^N\mathcal{E}_R^{\mathrm{af}}[u_j(t)]=\sum_{j=1}^N\int_{\R^2} |D_u^R u_j|^2+\frac{g}{2}N\int_{\R^2}\mathbf{a}_N^2[\rho]\rho.
\end{equation}
where we have defined the covariant derivative $D^R_u:=(-\im\nabla +\beta \mathbf{a}_N\left [\rho\right ])$.

\subsection{Conservation laws}
To express the conservation laws related to $\mathrm{EQ}$ in a more convenient way, we introduce the following notations. 
We will write $\mathrm{EQ}$ as 
\begin{equation}\im \partial_{t}u_j =\mathcal{H}_R u_j,
\label{def:HH}
\end{equation}
where we defined $\mathcal{H}_R:= (D_u^R)^2+\mathcal{A}_R$ the self-adjoint operator in which
\begin{align}
\mathcal{A}_R&:=- \beta \mathbf{a}_N\ast\left (2\beta \mathbf{a}_N\left [\rho\right ]\rho+\hbar J \right )+gN\mathbf{a}_N^2[\rho]\nonumber\\
&=:- \beta\mathbf{a}_N\ast \mathcal{J}_u^R +gN\mathbf{a}_N^2[\rho].
\end{align}
We define the current of the orbital $u_j$ as $\mathbf{J}_j^R:=\Re (\overline{u}_jD^R_u u_j)$. That way, a direct calculation provides the continuity equations
\begin{align*}
\partial_t \int|u_j|^2=&0\quad\Rightarrow\quad \partial_t \int\rho_u=0\\
\int \left(\partial_t |u_j|^2+\nabla\cdot\mathbf{J}_j^R\right)=&0\quad\Rightarrow\quad \int \left(\partial_t \rho_u+\nabla\cdot\mathcal{J}_u^R\right)=0
\end{align*}
obtained by using that $\nabla\cdot \Im(\overline{u}\nabla u)=\Im(\overline{u}\nabla^2 u)$.
\begin{lemma}[\textbf{Conservation of the energy}]\mbox{}\\\label{Thm:cons}
For any $\omega(t)$ solution of $\mathrm{EQ}(N,R,s,m,\omega(0))$ we have that
\begin{equation}
\mathcal{E}_R^{\mathrm{af}}[\omega(t)]=\mathcal{E}_R^{\mathrm{af}}[\omega(0)]
\end{equation}
for any $N, R, s, g, \beta, t\geq 0$.
\end{lemma}
\begin{proof}
In order to demonstrate the above lemma, we first notice the two useful identities
\begin{equation}\label{def:current}
2\sum_{j=1}^N\mathbf{J}_j^R=2\sum_{j=1}^N\Re(\overline{u}_jD_u^R u_j)=\mathcal{J}_u^R
\end{equation}
as well as
\begin{equation}\label{eq:nicepropjr}
	-\int_{\R^2}( \mathcal{J}_u^R\ast\mathbf{a}_N )\partial_t\rho_u= \int_{\R^2} \mathcal{J}_u^R\cdot(\mathbf{a}_N\ast \partial_t\rho_u)=\int_{\R^2} \mathcal{J}_u^R\cdot \partial_t( \mathbf{a}_N \left[\rho_u\right ]).
\end{equation}
Next, we multiply \eqref{eq:pilotu} by $\partial_t \overline{u}_j$ and take the real part, we obtain that for any orbital $j\in[1,\dots ,N]$,
	\begin{align}
		0 = 2\Re \im \partial_t |u_j|^2 &=  2\Re \int_{\R^2} \big(D_u^R u_j\big)^* D_u^R\partial_t u_j	+	\int_{\R^2} \mathcal{A}_R \partial_t |u_j|^2 \nonumber\\
		&=\partial_t \int_{\R^2} |D_u^R u_j|^2 -2 \int_{\R^2} \mathbf{J}_j^R \cdot \partial_t ( \beta \mathbf{a}_N\left [\rho_u\right ])+ \int_{\R^2} \mathcal{A}_R \partial_t |u_j|^2\label{eq:dtE}
	\end{align}
by using that
	\begin{align}
		\partial_t \int_{\R^2} |D_u^Ru_j|^2 
		&= 2 \Re \int_{\R^2} \big((-\im\nabla + \beta \mathbf{a}_N\left [\rho_u\right ])u_j\big)^* (-\im\nabla + \beta \mathbf{a}_N\left [\rho_u\right ])\partial_t u_j + 2 \int_{\R^2} \mathbf{J}_j^R \cdot \partial_t ( \beta \mathbf{a}_N\left [\rho_u\right ]). \nonumber 
	\end{align}
	Summing \eqref{eq:dtE} over all the orbitals, we find
\begin{align}
	0&=	\partial_t\mathcal{E}_R^{\mathrm{af}}[\omega(t)] - \int_{\R^2} \mathcal{J}_u^R \cdot \partial_t ( \beta \mathbf{a}_N\left [\rho_u\right ])- gN\int_{\R^2}\mathbf{a}_N^2[\rho_u]\partial_t\rho_u+ \int_{\R^2} \mathcal{A}_R \partial_t \rho_u\nonumber\\
		&= \partial_t\mathcal{E}_R^{\mathrm{af}}[\omega(t)]\nonumber
	\end{align}
	by the use of identities \eqref{eq:nicepropjr} and \eqref{def:current}.
\end{proof}

\subsection{A priori kinetic bounds}

\begin{lemma}[\textbf{Kinetic energy bounds}]\mbox{}\\\label{lem:aprioriK}
Let $\Psi_N(t)$ be the solution of the Schr\"odinger equation with initial data $\Psi_N(0)=\bigwedge_{j=1}^N u_j(0)$ for a given set  $\{u_j(t)\}_{j=1}^N$ solution of $\mathrm{EQ}(N,R,s,m,\omega(0))$. Let $\rho $ be the associated one-body density, we have that if there exists a constant $C>0$, such that
\begin{equation}\label{ine:ini}
\braket{\Psi_N(0)}{H_{N,R}\Psi_N(0)}\leq CN
\end{equation}
then there exist $C',\beta^*>0$ not depending on $s,t$ nor $N$ such that for all $0\leq \beta\leq \beta^*$, we have the controls
\begin{align}\label{ine:boundonK}
 \sum_{j=1}^N\int_{\R^2} |\nabla u_j(t)|^2\leq C'N\quad \text{and}\quad  \int_{\R^2} \rho(t)^2\leq C'N,
\end{align}
for any $t$, $s\leq 1/2$ and all $R=N^{-r}$, $0\leq  r\leq 1/2$.
\end{lemma}
\begin{proof}

Our initial data \eqref{ine:ini} combined with the energy conservation of Lemma \ref{Thm:cons} and the Hartree approximation of Lemma \ref{lem:Ha}, provides
\begin{align}
CN\geq\braket{\Psi_N(0)}{H_{N,R}\Psi_N(0)}&=\mathcal{E}^{\mathrm{af}}_{R}[\omega(0)]+O(N^{s(1+2r)})\nonumber\\
&=\mathcal{E}^{\mathrm{af}}_{R}[\omega(t)]+O(N^{s(1+2r)})\nonumber\\
&=\sum_{j=1}^N\int_{\R^2} |D_u^R u_j(t)|^2+N\frac{g}{2}\tr[\mathbf{a}_N^2[\rho_{u}(t)]\omega]+O(N^{s(1+2r)})\label{eq:cons}.
\end{align}
By a weighed Cauchy-Schwarz inequality we get
\begin{equation}\label{ine:wcs}
\left|-\im\nabla u_j(t)\cdot \mathbf{a}_N[\rho_{u}(t)]u_j \right|\leq \frac{1}{4} |\nabla u_j(t)|^2+4|\mathbf{a}_N[\rho_{u}(t)]|^2|u_j|^2
\end{equation}
and then, by opening the square $(-\im\nabla +\beta \mathbf{a}_N\left [\rho\right ])^2$, that
\begin{equation}
\left|D_u^R u_j(t)\right|^2\geq \frac{1}{2} |\nabla u_j(t)|^2-7\beta^2|\mathbf{a}_N[\rho_{u}(t)]|^2|u_j|^2.
\end{equation}
 Recall that $s,r\leq 1/2$. By Lieb-Thirring's inequality \cite{Lieb-76,Lie_Thi_75} and by dropping the positive $g$ term, we get
\begin{align}
CN&\geq\frac{1}{2}\sum_{j=1}^N\int_{\R^2} |\nabla u_j(t)|^2-7\beta^2\int_{\R^2}|\mathbf{a}_N[\rho_{u}(t)]|^2\rho_{u}(t)\label{ine:kinrho}\\
&\geq c\norm{\rho_{u}(t)}_2^2 -7\beta^2 N\norm{|\mathbf{a}_N|\ast\rho_{u}(t)}^2_{\infty}\nonumber\\
&\geq c\norm{\rho_{u}(t)}_2^2 -7\beta^2 N\norm{|\mathbf{a}_N|\one_{B(0,\sqrt{N})}\ast\rho_{u}(t)}^2_{\infty}-7\beta^2N^2N^{-s}N^{s-2}\nonumber\\
&\geq c\norm{\rho_{u}(t)}_2^2 -7\beta^2 N\norm{|\mathbf{a}_N|\one_{B(0,\sqrt{N})}}_2^2\norm{\rho_{u}(t)}^2_{2}-7\beta^2N\nonumber\\
&\geq (c-14\pi\beta^2)\norm{\rho_{u}(t)}_2^2 -7\beta^2N
\end{align}
 The above inequality provides 
 \begin{align}
 \norm{\rho_{u}(t)}_2^2\leq \frac{C-14\pi\beta^2}{c-7\beta^2}N\leq C'N
 \end{align}
 for $\beta$ small enough. We inject the above into \eqref{ine:kinrho} to bound $$N\norm{|\mathbf{a}_N|\ast\rho_{u}(t)}^2_{\infty}\leq N+ \norm{|\mathbf{a}_N|\one_{B(0,\sqrt{N})}}_2^2\norm{\rho_{u}(t)}^2_{2}\leq N+CN^{s-2}N^{2-s}N\leq CN$$ and obtain 
 \begin{align}
(C-\beta^2 C'')N&\geq\frac{1}{2}\sum_{j=1}^N\int_{\R^2} |\nabla u_j(t)|^2
\end{align}
which concludes the proof of \eqref{ine:boundonK} by taking $\beta\leq \beta_*$ for a $\beta_*$ small enough.
\end{proof}

\subsection{Convergence of the effective equation}
The rest of the section studies the behavior of $\mathrm{EQ}(N,R,s,m,\omega(0))$ when $R\to 0 $. For that reason, we recall the $R$-dependence of $\mathbf{a}_{R}=x^{\perp}/|x|_R^{1+s}$ and denote $\mathbf{a}_{N}^0:=N^{(s/2 -1)}x^{\perp}/|x|^{1+s}$.
\begin{lemma}[\textbf{Difference of vector potentials}]\mbox{}\\\label{lem:Rto0}
For any $0\leq R\leq 1$, any $\rho_u$ and $\rho_v$ with $\norm{\rho_{*}}_1=\norm{\rho_*}_2^2=N$, $*\in \{u,v\}$ and for any $0\leq s<1/4$, we have that
\begin{equation}
\norm{\mathbf{a}^0_N\left [\rho_v\right ]-\mathbf{a}_{N}\left [\rho_u\right ]}_{\infty}\lesssim N^{s/2 -1/2}R^{1-s}+N^{-1/2}\sqrt{\sum_{j=1}^N\norm{u_j -v_j}_2^2}.\label{ine:anr}
\end{equation}
\begin{equation}
\norm{|\mathbf{a}^0_N|^2\ast\rho_v -|\mathbf{a}_{N}|^2\ast\rho_u}_{\infty}\lesssim N^{-1}\left(N^{s -1/2}R^{1-2s}+N^{-1/2}\sqrt{\sum_{j=1}^N\norm{u_j -v_j}_2^2}\right).\label{ine:anr2}
\end{equation}
We note that the expression for the first term of the right-hand-side of the above inequality stays valid up to $s<1$ and the second, up to $s<1/2$. Moreover we can use \eqref{ine:anr} to control
\begin{equation}
 \norm{\mathbf{a}_{N}\ast  (\mathbf{a}_{N}\left [\rho_R\right ]\rho_R -\mathbf{a}^0_{N}\left [\rho\right ]\rho)}_{\infty}
	\lesssim \norm{\mathbf{a}^0_N\left [\rho_v\right ]-\mathbf{a}_{N}\left [\rho_u\right ]}_{\infty}.
\end{equation}
\end{lemma}
\begin{proof}
We start by proving \eqref{ine:anr}. We split the term into two differences
\begin{align}
\norm{\mathbf{a}^0_N\left [\rho_v\right ]-\mathbf{a}_{N}\left [\rho_u\right ]}_{\infty}&=\norm{(\mathbf{a}^0_N-\mathbf{a}_{N})\ast\rho_v-\mathbf{a}_{N}\ast(\rho_u-\rho_v)}_{\infty}\nonumber\\
&\leq a^{(1)}+a^{(2)}.
\end{align}
We have, for $0\leq s<1$, that
\begin{align}
a^{(1)}:=\norm{(\mathbf{a}^0_N-\mathbf{a}_{N})\ast\rho_v}_{\infty}&\leq N^{-1+s/2}\norm{|x|^{-s}\one_{B(0,R)}\ast\rho_v}_{\infty}\nonumber\\
&\leq N^{-1+s/2}\norm{\rho_v}_{2}\left(\int_0^R r^{1-2s}\mathrm{d}r\right)^{1/2}\nonumber\\
&\leq N^{-1/2+s/2}R^{1-s}.
\end{align}
In another hand, for $0\leq s<1/2$,
\begin{align}
a^{(2)}&:=\norm{\mathbf{a}_{N}\ast(\rho_u-\rho_v)}_{\infty}\nonumber\\
&\leq N^{-1+s/2}\norm{|x|^{-s}\ast\sum_{j=1}^N(u_j -v_j)\overline{u_j}-v_j(\overline{v}_j -\overline{u}_j)}_{\infty}\nonumber\\
&\lesssim N^{-1+s/2}\norm{|x|^{-s}\one_{B(0,\sqrt{N})}\ast \sqrt{\rho}\sqrt{\sum_{j=1}^N|u_j -v_j|^2}}_{\infty}+N^{-1+s/2}N^{-s/2}\norm{ \sqrt{\rho}\sqrt{\sum_{j=1}^N|u_j -v_j|^2}}_{1}\nonumber
\end{align}
and by Young's inequality for the convolution
\begin{align}
a^{(2)}&\lesssim N^{-1+s/2}\norm{|x|^{-s}\one_{B(0,\sqrt{N})}}_4\norm{ \sqrt{\rho}\sqrt{\sum_{j=1}^N|u_j -v_j|^2}}_{4/3}+N^{-1}\norm{ \sqrt{\rho}}_2\sqrt{\sum_{j=1}^N\norm{u_j -v_j}_2^2}\nonumber\\
&\lesssim \left(N^{-1+s/2}\left(\int_0^{\sqrt{N}}r^{1-4s}\mathrm{d}r\right)^{1/4}\norm{ \sqrt{\rho}}_4+N^{-1}\sqrt{N}\right)\sqrt{\sum_{j=1}^N\norm{u_j -v_j}_2^2}\nonumber\\
&\lesssim N^{-1/2}\sqrt{\sum_{j=1}^N\norm{u_j -v_j}_2^2}.\label{ine:a2}
\end{align}
This concludes the proof of \eqref{ine:anr}. For the proof of \eqref{ine:anr2}, we proceed the same way with $|x|^{-2s}$ instead of $|x|^{-s}$, namely
\begin{align}
\norm{|\mathbf{a}^0_N|^2\ast\rho_v -|\mathbf{a}_{N}|^2\ast\rho_u}_{\infty}&=\norm{(|\mathbf{a}^0_N|^2-|\mathbf{a}_{N}|^2)\ast\rho_v-|\mathbf{a}_{N}|^2\ast(\rho_u-\rho_v)}_{\infty}\nonumber\\
&\leq b^{(1)}+b^{(2)}.
\end{align}
We then have, for $0\leq s<1/2$, that
\begin{align}
b^{(1)}:=\norm{(|\mathbf{a}^0_N|^2-|\mathbf{a}_{N}|^2)\ast\rho_v}_{\infty}&\leq N^{-2+s}\norm{|x|^{-2s}\one_{B(0,R)}\ast\rho_v}_{\infty}\nonumber\\
&\leq N^{-2+s}\norm{\rho_v}_{2}\left(\int_0^R r^{1-4s}\mathrm{d}r\right)^{1/2}\nonumber\\
&\leq N^{-2+s}R^{1-2s}.
\end{align}
In another hand, for $0\leq s<1/4$,
\begin{align}
b^{(2)}&:=\norm{|\mathbf{a}_{N}|^2\ast(\rho_u-\rho_v)}_{\infty}\nonumber\\
&\lesssim N^{-1+s/2}\norm{|x|^{-2s}\one_{B(0,\sqrt{N})}}_4\norm{ \sqrt{\rho}\sqrt{\sum_{j=1}^N|u_j -v_j|^2}}_{4/3}+N^{-1}\norm{ \sqrt{\rho}}_2\sqrt{\sum_{j=1}^N\norm{u_j -v_j}_2^2}\nonumber\\
&\lesssim N^{-1}N^{-1/2}\sqrt{\sum_{j=1}^N\norm{u_j -v_j}_2^2}.\label{ine:b2}
\end{align}
which is the exact same calculation as $a^{(2)}$ with $|x|^{-2s}$ instead of $|x|^{-s}$.
We prove the last inequality as follows
\begin{align*}
	 \norm{\mathbf{a}_{N}\ast  (\mathbf{a}_{N}\left [\rho_u\right ]\rho_u -\mathbf{a}^0_{N}\left [\rho_v\right ]\rho_v)}_{\infty}
	&\lesssim\norm{\mathbf{a}_{N}\ast(\mathbf{a}_{N}\left [\rho_u\right ] -\mathbf{a}^0_{N}\left [\rho_v\right ])\rho_v}_{\infty}\\
	&\quad\quad + \norm{\mathbf{a}_{N}\ast(\mathbf{a}^0_{N}\left [\rho_v\right ](\rho_v-\rho_u))}_{\infty}\\
	&\lesssim \norm{\mathbf{a}_{N}\left [\rho_u\right ] -\mathbf{a}_{N}\left [\rho_v\right ]}_{\infty}\norm{|\mathbf{a}_{N}|\ast \rho_v}_{\infty}\\
	&\quad\quad +  \norm{\mathbf{a}_{N}\ast|\rho_v-\rho_u|}_{\infty}\norm{\mathbf{a}^0_{N}[\rho_v]}_{\infty}\\
	&\lesssim N^{s/2 -1/2}R^{1-s}+2N^{-1/2}\sqrt{\sum_{j=1}^N\norm{u_j -v_j}_2^2}.
	\end{align*}
	using the estimate \eqref{ine:anr} and the fact that $\norm{\mathbf{a}_{N}\ast(\rho_v-\rho_u)}_{\infty}=a^{(2)}$ has been bounded in \eqref{ine:a2}.
\end{proof}
\begin{lemma}[\textbf{Convergence of $\mathrm{EQ}$ in the case of Assumption \ref{ass:H2}}]\label{lem:Conv_pilot}
	Let $\omega(0)$ be an initial condition system, define $\omega^R(t)$ to be the solution of $\mathrm{EQ}(R,\omega(0))$ and $\omega(t)$ to be the solution of $\mathrm{EQ}(\omega(0))$. There exist $C>0$ such that for any $0<R\leq 1$, $0\leq s<1/4$ and any $t\in [0,T]$, we have
	\begin{equation}
		\left|\frac{1}{N}\partial_t\sum_{j=1}^N\norm{u_j -u^R_j}^2_2\right|\lesssim  N^{2s-1}R^{2-4s}+\frac{1}{N}\sum_{j=1}^N\norm{ u_j-u_j^R}_2^2
	\end{equation}
	which implies
	\begin{equation}
		\frac{1}{N}\sum_{j=1}^N\norm{u_j -u^R_j}^2_2\leq  CN^{2s-1}R^{2-4s}e^{Ct}
			\end{equation}
	by Gr\"onwall's lemma~\ref{lem:Gron}. We deduce from the above the convergence of the solutions
	\begin{align}
	\tr\left|\omega_R(t) -\omega(t)\right|\leq \frac{1}{N}\sum_{j=1}^N\norm{u_j -u^R_j}_2\leq CN^{s-1/2}R^{1-2s}e^{Ct}.
	\end{align}
\end{lemma}
\begin{proof} 
Recall that $\omega=\sum_{j=1}^N\ket{u_j}\bra{u_j}$. We denote $\omega_R=\sum_{j=1}^N\ket{u^R_j}\bra{u^R_j}$
together with the associated densities $$\rho:=\sum_{j=1}^N|u_j|^2,\quad\text{and}\quad \rho_R:=\sum_{j=1}^N|u^R_j|^2$$ together with the currents $$J:=\sum_{j=1}^N\im (u_j\nabla \overline{u}_j-\overline{u}_j\nabla u_j),\quad J_R:=\sum_{j=1}^N\im (u^R_j\nabla \overline{u}^R_j-\overline{u}^R_j\nabla u^R_j).$$
	Let us also denote $$\mathcal{E}_{R} := (-\im\nabla + \mathbf{a}_N[\rho_R])^2\quad\text{and}\quad \mathcal{A}_R:=\left [\mathbf{a}_N\ast\left (2\mathbf{a}_N\left [\rho_R\right ]\rho_R+ J_R \right )\right ]+gN\mathbf{a}^2_N[\rho_R]$$ such that we can write  $\mathrm{EQ}(N,R,s,\omega(0))$ as
	\begin{equation}
	\im \partial_{t}u_j^R=(\mathcal{E}_{R}+\mathcal{A}_{R})u_j^R
	\end{equation}
	As implicitly used in the above, we apply the convention of droping the symbol $R$ from the notation when it becomes zero, i.e $\mathcal{E}:=\mathcal{E}_{0}$, $\mathcal{A}:=\mathcal{A}_{0}$. We are ready to compute
	\begin{align}
		\partial_t \norm{u_j -u^R_j}_2^2 
		&= 2\Im \left[\bra{(\mathcal{E} + \mathcal{A})u_j}\ket{u^R_j} 
		+\bra{(\mathcal{E}_R + \mathcal{A}_R)u^R_j}\ket{u_j}\right]\nonumber \\
		&= 2\Im \left[\bra{(\mathcal{E} + \mathcal{A})u_j}\ket{(u_j^R -u_j)} 
		- \bra{(\mathcal{E}_R + \mathcal{A}_R)u^R_j}\ket{(u_j^R -u_j)}\right]\nonumber\\
		&=:E_j+A_j
		\label{eq:1bd-diff}
	\end{align}
	where we used the fact that both $\mathcal{E}_R$ and $ \mathcal{A}_R$ are self adjoint for all $R\geq 0$.
	We will calculate separately the $E$ part and the $A$ part defined as
		\begin{align}
		E_j:&=2\Im \bra{(\mathcal{E}_R-\mathcal{E})u_j^R}\ket{(u_j-u_j^R)}\nonumber\\
		A_j:&=2\Im \bra{(\mathcal{A}_R-\mathcal{A}^0)u_j^R}\ket{(u_j-u_j^R)}.\nonumber
	\end{align}
Before entering the computation, recall some basic bounds for our setting, such as
\begin{align}
\norm{\rho_R}_1\simeq \norm{\rho_R}_2^2\lesssim N,\quad \sum_{j=1}^N\norm{\nabla u_j^R}_2^2\lesssim N,\quad \norm{J_R}_{4/3}\lesssim N^{3/4}, \quad\text{and}\quad \norm{\mathbf{a}_{N}\left [\rho_R\right ]}_{\infty}\lesssim 1.\nonumber
\end{align}
	We begin with $E_j$ and estimate
	\begin{align*}
		\left|\sum_{j=1}^NE_j\right|&\leq 2\left |\bra{\sum_{j=1}^N-\im\nabla u_j^R}\ket{(\mathbf{a}^0_N\left [\rho\right ]-\mathbf{a}_{N}\left [\rho_R\right ])(u_j^R-u_j)}\right|\\
		&\quad\quad+\left|\bra{\sum_{j=1}^N(\mathbf{a}^0_N\left [\rho\right ]+\mathbf{a}_{N}\left [\rho_R\right ])u_j^R}\ket{(\mathbf{a}^0_N\left [\rho\right ]-\mathbf{a}_{N}\left [\rho_R\right ])(u_j^R-u_j)}\right|\\
		&\lesssim \left |\bra{\sqrt{\sum_{j=1}^N|\nabla u_j^R|^2}}\ket{(\mathbf{a}^0_N\left [\rho\right ]-\mathbf{a}_{N}\left [\rho_R\right ])\sqrt{\sum_{j=1}^N|u_j^R-u_j|^2}}\right|\\
		&\quad\quad
		+\norm{\mathbf{a}^0_N\left [\rho\right ]}_{\infty}\left|\bra{\sqrt{\rho_R}}\ket{(\mathbf{a}^0_N\left [\rho\right ]-\mathbf{a}_{N}\left [\rho_R\right ])\sqrt{\sum_{j=1}^N|u_j^R-u_j|^2}}\right|,\\
		\end{align*}
		and by Cauchy-Schwarz's inequality, we obtain
		\begin{align*}
		&\lesssim\norm{\mathbf{a}^0_N\left [\rho\right ]-\mathbf{a}_{N}\left [\rho_R\right ]}_{\infty}\sqrt{\sum_{j=1}^N\norm{\nabla u_j^R}_2^2}\sqrt{\sum_{j=1}^N\norm{u_j^R-u_j}_2^2}   \\
		&\quad\quad +\norm{\mathbf{a}^0_N\left [\rho\right ]-\mathbf{a}_{N}\left [\rho_R\right ]}_{\infty}\sqrt{\sum_{j=1}^N\norm{u_j^R-u_j}_2^2}\sqrt{\norm{\rho_R}_1 }  \\
		&\lesssim N\sqrt{\sum_{j=1}^N\norm{u_j^R-u_j}_2^2}\left(N^{s/2 -1}R^{1-s}+N^{-1}\sqrt{\sum_{j=1}^N\norm{u_j -u^R_j}_2^2}
\right)   \\
\end{align*}
for $s<1/2$ by Lemma \ref{lem:Rto0}.  We then obtain
	\begin{align}
		N^{-1}\left|\sum_{j=1}^NE_j\right|
		&\lesssim N^{s -1}R^{2-2s}+N^{-1}\sum_{j=1}^N\norm{u_j -u^R_j}_2^2  .\label{ine:E}
\end{align}
We turn to $A_j$ that we split into four differences
\begin{align}
		A_j&= \bra{(\mathcal{A}-\mathcal{A}_R)u^R_j}\ket{(u_j-u^R_j)}\nonumber\\
		&=  \bra{(\mathbf{a}^0_{N} -\mathbf{a}_{N})\ast(2 \mathbf{a}^0_{N}\left [\rho \right ] \rho+J[u])}\ket{\overline{u_j}^R(u_j-u^R_j)}\nonumber\\
		&\quad +  \bra{\mathbf{a}_{N}\ast 2 (\mathbf{a}_{N}\left [\rho_R\right ]\rho_R -\mathbf{a}^0_{N}\left [\rho\right ]\rho)}\ket{\overline{ u_j}^R( u_j- u^R_j)}\nonumber\\
		&\quad +  \bra{\mathbf{a}_{N}\ast\left [J[u]-J[u^R]\right]}\ket{\overline{ u_j}^R( u_j- u^R_j)}\nonumber\\
		&\quad + gN\braket{|\mathbf{a}^0_N|^2\ast\rho -|\mathbf{a}_{N}|^2\ast\rho_R}{\overline{ u_j}^R( u_j- u^R_j)}\nonumber\\
		&:= A_j^{(1)}+A_j^{(2)}+A_j^{(3)}+A_j^{(4)}.\label{eq:defM}
	\end{align}	
	The first one can be estimated as
\begin{align*}
	\left|\sum_{j=1}^N	A^{(1)}_j\right|&\leq  	\left|  \bra{|\mathbf{a}^0_{N} -\mathbf{a}_{N}|\ast|2 \mathbf{a}^0_{N}\left [\rho \right ] \rho+J[u]|}\ket{\sqrt{\rho_R}\sqrt{\sum_{j=1}^N|u_j-u^R_j|^2}}\right|\nonumber\\
	&\leq \sqrt{\sum_{j=1}^N\norm{u_j -u^R_j}_2^2}\norm{\sqrt{\rho_R}}_4\norm{|\mathbf{a}^0_{N} -\mathbf{a}_{N}|\ast|2 \mathbf{a}^0_{N}\left [\rho \right ] \rho+J[u]|}_4\\
		&\leq N^{1/4-1 +s/2}\sqrt{\sum_{j=1}^N\norm{u_j -u_j^R}_2^2}\norm{|x|^{-s}\one_{B(0,R)}}_2\norm{ \mathbf{a}^0_{N}\left [\rho \right ] \rho+J[u]}_{4/3}\\
			&\leq N^{1/4-1 +s/2}R^{1-s}\sqrt{\sum_{j=1}^N\norm{u_j -u_j^R}_2^2}\left(\norm{ \mathbf{a}^0_{N}\left [\rho \right ] \rho}_{4/3}+\norm{J[u]}_{4/3}\right)\\
				&\leq N^{1/4-1 +s/2}R^{1-s}\sqrt{\sum_{j=1}^N\norm{u_j -u_j^R}_2^2}\left(\norm{ \mathbf{a}^0_{N}\left [\rho \right ] }_{\infty}\norm{\rho}_{4/3}+N^{3/4}\right)\leq N^{s/2}R^{1-s}\sqrt{\sum_{j=1}^N\norm{u_j -u_j^R}_2^2}
						\end{align*}
where we used an interpolation inequality $\norm{\rho}_{4/3}\leq \norm{\rho}_{1}^{1/2}\norm{\rho}_{2}^{1/2}\leq N^{3/4}$ and the bound \eqref{ine:LRs}. We deduce that
\begin{align}
	N^{-1}\left|\sum_{j=1}^N	A^{(1)}_j\right|	&\leq N^{s-1}R^{2-2s}+\frac{1}{N}\sum_{j=1}^N\norm{u_j -u_j^R}_2^2.\label{ine:A1}
\end{align}
The next term is
	\begin{align}
	\left|\sum_{j=1}^N	A^{(2)}_j\right|&\leq 2 	\left|  \bra{\mathbf{a}_{N}\ast  (\mathbf{a}_{N}\left [\rho_R\right ]\rho_R -\mathbf{a}^0_{N}\left [\rho\right ]\rho)}\ket{\sqrt{\rho_R}\sqrt{\sum_{j=1}^N|u_j-u^R_j|^2}}\right|\nonumber\\
	&\lesssim N^{1/2}\sqrt{\sum_{j=1}^N\norm{u_j -u_j^R}_2^2}\norm{\mathbf{a}_{N}\ast  (\mathbf{a}_{N}\left [\rho_R\right ]\rho_R -\mathbf{a}^0_{N}\left [\rho\right ]\rho)}_{\infty}\nonumber\\
	&\lesssim N^{1/2}\sqrt{\sum_{j=1}^N\norm{u_j -u_j^R}_2^2}\left(N^{s/2 -1/2}R^{1-s}+N^{-1/2}\sqrt{\sum_{j=1}^N\norm{u_j -u_j^R}_2^2}\right)	\nonumber
			\end{align}
by the use of Lemma \ref{lem:Rto0}, restricting to $s<1/2$. We conclude that
\begin{align}
	N^{-1}\left|\sum_{j=1}^N	A^{(2)}_j\right|
&\lesssim N^{s-1}R^{2(1-s)} +\frac{1}{N}\sum_{j=1}^N\norm{u_j -u_j^R}_2^2.\label{ine:A2}
\end{align}
We continue with
\begin{align}
	\left|\sum_{j=1}^N	A^{(3)}_j\right|&\leq 	\left|  \bra{\mathbf{a}_{N}\ast\left [J[u]-J[u^R]\right]}\ket{\sqrt{\rho_R}\sqrt{\sum_{j=1}^N|u_j-u^R_j|^2}}\right|\nonumber\\
		&=\left|  \bra{\left [J[u]-J[u^R]\right]}\ket{\mathbf{a}_{N}\ast\sqrt{\rho_R}\sqrt{\sum_{j=1}^N|u_j-u^R_j|^2}}\right|\nonumber\\
	&\lesssim 	\left|  \bra{\sum_{j=1}^N|u_j -u_j^R||\nabla u_j|}\ket{\mathbf{a}_{N}\ast\sqrt{\rho_R}\sqrt{\sum_{j=1}^N|u_j-u^R_j|^2}}\right|\nonumber
	\end{align}
by noting that $J[u]-J[u^R]=\sum_{j=1}^N(u_j -u_j^R)\nabla \overline{u}_j +u_j^R\nabla\overline{(u_j -u_j^R)}+\overline{(u_j^R -u_j)}\nabla u_j^R +\overline{u}_j^R\nabla(u_j^R -u_j) $ and that $\nabla \cdot \mathbf{a}_{N}=0$ to perform an integration by part when $\nabla$ hits $(u_j -u_j^R)$. We obtain
\begin{align}
	\left|\sum_{j=1}^N	A^{(3)}_j\right|
	&\lesssim 	\left|  \bra{\sqrt{\sum_{j=1}^N|u_j-u^R_j|^2}\sqrt{\sum_{j=1}^N|\nabla u_j|^2}}\ket{\mathbf{a}_{N}\ast\sqrt{\rho_R}\sqrt{\sum_{j=1}^N|u_j-u^R_j|^2}}\right|\nonumber\\
	&\lesssim \norm{\sqrt{\sum_{j=1}^N|\nabla u_j|^2}}_2\norm{\sqrt{\sum_{j=1}^N|u_j-u^R_j|^2}}_2\norm{\mathbf{a}_{N}\ast\sqrt{\rho_R}\sqrt{\sum_{j=1}^N|u_j-u^R_j|^2}}_{\infty}\nonumber\\
	&\leq a^{(2)} \sqrt{\sum_{j=1}^N\norm{\nabla u_j}_2^2}\sqrt{\sum_{j=1}^N\norm{u_j -u_j^R}_2^2}\lesssim N^{-1/2} \sqrt{N}\sum_{j=1}^N\norm{ u_j-u_j^R}_2^2 \label{ine:A3}
	\end{align}
	where $a^{(2)}$ has been estimated in \eqref{ine:a2}. The last term to treat is
	\begin{align}
	\left|\sum_{j=1}^N	A^{(4)}_j\right|&\leq gN 	\left| \braket{|\mathbf{a}^0_N|^2\ast\rho -|\mathbf{a}_{N}|^2\ast\rho_R}{\overline{ u_j}^R( u_j- u^R_j)}\right|\nonumber\\
	&\lesssim NN^{1/2}\sqrt{\sum_{j=1}^N\norm{u_j -u_j^R}_2^2}\norm{|\mathbf{a}^0_N|^2\ast\rho -|\mathbf{a}_{N}|^2\ast\rho_R}_{\infty}\nonumber\\
	&\lesssim N^{1/2}\sqrt{\sum_{j=1}^N\norm{u_j -u_j^R}_2^2}\left(N^{s -1/2}R^{1-2s}+N^{-1/2}\sqrt{\sum_{j=1}^N\norm{u_j -u_j^R}_2^2}\right)		\label{ine:A4}					
	\end{align}
	by \eqref{ine:anr2} of Lemma \ref{lem:Rto0}, with $0\leq s<1/4$.
	We can conclude the proof by summing together \eqref{ine:E}, \eqref{ine:A1},\eqref{ine:A2}, \eqref{ine:A4} and \eqref{ine:A4} into Equation \eqref{eq:1bd-diff} to get that
	\begin{align*}
	\left|\partial_t\frac{1}{N}\sum_{j=1}^N\norm{ u_j-u_j^R}_2^2\right|&= \frac{1}{N}\left|\sum_{j=1}^N E_j+A_j^{(1)}+A_j^{(2)}+A_j^{(3)}+A_j^{(4)}\right|\\&\lesssim N^{2s-1}R^{2-4s}+\frac{1}{N}\sum_{j=1}^N\norm{ u_j-u_j^R}_2^2.
	\end{align*}
This concludes the proof.
\end{proof}
We have now established all the knowledge we needed about the one-body effective equation. We turn to the $N$-body part of the proof. In the next section we introduce the number operators and their basic properties.
\section{Projectors and number operator}\label{sec:proj}
In this section we define the basic objects and properties of our method. Its contain is a direct generalization of \cite[Section 4 and 6]{Pet_Pic_16} to three-body operators and two-body operators with kinetic terms. \subsection{Projectors and Pauli principle}
For the orthonormal system $\left\{u_j(t)\right\}_{j=1}^N$ solution of $\mathrm{EQ}(N,s,R,m,\omega(0))$,  we define
\begin{equation}
p^{j}_m:=p^{u_j}_m:=\ket{u_j(x_m)}\bra{u_j(x_m)}:=\ket{u_j}\bra{u_j}_m\nonumber
\end{equation}
with the short-hand notation $p^{j}_m$ only valid for the basis  $\left\{u_j(t)\right\}_{j=1}^N$  solution of $\mathrm{EQ}$. We then construct the projector on $\mathrm{span}\left\{ u_1(t),\dots ,u_N(t)\right\}$ as
\begin{equation}
p_m:=\sum_{j=1}^Np_m^j\quad\text{as well as}\quad q_m:=\one_{L^2(\R^{2})}-p_m.
\end{equation}
These projectors have the properties $p_mq_m=0$, $[p_m,p_n]=[p_m,q_n]=[q_m,q_n]=0$ and, due to the antisymmetry of the wave-function $\Psi_N$, we have 
\begin{equation}
p_m^{\phi}p_n^{\phi}=0\label{eq_exc}
\end{equation}
for any $\phi\in L^2(\R^2)$, see \cite[Section 4]{Pet_Pic_16}.
We write
$\one_{L^2(\R^{2N})}=\otimes_{m=1}^N(p_m +q_m)$
and define $P_N^{(k)}$ by collecting all summands containing $k$ factors of $q$ operators. For any $0\leq k\leq N$, we then have that
\begin{equation}
P_N^{(k)}:=\sum_{\vec{a}\in\mathcal{A}_k}\prod_{m=1}^N(p_m)^{1-a_m}(q_m)^{a_m};\,\,\mathcal{A}_k:=\left\{ \vec{a}=(a_1,\dots , a_N)\in \left\{ 0,1\right\}^N : \sum_{m=1}^Na_m=k\right\}.\label{def:P}
\end{equation}
  We use the convention that $P_N^{(k)}=0$ if $k\notin \{0,\dots ,N\}$. That way $P_N^{(k)}$ can be seen as the symmetrized version of the projector $q_1q_2\dots q_k p_{k+1}\dots p_N$. We can check that it has the basic properties $\sum_{k=0}^NP_N^{(k)}=\one$ and $P_N^{(k)}P_N^{(\ell)}=\delta_{k,\ell}P_N^{(k)}$. A direct calculations using the Effective Equation \eqref{def:HH}, also provides the Heisenberg equations of motions of the projectors
  \begin{equation}\ii\partial_t p_m=[\mathcal{H}_R(x_m),p_m]\quad\text{and}\quad\ii\partial_t q_m=[\mathcal{H}_R(x_m),q_m].
  \end{equation} 
  Using that $[\mathcal{H}_R(x_m),q_n]=0$ when $m\neq n$, we then deduce that
\begin{equation}\label{def:eqmotionPN}
\ii\partial_tP_N^{(k)}(t)=\left[\sum_{m=1}^N\mathcal{H}_R(x_m),P_N^{(k)}(t)\right].
\end{equation}
An important property of the projectors is the exclusion principle described in the following lemma.
\begin{lemma}[\textbf{Exclusion Principle}]\mbox{}\\\label{lem:exclusion}
Take $N\geq 2$ and let $A\subseteq\{1,2,\dots,N\}$. Define $a=|A|$ and denote by
 $\psi_{\mathrm{as}}^{a}\in L^2(\R^{2N})$ a normalized function, antisymmetric in all its variables except in the variables $\cup_{j\in A}\{x_j\}$.  Then, for all $m\notin A$, we have
\begin{equation}
\braket{\psi_{\mathrm{as}}^{a}}{p_m^{\phi}\psi_{\mathrm{as}}^{a}}_{N-a}(x_a)\leq (N-a)^{-1}\braket{\psi^{a}_{\mathrm{as}}}{\psi^{a}_{\mathrm{as}}}_{N-a}(x_a)
\end{equation}
for almost all $x_a$, the line vector made of the variables $\cup_{j\in A}\{x_j\}$ and where $\braket{\cdot}{\cdot}_{N-a}$ is the partial scalar product in the variables $\cup_{j\notin A}\{x_j\}$. We have in particular that
\begin{equation}
\braket{\psi_{\mathrm{as}}^{a}}{p_m^{\phi}\psi_{\mathrm{as}}^{a}}\leq (N-a)^{-1}\quad\text{and}\quad p_m\leq N(N-a)^{-1}\one_{L^2(\R^2)}\lesssim \one_{L^2(\R^2)}\label{ine:pauli}.
\end{equation}
 for finite $a$.
\end{lemma}
\begin{proof}
The dimension does not play any role and our proof is a simple re-writing of \cite[Lemma 4.1]{Pet_Pic_16} where the non symmetric particles don't have to be the $a$ first. By symmetry we write
\begin{align}
\braket{\psi_{\mathrm{as}}^{a}}{p_m^{\phi}\psi_{\mathrm{as}}^{a}}_{N-a}=(N-a)^{-1}\braket{\psi_{\mathrm{as}}^{a}}{\sum_{n\notin A}p_n^{\phi}\psi_{\mathrm{as}}^{a}}_{N-a}
\end{align}
and a Cauchy-Schwarz provides
\begin{align}
\braket{\psi_{\mathrm{as}}^{a}}{p_m^{\phi}\psi_{\mathrm{as}}^{a}}^2_{N-a}&\leq(N-a)^{-2}\braket{\psi_{\mathrm{as}}^{a}}{\psi_{\mathrm{as}}^{a}}_{N-a}\braket{\psi_{\mathrm{as}}^{a}}{\sum_{n\notin A}p_n^{\phi}\sum_{j\notin A}p_j^{\phi}\psi_{\mathrm{as}}^{a}}_{N-a}\nonumber\\
&=(N-a)^{-1}\braket{\psi_{\mathrm{as}}^{a}}{\psi_{\mathrm{as}}^{a}}_{N-a}\braket{\psi_{\mathrm{as}}^{a}}{p_m^{\phi}\psi_{\mathrm{as}}^{a}}_{N-a}
\end{align}
we divide both sides by $\braket{\psi_{\mathrm{as}}^{a}}{p_m^{\phi}\psi_{\mathrm{as}}^{a}}_{N-a}
$ and get the result.
\end{proof}
In the following lemma, we show how to use the symmetry of the wave function to reconstruct $q_j$ projectors from sums of projections on specific orbitals $p_j^i$
\begin{lemma}
Take an orthonormal basis $\left\{v_j(t)\right\}_{j=1}^N$ and define $p_m^{v_j}=\ket{v_j}\bra{v_j}_m$ acting on the particle $x_m$. For any normalized function $\Psi^1_N\in L^2(\R^{2N})$ antisymmetric in all variables except $x_1$, we have that
\begin{equation}\label{ine:alphat}
\left|\braket{p_1}{\sum_{j=1}^N\left(\one -\sum_{m=2}^Np_m^{v_j}\right)p_1}_{\Psi^1_N}\right|\leq (N-1)\braket{\Psi^1_N}{q_2\Psi^1_N}+\braket{\Psi^1_N}{\Psi^1_N}
\end{equation}
\begin{equation}\label{ine:alphat2}
\left|\braket{p_1}{\sum_{i,j=1}^N\left(\one -\sum_{m=2}^Np_m^{v_j}\sum_{n=2}^Np_n^{v_i}\right)p_1}_{\Psi^1_N}\right|\lesssim N^2\braket{\Psi^1_N}{q_2\Psi^1_N}+N\braket{\Psi^1_N}{\Psi^1_N}
\end{equation}
\end{lemma}
\begin{proof}
We use that the projector $p_m=\sum_{j=1}^N\ket{u_j}\bra{u_j}_m=\sum_{j=1}^N\ket{v_j}\bra{v_j}_m$ does not depend on the basis  to get
\begin{align*}
&\left|\braket{p_1}{\sum_{j=1}^N\left(\one -\sum_{m=2}^Np_m^{v_j}\right)p_1}_{\Psi^1_N}\right|=\left|\braket{p_1}{\left(N-\sum_{j=1}^N\sum_{m=2}^Np_m^{v_j}\right)p_1}_{\Psi^1_N}\right|\\
&\quad\quad\quad\quad=\left|\braket{p_1}{\left(N-\sum_{m=2}^Np_m\right)p_1}_{\Psi^1_N}\right|\leq \left|\braket{p_1}{\sum_{m=2}^N\left(\one-p_m\right)p_1}_{\Psi^1_N}\right|+\braket{\Psi^1_N}{\Psi^1_N}\\
&\,\quad\quad\quad\quad\quad\quad\quad\quad\quad\quad\quad\quad\quad\quad\quad\quad\quad\quad\leq \left|\braket{p_1}{(N-1)q_2p_1}_{\Psi^1_N}\right|+\braket{\Psi^1_N}{\Psi^1_N}
\end{align*}
where we used the symmetry of $\Psi^1_N$  and $p_1\lesssim \one$. The proof of the second inequality proceeds the same way
\begin{align}
\left|\braket{p_1}{\sum_{i,j=1}^N\left(\one -\sum_{m=2}^Np_m^{v_j}\sum_{n=3}^Np_n^{v_i}\right)p_1}_{\Psi^1_N}\right|&=\left|\braket{p_1}{\left(N^2-\sum_{j=1}^N\sum_{m=2}^Np_m^{v_j}\sum_{i=1}^N\sum_{n=3}^Np_n^{v_i}\right)p_1}_{\Psi^1_N}\right|\nonumber\\\
&=\left|\braket{p_1}{\left(N^2-\sum_{m=2}^Np_m\sum_{n=3}^Np_n\right)p_1}_{\Psi_N}\right|\nonumber\\
&\lesssim N^2\left|\braket{p_1}{\left(\one-p_2p_3\right)p_1}_{\Psi^1_N}\right|+N\braket{\Psi^1_N}{\Psi^1_N}\nonumber\\
&\lesssim N^2\left|\braket{p_1}{q_2p_1}_{\Psi^1_N}\right|+N\braket{\Psi^1_N}{\Psi^1_N}
\end{align}
where we used that $\one -p_2p_3=p_2q_3 +q_2p_3 +q_2q_3$ and $p_1\lesssim \one$.
\end{proof}
\subsection{Weighted number operators}
We now introduce some more properties of the projector allowing to deal with increasing weights $f:\mathbb{Z}\to [0,1]$ with $f(0)=0$, $f(N)=1$ and $f(k)=0$ if $k\notin \{0,\dots ,N\}$. We use it to define 
\begin{equation}
\hat{f}:=\sum_{k\in \mathbb{Z}}^Nf(k)P^{(k)}_N
\end{equation}
and the weighted number counting functional
\begin{equation}
\alpha_f(\Psi_N(t)):=\braket{\Psi_N(t)}{\hat{f}\Psi_N(t)}_{L^2(\R^{2N})}
\end{equation}
for $\Psi_N\in L^2(\R^{2N})$ the solution to the Schr\"odinger equation \eqref{def:schro1}. Note that when the weight is $f=k/N$ the functional is $\alpha_f=:\alpha(t)=\braket{\psi}{q_1\psi}_{L^2(\R^{2N})}$ and directly counts the number of particles not in $\mathrm{span}(\{u_j\}_{j=1}^N)$. As explained in Section \ref{sec:strpr}, we will need to work with the square root version of this functional, i.e with the weight $m(k):=\sqrt{k/N}$. We define it as
\begin{equation}\label{def:alpham2}
\alpha_{m}(t):=\braket{\Psi_N(t)}{\widehat{m}(1/2)\Psi_N(t)}.
\end{equation}
with the generalization to all power $\xi>0$ through the notation $\widehat{m}(\xi):=\sum_{k\in \mathbb{Z}}^N(k/N)^{\xi}P^{(k)}_N$.
\begin{definition}\label{def:derivatives}
We define the shift operator $\tau(d)$ as $\tau(d)f(k)=f(k+d)$ and the discrete derivatives of order $d>0$ of $f$ by
\begin{align}
f^{(d)}(k)&:=(\one -\tau(-d))\one_{\{d, \dots ,N\}}(k)f(k)\\
f^{(-d)}(k)&:=(\tau(d)-\one)\one_{\{0, \dots ,N-d\}}(k)f(k)
\end{align}
where we denoted $\one_A(k)=1$ if $k\in A$ and zero otherwise. We also define the projector
\begin{equation}
P^{(a)}_{i_1 \dots i_n}:=\left(\prod_{m=1}^aq_{i_m}\prod_{m=a+1}^np_{i_m}\right)_{\mathrm{sym}}
\end{equation}
for all $a,n \in \mathbb{N}$ with $a\leq n<N$ and $\{i_1 ,\dots ,i_n\}\subset\{1 ,\dots ,N\}$. We then have $P^{(k)}_{1\dots N}=P^{(k)}_N$. Those projectors will be used via the relations $\sum_{a=0}^2P^{(a)}_{12}=\one_{L^2(\R^4)}$ and $\sum_{a=0}^3P^{(a)}_{123}=\one_{L^2(\R^6)}$ for $n=2$ and $n=3$ in the later calculations.
\end{definition}

\begin{lemma}[\textbf{Commutation of the weight}]\mbox{}\\\label{lem:comw}
Let $h_{1\dots n}$ be two operators respectively acting on the $n$-th first variables, $0<n\leq N$, we then have
\begin{align}
\left(P^{(a)}_{1\dots n}h_{1\dots n}P^{(b)}_{1\dots n}\right)\hat{f}&=\widehat{\tau_{b-a}f}\left(P^{(a)}_{1\dots n}h_{1\dots n}P^{(b)}_{1\dots n}\right)\\
\hat{f}\left(P^{(a)}_{1\dots n}h_{1\dots n}P^{(b)}_{1\dots n}\right)&=\left(P^{(a)}_{1\dots n}h_{1\dots n}P^{(b)}_{1\dots n}\right)\widehat{\tau_{a-b}f}\\
\end{align}
\end{lemma}
\begin{proof}
The above lemma is a direct generalization of \cite[Lemma 6.4]{Pet_Pic_16} to $n$ particles operators. The proof is differed to Appendix \ref{pro:comm}.
\end{proof}
Using Definition \ref{def:derivatives}, for $a,b=0,1,\dots , n$ with $a>b$, we have the useful relation
\begin{align}
(\widehat{f}-\widehat{\tau_{b-a}f})P^{(a)}_{1\dots n}&=\sum_{k\in\mathbb{Z}}\left(f(k)-f(k-(a-b))\right)\one_{\{a-b,\dots ,N\}}(k)P^{(k)}_N =\widehat{f^{(a-b)}}\geq 0,
\label{eq:dev}
\end{align}
coming from the fact that $P^{(a)}_{1,\dots ,n}P_N^{(k)}=0, \,\forall k<a$ and from the monotonicity of $f$.
Note that we also have the identity
\begin{align}
\widehat{\left(\tau_{a-b}\sqrt{f^{(a-b)}}\right)}&=\sum_{k\in\mathbb{Z}}\left(f(k+a-b)-f(k)\right)^{1/2}\one_{\{a-b,\dots ,N\}}(k+a-b)P^{(k)}_N\nonumber\\
&=\left(\sum_{k\in\mathbb{Z}}\left(f(k+a-b)-f(k)\right)\one_{\{0,\dots ,N-(a-b)\}}(k)P^{(k)}_N\right)^{1/2}=\widehat{f^{(b-a)}}^{1/2}.\label{eq:transdev}
\end{align}
using that $\widehat{f}^{1/2}=\widehat{f^{1/2}}$.
We will use the above relations in the following proposition highlighting the link between the commutations formulas $[W_{1,\dots,n},\hat{f}]$ and the discrete derivatives of $\widehat{f}$.
\begin{proposition}
Let $W_{12}$ and $W_{123}$ respectively be a two and a three-body operators. Let $\Psi_N\in L_{asym}^{2}(\R^{2N})$.
By the commutation rules of Lemma \ref{lem:comw}, we obtain the two identities
\begin{align}
\braket{\Psi_N}{\left[W_{12},\hat{f}\right]\Psi_N}
&=\sum_{a,b=0}^2\braket{\Psi_N}{P_{12}^{(a)}\left(W_{12}\hat{f}-\hat{f}W_{12}\right)P_{12}^{(b)}\Psi_N}\nonumber\\
&=-\sum_{a,b=0}^2\braket{\Psi_N}{\left(\hat{f}-\widehat{\tau_{b-a}f}\right)P_{12}^{(a)}W_{12}P_{12}^{(b)}\Psi_N}\nonumber\\
&=2\ii \Im \sum_{a>b}\braket{\Psi_N}{\widehat{f^{(a-b)}}P_{12}^{(a)}W_{12}P_{12}^{(b)}\Psi_N}\nonumber\\
&=2\ii \Im \sum_{a>b}\braket{\Psi_N}{\widehat{f^{(a-b)}}^{1/2}P_{12}^{(a)}W_{12}P_{12}^{(b)}\widehat{\left(\tau_{a-b}\sqrt{f^{(a-b)}}\right)}\Psi_N}\nonumber\\
&=2\ii \Im \sum_{a>b}\braket{\Psi_N}{\widehat{f^{(a-b)}}^{1/2}P_{12}^{(a)}W_{12}P_{12}^{(b)}\widehat{f^{(b-a)}}^{1/2}\Psi_N},\label{eq:W12}
\end{align}
for the two-body operator by using \eqref{eq:dev} and \eqref{eq:transdev}, as well as
\begin{align}
\braket{\Psi_N}{\left[W_{123},\hat{f}\right]\Psi_N}&=\braket{\Psi_N}{\sum_{a=0}^3P_{123}^{(a)}\left(W_{123} \hat{f}-\hat{f}W_{123}\right)\sum_{b=0}^3P_{123}^{(b)}
\Psi_N}\nonumber\\
&=-\sum_{a,b=0}^3\braket{\Psi_N}{\left(\hat{f}-\widehat{\tau_{b-a}f}\right)P_{123}^{(a)}W_{123}P_{123}^{(b)}
\Psi_N}\nonumber\\
&=2\ii\Im\sum_{a>b}\braket{\Psi_N}{\widehat{f^{(a-b)}}P_{123}^{(a)}W_{123}P_{123}^{(b)}
\Psi_N}\nonumber\\
&=2\ii\Im\sum_{a>b}\braket{\Psi_N}{\widehat{f^{(a-b)}}^{1/2}P_{123}^{(a)}W_{123}P_{123}^{(b)}
\widehat{f^{(b-a)}}^{1/2}\Psi_N}.\label{eq:W123}
\end{align}
for the three-body operator, again by the use of \eqref{eq:dev} and \eqref{eq:transdev}.
\end{proposition}
We now quantify the average of the discrete derivatives of $\widehat{m}:=\widehat{m}(1/2)$ in terms of its average $\alpha_{m}(t)=\braket{\Psi_N(t)}{\widehat{m}(1/2)\Psi_N(t)}$.
\begin{lemma}
Let $m(k)=(k/N)^{1/2}$, we have for any $\Psi_N\in L_{\mathrm{asym}}^2(\R^{2N})$, $1\leq n\leq N$ and the discrete derivative of order $d>0$, that
\begin{align}
\braket{  \Psi_N}{\widehat{m^{(\pm d)}}  \Psi_N}&\lesssim \frac{d}{\sqrt{N}},\label{ine:m}
\end{align}
\begin{align}
\braket{ q_1\dots q_n \Psi_N}{\widehat{m^{(\pm d)}} q_1\dots q_n \Psi_N}&\lesssim \frac{d}{N}\alpha_m(t) .\label{ine:qm}
\end{align}
Finally, we have
\begin{align}
\braket{\nabla_1 q_1 \widehat{m^{(\pm d)}}^{1/2}q_2\Psi_N}{\nabla_1 q_1 \widehat{m^{(\pm d)}}^{1/2}q_2\Psi_N}&\leq \frac{d}{N}\norm{\nabla_1 q_1\Psi_N}^2\label{ine:nabm}
\end{align}
\end{lemma}
\begin{proof}
In this proof we will only treat the case $+d$, the proof for $-d$ is similar.
The first inequality comes from a direct computation of the derivative
\begin{align*}
(m(k)-m(k-d))\one_{\{d,\dots ,N\}}(k)&=\frac{1}{\sqrt{N}}\left(\sqrt{k}-\sqrt{k-d}\right)\one_{\{d,\dots ,N\}}(k)\\
&=\frac{1}{\sqrt{N}}\frac{d}{\sqrt{k}+\sqrt{k-d}}\one_{\{d,\dots ,N\}}(k)\leq \frac{d}{\sqrt{N}}.
\end{align*}
To show the next one, we use that on a antisymmetric $\phi \in L^2(\R^{2N})$. We have the properties
\begin{align}
\braket{\phi}{q_1\phi}=N^{-1}\sum_{j=1}^N\braket{\phi}{q_j\phi}=N^{-1}\sum_{j=1}^N\braket{\phi}{q_j\sum_{k=0}^NP_N^{(k)}\phi}=\sum_{k=1}^N\frac{k}{N}\braket{\phi}{P_N^{(k)}\phi}\\
\braket{\phi}{q_1q_2\phi}=N^{-1}(N-1)^{-1}\sum_{\substack{i,j=1\\ i\neq j}}^N\braket{\phi}{q_jq_i\phi}=\sum_{k=1}^N\frac{k(k-1)}{N(N-1)}\braket{\phi}{P_N^{(k)}\phi}
\end{align}
and so one, so forth for $q_1q_2\dots q_n$. By the Definition \ref{def:derivatives} of the discrete derivative, the definition of $\widehat{m}=\sum_{k\in\mathbb{Z}}m(k)P_N^{(k)}$ and by using the property $P_N^{(k)}P_N^{(m)}=P_N^{(k)}\delta_{km}$, we can write

\begin{align}
&\braket{ q_1\dots q_n \Psi_N}{\widehat{m^{( d)}} q_1\dots q_n \Psi_N}=\nonumber\\
&\quad\quad \sum_{k\in\mathbb{Z}}\frac{k(k-1)\dots(k-n+1)}{N(N-1)\dots(N-n+1)}(m(k)-m(k-d))\one_{\{d,\dots ,N\}}(k)\braket{ \Psi_N}{P_N^{(k)}  \Psi_N}\nonumber\\
&\quad\quad\quad\leq \frac{d}{N}\sum_{k\in\mathbb{Z}}\frac{k}{\sqrt{N}(\sqrt{k}+\sqrt{k-d})}\one_{\{d,\dots ,N\}}(k)\braket{ \Psi_N}{P_N^{(k)}  \Psi_N}\nonumber\\
&\quad\quad\quad\leq \frac{d}{N}\braket{ \Psi_N}{\sum_{k\in\mathbb{Z}}(k/N)^{1/2}P_N^{(k)}  \Psi_N}=\frac{d}{N}\braket{ \Psi_N}{\widehat{m}  \Psi_N}.
\end{align}
 The last inequality of the lemma uses
 that $(p_1+q_1)=\one_{L^2(\R^2)}$ to get
\begin{align}
\nabla_1 q_1q_2 \widehat{m^{(d)}}^{1/2}&=\nabla_1 q_1q_2\sum_{k\in\mathbb{Z}}(m^{(d)}(k))^{1/2}P_N^{(k)}\nonumber\\
&=\sum_{k\in\mathbb{Z}}(m^{(d)}(k))^{1/2}P_{2\dots N}^{(k-1)}q_2\nabla_1 q_1(p_1 +q_1)\nonumber\\
&=\sum_{k\in\mathbb{Z}}(m^{(d)}(k))^{1/2}(p_1+q_1)P_{2\dots N}^{(k-1)}q_2\nabla_1 q_1\nonumber\\
&= \widehat{m^{(d)}}^{1/2}\nabla_1 q_1q_2
\end{align}
using that $q_1p_1 =0$, $q_1^2=q_1$ and that $P_N^{(k)}$ commutes with $q_2$. There is left to bound
\begin{align}
\braket{\nabla_1 q_1 q_2\Psi_N}{\widehat{m^{(d)}}\nabla_1 q_1q_2 \Psi_N}&=\sum_{k\in\mathbb{Z}}\frac{k}{N}(m(k)-m(k-d))\one_{\{d,\dots ,N\}}(k)\braket{\nabla_1 q_1 \Psi_N}{P_N^{(k)}\nabla_1 q_1 \Psi_N}\nonumber\\
&\leq \frac{d}{N}\sum_{k\in\mathbb{Z}}\frac{k}{\sqrt{N}(\sqrt{k}+\sqrt{k-d})}\one_{\{d,\dots ,N\}}(k)\braket{\nabla_1 q_1 \Psi_N}{P_N^{(k)} \nabla_1 q_1 \Psi_N}\nonumber\\
&\leq \frac{d}{N}\braket{\nabla_1 q_1 \Psi_N}{\sum_{k\in\mathbb{Z}}\frac{\sqrt{k}}{\sqrt{N}}P_N^{(k)} \nabla_1 q_1 \Psi_N}= \frac{d}{N}\norm{\nabla_1 q_1\Psi_N}^2
\end{align}
concluding the proof of the lemma.
\end{proof}
\section{Diagonalization procedure}\label{sec:diad}
In this section we develop the lemmas needed to deal, among others, with the diagonalization of the following operators. The second part of the section is dedicated to the control of the traces we obtained by the diagonalization procedure.
\begin{definition}\label{def:vw}
	Given  an orthonormal set $\{u_i\}_{i=1}^N$,  we define the operators:
	\begin{align}
		v'(x)&:=2 \beta(-\im \nabla_x)\cdot \mathbf{a}_N\left [\rho\right ](x
		) - \beta( \mathbf{a}_N\ast J)(x),\label{def:v'}\\
		w'(x)&:= \beta\mathbf{a}_N\left [\rho\right ]^2(x)-2 \beta\mathbf{a}_N\ast \beta\mathbf{a}_N\left [\rho\right ]\rho)(x),\label{def:w'}\\
		v(x-y)&:= 2 \beta(-\im \nabla_x)\cdot \mathbf{a}_N(x-y), \label{def:v}\\
		w(x-y, x-z)&:= \beta^2\mathbf{a}_N(x -y)\cdot\mathbf{a}_N(x-z)\label{def:w}
	\end{align}
	where $\rho:=\sum_{j=1}^N|u_j|^2$ and $J:=\sum_{j=1}^N\im(u_j\nabla \overline{u}_j-\overline{u}_j\nabla u_j)
$.
\end{definition}
From now one we will often use the short-hand notation $v_{ij}:=v(x_i -x_j)$ together with $w_{ijk}:=w(x_i-x_j, x_i-x_k)$ to simplify the exposition.
\subsection{Diagonalisation of the two-body operator}\label{sec:diag}
We will begin with the two-body potential. Recall that $v_{12}\neq v_{21}$ due to the presence of $\nabla_1$ in $v_{12}$. This asymmetry is responsible for having two different convolution terms in $v'$.
\begin{lemma}\label{lem:diag2b}
Let  $\{u_i\}_{i=1}^N$ be an orthonormal set and $p_j=\sum_{i=1}^N\ket{u_i}\bra{u_i}_j$, then there exist three sets of eigenvalues, eigenfunctions and unitaries $\{\lambda_i^{(\mathbf{a}\cdot \nabla)}(x_2), \chi_i(x_2),U^{\chi}(x_2)\}_{i=1}^N$, $\{\lambda_i^{(\mathbf{a^2})}(x_2), \eta_i(x_2),U^{\eta}(x_2)\}_{i=1}^N$ and $\{\lambda_i^{(\mathbf{a})}(x_1), \sigma_i(x_1), U^{\sigma}(x_1)\}_{i=1}^N$ such that
\begin{align}
p_1(-\im \nabla_1)\cdot \mathbf{a}_N(x_1-x_2)p_1&=\sum_{i=1}^N\lambda^{(\mathbf{a}\cdot \nabla)}_i(x_2)p_1^{\chi_i(x_2)},\label{def:diaglambda}\\
p_2(-\im \nabla_1)\cdot \mathbf{a}_N(x_1-x_2)p_2&=(-\ii\nabla_1)\cdot\sum_{i=1}^N\lambda^{(\mathbf{a})}_i(x_1)p_2^{\sigma_i(x_1)},\label{def:diaglambdabis}\\
p_2| \mathbf{a}_N(x_1-x_2)|^2p_2&=\sum_{i=1}^N\lambda^{(\mathbf{a^2})}_i(x_1)p_2^{\eta_i(x_1)},\label{def:diaglambdabisbis}
\end{align}
 with the change of basis rules $\chi =U^{\chi} u$,  $\sigma =U^{\sigma} u$ and $\eta =U^{\eta} u$ and the traces
\begin{align}
2\sum_{i=1}^N\lambda^{(\mathbf{a}\cdot \nabla)}_i(x_2)&=-(\mathbf{a}_N\ast J)(x_2),\,\,\sum_{i=1}^N\lambda^{(\mathbf{a})}_i(x_1)=  \mathbf{a}_N\left [\rho\right ](x_1),\,\,
\sum_{i=1}^N\lambda^{(\mathbf{a^2})}_i(x_1)=  \mathbf{a}^2_N\left [\rho\right ](x_1)\label{eq:tracea}.
\end{align}
\end{lemma}
Here and in the following any $x_1$-dependent basis element $\omega_j(x_1)$ will be understood as an $x_1$-dependent projector $\bra{\omega_j(x_1)}_k$, where $k\neq 1$ is the variable the projector acts on. To prevent any ambiguity we will systematically recall the integration index at the foot of the scalar-product.
\begin{proof}

We observe that $p_1v_{12}p_1$ is a multiplication operator in the variable $x_2$. We use this to overcome the algebraic issue of the double summation by seeing it at an $x_2$-dependent self-adjoint $(N\times N)$-matrix acting on $x_1$, in other words
\begin{align}
p_1(-\im \nabla_1)\cdot \mathbf{a}_N(x_1-x_2)p_1&=\sum_{i,j=1}^N\braket{u_i}{(-\im \nabla_1)\cdot \mathbf{a}_N(x_1-x_2)u_j}_1\ket{u_i}\bra{u_j}_1\\
&=\sum_{i}^N\lambda^{(\mathbf{a}\cdot \nabla)}_i(x_2)\ket{\chi_i^{x_2}}\bra{\chi_i^{x_2}}_1\\
&=\sum_{i}^N\lambda^{(\mathbf{a}\cdot \nabla)}_i(x_2)p_1^{\chi_i^{x_2}}\label{def:diaglambda2}
\end{align}
where there exists an $(N\times N)$ unitary $U^{\chi}$ such that $\ket{\chi_i^{x_2}}_1=\sum_{k=1}^N U^{\chi}_{ik}(x_2)\ket{u_k}_1$ and where the eigenvalues $\lambda^{(\mathbf{a}\cdot \nabla)}_i(x_2)$ are real numbers. Note that this is no more true for $p_2 v_{12}p_2$ as $\nabla_1$ is not a multiplication operator. In this case we get

\begin{align}
p_2(-\im \nabla_1)\cdot \mathbf{a}_N(x_1-x_2)p_2&=\sum_{i,j=1}^N\braket{u_i}{(-\im \nabla_1)\cdot \mathbf{a}_N(x_1-x_2)u_j}_2\ket{u_i}\bra{u_j}_2\nonumber\\
&=(-\ii\nabla_1)\cdot\sum_{i,j=1}^N\braket{u_i}{\mathbf{a}_N(x_1-x_2)u_j}_2\ket{u_i}\bra{u_j}_2\nonumber \\
&=(-\ii\nabla_1)\cdot\sum_{i=1}^N\lambda_i^{(\mathbf{a})}(x_1)\ket{\sigma_i^{x_1}}\bra{\sigma_i^{x_1}}_2 \label{def:diagmu}
\end{align}
where there exists an $(N\times N)$ unitary $U^{\sigma}$ such that $\ket{\sigma_i^{x_1}}_2=\sum_{k=1}^N U^{\sigma}_{ik}(x_1)\ket{u_k}_2$ and where the eigenvalues $\lambda_i^{(\mathbf{a})}(x_1)$ are real numbers. The third diagonalisation follows the same way. In order to compute the traces, note that
\begin{equation}
\sum_{i=1}^N\ket{u_i}\bra{u_i}_1=\sum_{i=1}^N\ket{\chi_i}\bra{\chi_i}_1=\sum_{i=1}^N\ket{\sigma_i}\bra{\sigma_i}_1
\end{equation}
because  $\mathrm{span}(\sigma^{x_2}_1 ,\dots ,\sigma^{x_2}_N)=\mathrm{span}(\chi_1^{x_1},\dots ,\chi_N^{x_1})=\mathrm{span}(u_1 ,\dots ,u_N)$.
The traces are then computed as follows, the first one reads
\begin{align}
2\sum_{i=1}^N\lambda^{(\mathbf{a}\cdot \nabla)}_i(x_2)&=2\sum_{i=1}^N\braket{\chi_i}{(-\im \nabla_1)\cdot \mathbf{a}_N(x_1-x_2)\chi_i}_1=2\sum_{i=1}^N\braket{u_i}{(-\im \nabla_1)\cdot \mathbf{a}_N(x_1-x_2)u_i}_1\nonumber\\
&=\sum_{i=1}^N\int_{\R^2}\left(\overline{u}_j(x_1)(-\ii\nabla_1)u_j(x_1)\cdot \mathbf{a}_N(x_1 -x_2)+\hc\right)\mathrm{d}x_1=-(\mathbf{a}_N\ast J)(x_2)
\end{align}
and the second one
\begin{align}
\sum_{i=1}^N\lambda^{(\mathbf{a})}_i(x_1)&=2\sum_{i=1}^N\braket{\sigma_i}{\mathbf{a}_N(x_1-x_2)\sigma_i}_2=\sum_{i=1}^N\braket{u_i}{\mathbf{a}_N(x_1-x_2)u_i}_2=(\mathbf{a}_N\ast \rho)(x_1).
\end{align}
The third follows the exact same way which concludes the proof of the lemma.
\end{proof}
We now bound the traces obtaines above in $L^1$ and $L^2$ norms. These quantities appears in the core of the proof in Section \ref{sec:Gron}.
\begin{lemma}\label{lem:pq}
 Assume that $\norm{\rho}_1\simeq\norm{\rho_{\nabla}}_1\lesssim N$, we have, as operator inequality on $ L_{\mathrm{as}}^2(\R^{2N})$, that 
  \begin{align}
p_1(-\Delta_1)p_1&\lesssim 1\label{ine:Delta},
\end{align}
together with
\begin{align}
q_2p_1v^2_{12}p_1q_2&\lesssim N^{-1}q_2
 \norm{|\mathbf{a}_N|^2\ast\rho_{\nabla}}_{\infty}\\
&\lesssim N^{-1}q_2\left(\norm{\nabla_{1}|\mathbf{a}_R|^2\sqrt{\rho}\sqrt{\rho_{\nabla}}}_{\infty}+ \norm{|\mathbf{a}_R|^2\sqrt{\rho}\sqrt{\rho_{\Delta}}}_{\infty}\right)\label{ine:v23c}
\end{align}
and with
\begin{align}
q_2p_1v^2_{21}p_1q_2&\lesssim N^{-1}\norm{|\mathbf{a}_N|^2\ast\rho}_{\infty}q_2(-\Delta_2)q_2\label{ine:v2pq}\\
p_1p_2v^2_{12}p_1p_2&\lesssim N^{-1} \norm{|\mathbf{a}_N|^2\ast\rho}_{\infty}\label{ine:ppv2pp}\\
p_1q_2(v'_{1})^2 p_1q_2&\lesssim  \norm{\mathbf{a}_N[\rho]}^2_{\infty}q_2+\norm{|\mathbf{a}_N|\ast J}_{\infty}^2q_2 \label{ine:v'p}\\
p_2|\mathbf{a}_N|_{12}^2p_2&\lesssim N^{-1}\norm{|\mathbf{a}_N|^2\ast\rho}_{\infty}\label{ine:pa2p}\\
p_2|\mathbf{a}_N|_{12}^4p_2&\lesssim N^{-1}\norm{|\mathbf{a}_N|^4\ast\rho}_{\infty}.\label{ine:pa4}
\end{align}
 \end{lemma}

\begin{proof}
We start by proving \eqref{ine:Delta}.
The diagonalisation procedure for the Laplacian provides a family $\{\lambda_j^{(\Delta)},\chi^{\Delta}_j, U^{(\Delta)}\}_{j=1}^N$ such that
\begin{align}
p_1(-\Delta_1) p_1=\sum_{j=1}^N\lambda_j^{(\Delta)}p_1^{\chi^{\Delta}_j}\leq N^{-1}\sum_{j=1}^N\norm{\nabla u_j}_2^2\lesssim 1
\end{align}
where we used the Exclusion Principle of Lemma \ref{lem:exclusion}, $\lambda_j^{(\Delta)}\geq 0$ and the fact that $\norm{\rho_{\nabla}}_1\lesssim N$.\newline
For the next estimate, we apply the diagonalisation procedure of Lemma \ref{lem:diag2b} on the operator $v_{12}^2$. Its provides a family $\{\lambda_i^{(v^2)}(x_2),\chi^{(v^2)}_i(x_2),U^{(v^2)}\}_{i=1}^N$ such that
\begin{align}
\braket{q_2}{p_1v^2_{12}p_1q_2}_{\phi}&=\braket{q_2}{\sum_{i=1}^N\lambda_i^{(v^2)}(x_2)p_1^{\chi^{(v^2)}_i(x_2)}q_2}_{\phi}\nonumber\\
&\leq \sup_{x_2}\sum_{i=1}^N\lambda_i^{(v^2)}(x_2)\braket{q_2}{p_1^{\chi^{(v^2)}_i(x_2)}q_2}_{\phi}\nonumber\\
&\leq \sup_{x_2} \sum_{i=1}^N\braket{u_i}{v^2_{12}u_i}_1(N-1)^{-1}\norm{q_2\phi}^2\label{eq:supv2}
\end{align}
by the Exclusion Principle of Lemma \ref{lem:exclusion} applied to a function antisymmetric in all variables except $x_2$, combined with the fact that $\lambda_i^{v^2}(x_2)\geq 0$. At this stage, we got the first inequality of \eqref{ine:v23c}. For the other one, we compute
\begin{align}
\sum_{j=1}^N\braket{u_j}{(v_{12})^2u_j}_1&=4\sum_{j=1}^N\braket{-\im \nabla_{1}u_j}{|a_R|^2(x_1 -x_2)(-\im \nabla_{1})u_j}_1\nonumber\\
& =4 \sum_{j=1}^N\braket{u_j}{\left(-\im \nabla_{1}|a_R|^2_{12}\right)\cdot(-\im \nabla_{1})u_j)}_1 +4\sum_{j=1}^N\braket{u_j}{|a_R|^2_{12}(-\Delta_{1}u_j)}_1\nonumber\\
&\lesssim \norm{\nabla_{1}|\mathbf{a}_R|^2\sqrt{\rho}\sqrt{\rho_{\nabla}}}_{\infty}+ \norm{|\mathbf{a}_R|^2\sqrt{\rho}\sqrt{\rho_{\Delta}}}_{\infty}
\end{align}
by Cauchy-Schwarz's inequality.
 Let's turn to \eqref{ine:v2pq}. For any $\phi\in L_{\mathrm{as}}^2(\R^{2N}) $, we have
\begin{align}
\braket{p_1q_2}{v^2_{21}p_1q_2}_{\phi}&=4\braket{p_1q_2}{(-\ii\nabla_2)|\mathbf{a}_N|_{21}^2(-\ii\nabla_2)p_1q_2}_{\phi}\nonumber\\
&\leq C \braket{(-\ii\nabla_2)q_2}{p_1|\mathbf{a}_N|_{21}^2p_1(-\ii\nabla_2)q_2}_{\phi}\nonumber\\
&\leq C\sum_{i=1}^N\braket{(-\ii\nabla_2)q_2}{\lambda_i^{(\mathbf{a^2})}(x_2)p_1^{\eta_i(x_2)}(-\ii\nabla_2)q_2}_{\phi}\nonumber
\end{align}
where we applied Lemma \ref{lem:diag2b} to diagonalise $p_1|\mathbf{a}_N|_{21}^2p_1$. We conclude the estimate with
\begin{align}
(N-1)^{-1}\sup_{x_2}\sum_{i=1}^N\lambda_i^{(\mathbf{a}^2)}(x_2)\norm{\nabla_2q_2\phi}^2\lesssim (N-1)^{-1}\norm{|\mathbf{a}_N|^2\ast\rho}_{\infty} \norm{\nabla_2q_2\phi}^2\nonumber
\end{align}
using the Exclusion Principle \ref{lem:exclusion} and the fact that $ \lambda_i^{(\mathbf{a}^2)}(x_2)\geq 0$. For the estimate \eqref{ine:ppv2pp}, we need to apply two diagonalizations. A first one for $p_1$,
\begin{align}
p_1p_2v^2_{12}p_1p_2&=p_2\sum_{i=1}^N\lambda_i^{v^2}(x_2)p_1^{\chi_i(x_2)}p_2=p_2\sum_{i=1}^N\lambda_i^{v^2}(x_2)p_1^{\chi_i(x_2)}p_2\lesssim N^{-1}p_2\sum_{i=1}^N\lambda_i^{v^2}(x_2)p_2
\end{align}
where we used Lemma \ref{lem:exclusion}. Now the operator $p_2\sum_{i=1}^N\lambda_i^{v^2}(x_2)p_2$ can again be diagonalized to get
\begin{align}
p_1p_2v^2_{12}p_1p_2&\leq (N-1)^{-1}\sum_{i=1}^N\mu^{v^2}_ip_2^{\tilde{\varphi}_i}\nonumber\\
&\leq N^{-1}(N-1)^{-1}\sum_{i=1}^N\mu^{v^2}_i\nonumber\\
&\lesssim N^{-2}\sum_{i,j=1}^N\braket{u_j(x)u_i(y)}{v^2(x-y)u_j(x)u_i(y)}\nonumber\\
&\lesssim N^{-2}\norm{|\mathbf{a}_N|^2\ast\rho}_{\infty}\sum_{i=1}^N\norm{\nabla u_i}_2^2.
\end{align}
We treat the estimate \eqref{ine:v'p}. Recall that 
$v'(x_1)=-\im \nabla_{x_1}\cdot \mathbf{a}_N\left [\rho\right ](x_1
		)+\hc -(\mathbf{a}_N\ast J)(x_1)$. We use the diagonalization procedure in the particle $x_1$ for $(v'_1)^2$ to get eigenvalues $\lambda_j^{(v'^2)}$ and a basis $\omega_j^{(v'^2)}$ such that
\begin{align*}
q_2p_1(v'_{1})^2p_1q_2&=q_2\sum_{j=1}^N\lambda_j^{(v'^2)}p_1^{\omega_j^{(v'^2)}}\leq (N-1)^{-1}q_2\sum_{j=1}^N\lambda_j^{(v'^2)}\lesssim N^{-1} \sum_{j=1}^N\braket{u_j}{(v'_1)^2u_j}_1q_2
\end{align*}
by the Exclusion Principle \ref{lem:exclusion} and the fact that $\lambda_j^{(v'^2)}\geq 0$. We now have to bound
\begin{align}
 \sum_{j=1}^N\braket{u_j}{(v'_1)^2u_j}_1&\lesssim \sum_{j=1}^N\left(2\norm{\mathbf{a}_N[\rho]\cdot\nabla u_j}_2^2 +\norm{(\mathbf{a}_N\ast J) u_j}_2^2\right)\nonumber\\
&\lesssim \norm{\mathbf{a}_N[\rho]}^2_{\infty}\sum_{j=1}^N\norm{\nabla u_j}_2^2+\norm{|\mathbf{a}_N\ast J|^2\rho}_1.
\end{align}
We obtain the result using that $\norm{\rho}_1\simeq \norm{\rho_{\nabla}}_1\lesssim N$. For \eqref{ine:pa2p}, we use Lemma \ref{lem:diag2b} to get 
\begin{align}
p_2|\mathbf{a}_N|_{12}^2p_2=\sum_{j=1}^N\lambda_j^{(\mathbf{a^2})}(x_1)p_2^{\eta_j(x_1)}\leq N^{-1}\norm{|\mathbf{a}_N|^2\ast\rho}_{\infty}
\end{align}
using the exclusion principle of Lemma \ref{lem:exclusion}, and the fact that $\lambda_j^{(\mathbf{a^2})}(x_1)\geq 0$. The last inequality is obtained in a similar way.
\end{proof}

\subsection{Diagonalisation of the three-body operator}
The main goal of this paragraph is to extend the diagonalisation procedure of Lemma \ref{lem:diag2b} to a three-body operator. The main results are contained is the following lemmas.

\begin{lemma}[\textbf{Diagonalisation of the three-body operator, the general case}]\mbox{}\\\label{lem:diag3b}
Let $w(x_1,x_2,x_3)$ be a positive, self-adjoint, three-body, multiplicative operator. 
There exist two basis denoted $\{\omega_j(x_1 ,x_2)\}_{j=1}^N$ and $\{\tilde{\omega}_j(x_1)\}_{j=1}^N$ as well as a sequence of functions $\lambda_j^{w}(x_1 ,x_2) $ depending on the variables $x_1$ and $x_2$ such that $w(x_1,x_2,x_3)$ has the following diagonalized form
\begin{equation}\label{eq:diagw123}
p_2p_3w(x_1, x_2, x_3)p_3p_2=\sum_{i=1}^N\sum_{j=1}^N\braket{\tilde{\omega}_i(x_1)}{\lambda^{w}_j(x_1 ,x_2)p^{\omega_j(x_1 ,x_2)}_3\tilde{\omega}_i(x_1)}_2p_2^{\tilde{\omega}_i(x_1)}
\end{equation}
with the trace
\begin{equation}\label{eq:eignvw123}
\sum_{i,j=1}^N\braket{\tilde{\omega}_i(x_1)}{\lambda^{w}_j(x_1 ,x_2)\tilde{\omega}_i(x_1)}_2=\int_{\R^4}w(x_1,x_2,x_3)\rho(x_2)\rho(x_3)\mathrm{d}x_2\mathrm{d}x_3.
\end{equation}
\end{lemma}
\begin{proof}
We start by writing the operator in the bracket notation
\begin{align}
p_2p_3w_{123}p_3p_2&=\sum_{i,j,k,\ell=1}^N\braket{u_i(x_2)u_j(x_3)}{w_{123}u_k(x_2)u_{\ell}(x_3)}_{2,3}\ket{u_i(x_2)u_j(x_3)}\bra{u_k(x_2)u_{\ell}(x_3)}\nonumber\\
&=\sum_{i,k=1}^N\braket{u_j}{\sum_{j,\ell=1}^N\braket{u_i}{w_{123}u_{\ell}}_3u_k}_2\ket{u_i}\bra{u_k}_2\otimes\ket{u_j}\bra{u_{\ell}}_3.\label{eq:diagw}
\end{align}
Here we can diagonalize the operator acting on $x_3$ inside the bracket in the variable $x_2$. There exists a basis of $\mathrm{span}(u_1,\dots ,u_N)$ denoted $\{\omega_j(x_1 ,x_2)\}_{j=1}^N$ together with the eigenvalues $\{\lambda^w_j(x_1 , x_2)\}_{j=1}^N$ such that
\begin{align}\label{eq:firstdiag}
\sum_{j,\ell=1}^N\braket{u_j}{w_{123}u_{\ell}}_3\ket{u_j}\bra{u_{\ell}}_3=\sum_{j=1}^N\lambda^w_j(x_1 , x_2)p_3^{\omega_j(x_1 ,x_2)}
\end{align}
and for which there exists a unitary $U^{w}(x_1, x_2)$ with matrix elements $\left(U^{w}(x_1, x_2)\right)_{ij}=U_{ij}^w(x_1, x_2)$ such that $\ket{\omega_j(x_1 ,x_2)}_3=\sum_{k=1}^NU^{w}_{jk}\ket{u_k}_3$. We plug \eqref{eq:firstdiag} into \eqref{eq:diagw} and get that
\begin{align}
p_2p_3w_{123}p_3p_2
&=\sum_{i,j,k=1}^N\braket{u_i}{\lambda^w_j(x_1 , x_2)p_3^{\omega_j(x_1 , x_2)}u_k}_2\ket{u_i}\bra{u_k}_2.\label{eq:diagw2}
\end{align}
At this point we would like to repeat the procedure for the sums in $i$ and $k$ but this is impossible for the reason that, inside the bracket, we don't have a multiplication operator anymore. We would need to extract $p_3^{\omega_j(x_1 ,x_2)}$ from the scalar product. The issue is that it depends on $x_2$. The technique will then consist in extracting the projector structure of $p_3^{\omega_j(x_1 ,x_2)}$ from the scalar product while keeping the $x_2$ dependence inside the scalar product. This is possible because the $x_2$ dependence can be left in the unitary trough the formula  $\ket{\omega_j}^{(x_1 ,x_2)}_3=\sum_{k=1}^NU^{w}_{jk}(x_1 ,x_2)\ket{u_k}_3$. We then extract the projector structure from the scalar product
\begin{align}
\braket{u_i}{\lambda^w_j(x_1 , x_2)p_3^{\omega_j(x_1 , x_2)}u_k}_2&=\braket{u_i}{\lambda^w_j(x_1 , x_2)\ket{w_j(x_1 , x_2)}\bra{w_j(x_1 , x_2)}_3u_k}_2\nonumber\\
&=\sum_{p,q=1}^N\braket{u_i}{\lambda^w_j(x_1 , x_2)U^w_{jp}(x_1 , x_2)\overline{U}^w_{jq}(x_1 , x_2)u_k}_2\ket{u_p}\bra{u_q}_3.
\label{eq:cb}
\end{align}
If we now plug \eqref{eq:cb} into \eqref{eq:diagw2} we can diagonalize and obtain
\begin{align}
p_2p_3w_{123}p_3p_2
&=\sum_{p,q,i,j,k=1}^N\braket{u_i}{\lambda^w_j(x_1 , x_2)U^w_{jp}(x_1 , x_2)\overline{U}^w_{jq}(x_1 , x_2)u_k}_2\ket{u_i}\bra{u_k}_2\otimes\ket{u_p}\bra{u_q}_3 \\
&=\sum_{p,q,i,j}^N\mu^w_{ijpq}(x_1)p_2^{\tilde{\omega}_i(x_1)}\otimes\ket{u_p}\bra{u_q}_3\label{eq:mu}
\end{align}
where we diagonalized in the variable $x_2$ for a basis $\{\tilde{\omega}_i(x_1)\}_{i=1}^N$ where we obtain $N$ eigenfunctions $\tilde{\omega}_i$ associated to $j,p,q$ dependent eigenvalues 
\begin{align}
\mu^w_{ijpq}(x_1)
=\braket{\tilde{\omega}_i(x_1)}{\lambda^w_j(x_1 , x_2)U^w_{jp}(x_1 , x_2)\overline{U}^w_{jq}(x_1 , x_2)\tilde{\omega}_i(x_1)}_2.
\end{align}
If we plug them back into \eqref{eq:mu} we obtain
\begin{align*}
p_2p_3w_{123}p_3p_2
&=\sum_{p,q,i,j=1}^N\braket{\tilde{\omega}_i(x_1)}{\lambda^w_j(x_1 , x_2)U^w_{jp}(x_1 , x_2)\overline{U}^w_{jq}(x_1 , x_2)\tilde{\omega}_i(x_1)}_2p_2^{\tilde{\omega}_i(x_1)}\otimes\ket{u_p}\bra{u_q}_3\\
&=\sum_{p,q,i,j=1}^N\braket{\tilde{\omega}_i(x_1)}{\lambda^w_j(x_1 , x_2)p_3^{\omega_j(x_1 ,x_2)}\tilde{\omega}_i(x_1)}_2p_2^{\tilde{\omega}_i(x_1)}
\end{align*}
by reconstructing $p_3^{\omega_j(x_1 ,x_2)}$ inside the bracket and re-using that $\ket{\omega_j}^{(x_1 ,x_2)}_3=\sum_{k=1}^NU^{w}_{jk}\ket{u_k}_3$. This proves Equation \eqref{eq:diagw123}. There is left to prove the trace equivalence of Equation \eqref{eq:eignvw123}. To this end, we compute
\begin{align}
\sum_{i,j=1}^N\braket{\tilde{\omega}^{(x_1)}_i}{\lambda^{w}_j(x_1 ,x_2)\tilde{\omega}^{(x_1)}_i}_2&=\sum_{i,j=1}^N\braket{\tilde{\omega}^{(x_1)}_i}{\braket{u_j}{w_{123}u_j}_3\tilde{\omega}^{(x_1)}_i}_2\nonumber\\
&=\sum_{i=1}^N\braket{\tilde{\omega}^{(x_1)}_i}{(w_{123}\ast\rho)(x_1 ,x_2)\tilde{\omega}^{(x_1)}_i}_2\nonumber\\
&=\sum_{k,p,i=1}^N\braket{u_k}{\overline{U}_{ik}^{\tilde{w}}(x_1)U_{ip}^{\tilde{w}}(x_1)(w_{123}\ast\rho)(x_1 ,x_2)u_p}_2\nonumber\\
&=\sum_{k,p=1}^N\braket{u_k}{\delta_{kp}(w_{123}\ast\rho)(x_1 ,x_2)u_p}_2\nonumber\\
&=\int_{\R^4}w(x_1,x_2,x_3)\rho(x_2)\rho(x_3)\mathrm{d}x_2\mathrm{d}x_3.\label{eq:eignvw123bis}
\end{align}
where we used that there exists a  unitary $U^{\tilde{w}}(x_1)$ such that $\ket{\tilde{\omega}_j}^{(x_1 )}_2=\sum_{k=1}^NU^{\tilde{w}}_{jk}(x_1)\ket{u_k}_2$.
\end{proof}
We now apply the above lemma to a case of interest, i.e , when  $w(x_1,x_2,x_3)$ is endowed with the product structure $\mathbf{a}_N(x_1 -x_2)\cdot\mathbf{a}_N(x_1 -x_3)$.

\begin{corollary}[Diagonalisation of the three-body product operator]\label{cor:diag3bprod}
Consider the eigenvalues, eigenfunctions couple $\{\lambda^{(\mathbf{a})}_j, \sigma_j\}_{j=1}^N$, associated to the operator $\mathbf{a}_N$, obtained in Lemma \ref{lem:diag2b}. We have that the non-intricate operator $p_2p_3\mathbf{a}_N(x_1 -x_2)\cdot\mathbf{a}_N(x_1 -x_3)p_3p_2$ can be factorized and takes the form
 \begin{equation}\label{eq:diagw123prod}
p_2p_3\mathbf{a}_N(x_1 -x_2)\cdot\mathbf{a}_N(x_1 -x_3)p_3p_2=\sum_{i=1}^N\sum_{j=1}^N\lambda^{(\mathbf{a})}_j(x_1)\cdot\lambda^{(\mathbf{a})}_i(x_1)p^{\sigma_j(x_1)}_3p_2^{\sigma_i(x_1)}.
\end{equation}
Moreover, there exist a basis $\{\omega_j(x_1, x_2)\}_{j=1}^N$ and a sequence $\{\lambda^{(\mathbf{a}\cdot \mathbf{a})}_j(x_1 ,x_2)\}_{j=1}^N$ such that the intricate product operator
  \begin{equation}\label{eq:diagw123}
p_2p_3\mathbf{a}_N(x_2 -x_1)\cdot\mathbf{a}_N(x_2 -x_3)p_3p_2=\sum_{i=1}^N\sum_{j=1}^N\braket{\omega^{(x_1)}_i}{\lambda^{(\mathbf{a}\cdot \mathbf{a})}_j(x_1 ,x_2)p^{\sigma^{x_2}_j}_3\omega^{(x_1)}_i}_2p_2^{\omega^{x_1}_i}\end{equation}
in which $\lambda^{(\mathbf{a}\cdot \mathbf{a})}_j(x_1 ,x_2)=\mathbf{a}_N(x_2 -x_1)\cdot\lambda^{(\mathbf{a})}_j(x_2)$.
The above identities come with the traces
\begin{equation}\label{eq:eignvw123prod}
\sum_{i=1}^N\sum_{j=1}^N\lambda^{(\mathbf{a})}_j(x_1)\cdot \lambda^{(\mathbf{a})}_i(x_1)=|\mathbf{a}_N\ast\rho|^2(x_1),
\end{equation}
\begin{equation}\label{eq:trace3b}
\sum_{i=1}^N\sum_{j=1}^N\braket{\omega^{(x_1)}_i}{\lambda^{(\mathbf{a}\cdot \mathbf{a})}_j(x_1 ,x_2)\omega^{(x_1)}_i}_2=-\mathbf{a}_N\ast(\mathbf{a}_N\ast\rho)(x_1).
\end{equation}
\end{corollary}
\begin{proof}
The first diagonalization formula \eqref{eq:diagw123prod} is straightforward due to the product form and the fact that the operator only repeats the variable $x_1$ which is not involved in the projections. As previously done in the proof of Lemma \ref{lem:diag2b}, we write, with $\mathbf{a}_N(x_1 -x_2)\cdot\mathbf{a}_N(x_1 -x_3):=w_{123}$ and $\mathbf{a}_N(x_1 -x_2):=\mathbf{w}_{12}$,
\begin{align}
p_2p_3w_{123}p_3p_2&=\sum_{i,j,k,\ell=1}^N\braket{u_i(x_2)u_j(x_3)}{\mathbf{w}_{12}\cdot\mathbf{w}_{13}u_k(x_2)u_{\ell}(x_3)}\ket{u_i(x_2)u_j(x_3)}\bra{u_k(x_2)u_{\ell}(x_3)}\nonumber\\
&=\sum_{i,k=1}^N\braket{u_j}{\mathbf{w}_{12}u_k}_2\ket{u_i}\bra{u_k}_2\otimes\sum_{j,\ell=1}^N\braket{u_i}{\mathbf{w}_{13}u_{\ell}}_3\ket{u_j}\bra{u_{\ell}}_3\nonumber\\
&=\sum_{i=1}^N\sum_{j=1}^N\lambda^{(\mathbf{a})}_j(x_1)\cdot\lambda^{(\mathbf{a})}_i(x_1)p^{\sigma_j(x_1)}_3p_2^{\sigma_i(x_1)}\label{eq:diagwprod}
\end{align}
by the diagonalisation of 
\begin{align*}
\sum_{i,k=1}^N\braket{u_j}{\mathbf{w}_{12}u_k}_2\ket{u_i}\bra{u_k}_2=\sum_{i}^N\lambda^{(\mathbf{a})}_i(x_1)\ket{\sigma_i(x_1)}\bra{\sigma_i(x_1)}_2
\end{align*}
of \eqref{def:diagmu} from where we obtain that
\begin{align}
\sum_{i}^N\lambda^{(\mathbf{a})}_i(x_1)=\sum_{i}^N\braket{\sigma_i(x_1)}{\mathbf{w}_{12}\sigma_i(x_1)}_2=\sum_{i}^N\braket{u_i}{\mathbf{w}_{12}u_i}_2=\int_{\R^2}\mathbf{a}_N(x_1 -x_2)\rho(x_2)\mathrm{d}x_2
\end{align}
by invariance of the trace. This concludes the proof for the product operator. The second part of the Corollary comes from the direct application of Lemma \ref{lem:diag3b} using that
\begin{align}\label{eq:firstdiag2}
\lambda^{(\mathbf{a}\cdot\mathbf{a})}_j(x_1 , x_2)=\braket{u_j}{w_{123}u_{j}}_3=\mathbf{a}_N(x_2 -x_1)\cdot \lambda^{(\mathbf{a})}_j(x_2)
\end{align}
as obtained in Equation \eqref{eq:firstdiag}.
\end{proof}
We now bound the traces coming from three-body operators in the next lemma.
\begin{lemma}\label{lem:wpp}
Denote $w_{123}:=\mathbf{a}_N(x_1 -x_2)\cdot\mathbf{a}_N(x_1 -x_3)$, we have for any $0\leq s<1/2$, that
\begin{align}
p_2p_3(w_{123}+w_{213}+ w_{321})^2p_3p_2\lesssim N^{-2}\norm{|\mathbf{a}_N|^2\ast\rho}_{\infty}^2\label{ine:ppw}
\end{align}
as operator on wave functions antisymmetric in all the variables except $x_j$, $j\in [1,4,\dots, N]$.
\end{lemma}
\begin{proof}
The symmetry of the operator in the variables $ x_2 ,x_3$ allows to reduce the study on, first
\begin{align*}
p_2p_3w^2_{123}p_3p_2&\leq p_2p_3|\mathbf{a}_N(x_1 -x_2)|^2|\mathbf{a}_N(x_1 -x_3)|^2p_3p_2\\
&\leq  \sum_{i,j=1}^N\lambda_i^{(\mathbf{a}^2)}(x_1)\lambda_j^{(\mathbf{a}^2)}(x_1)p_2^{\tilde{\chi}_j} p_3^{\tilde{\chi}_i}\\
& \lesssim N^{-2}\norm{|\mathbf{a}_N|^2\ast\rho}^2_{\infty}
\end{align*}
obtained using Corollary \ref{cor:diag3bprod} for $|\mathbf{a}_N|^2$, that $\lambda_j^{(\mathbf{a}^2)}(x_1)\geq 0$ and applying the exclusion principle for $p_2^{\tilde{\chi}_j} p_3^{\tilde{\chi}_i}$ and second,
\begin{align*}
p_2p_3w^2_{213}p_3p_2&\leq p_2p_3|\mathbf{a}_N(x_2 -x_1)|^2|\mathbf{a}_N(x_2 -x_3)|^2p_3p_2\\
&\leq \sum_{i=1}^N\sum_{j=1}^N\braket{\tilde{\omega}^{(x_1)}_i}{|\mathbf{a}_N|^2(x_2-x_1)\lambda^{(\mathbf{a}^2)}_j(x_1 ,x_2)p^{\eta_j(x_2)}_3\tilde{\omega}^{(x_1)}_i}_2p_2^{\tilde{\omega}^{x_1}_i}\\
&\lesssim N^{-2}\norm{|\mathbf{a}_N|^2\ast\left(\mathbf{a}^2[\rho]\rho\right)}_{\infty}\lesssim N^{-2}\norm{|\mathbf{a}_N|^2\ast\rho}^2_{\infty}
\end{align*}
obtained by \eqref{eq:firstdiag} and \eqref{eq:diagw123} of Corollary \ref{cor:diag3bprod} applied on $|\mathbf{a}_N|^2$, the positivity of the eigenvalues and the exclusion principle lemma \ref{lem:exclusion}.
This proves \eqref{ine:ppw}. \end{proof}

\section{Estimate of the number of excited particles}\label{sec:Gron}
We now enter in the core of the computations whose aim is to bound $\alpha_m(t)$ to apply a Gr\"onwall argument. Let us recall that $\Psi_N(t)$ is the solution of
\begin{equation*}
	\ii \partial_t \Psi_N(t) = \HNR \Psi_N(t)
\end{equation*}
for initial state $ \Psi_N(0)=\bigwedge_{j=1}^N u_j(0)$ and that $u_j(t)$ is driven by the Hartree motion
	\begin{align*}
		\im \partial_{t}u_j(t,x)=\mathcal{H}_R(x)u_j (t,x).
			\end{align*}
			In this section we will compute the time derivative $\partial_t$ of the function 
\begin{equation}\label{def:alpham}
\alpha_{m}(t):=\braket{\Psi_N(t)}{\hat{m}(1/2)\Psi_N(t)},
\end{equation}
where $\hat{m}(\xi):=\sum_{k=0}^N\left(\frac{k}{N}\right)^{\xi}P^{(k)}_N(t)$ where $P^{(k)}_N$ was defined in \eqref{def:P}. Recall the short-hand notation $\hat{m}:=\hat{m}(1/2)$.
\subsection{Control of $\alpha_{m}(t)$}
By using the identity \eqref{def:eqmotionPN} we obtain that $\ii\partial_t \hat{m}(1/2)=\left[\sum_{m=1}^N\mathcal{H}_R(x_m),\hat{m}(1/2)\right]$ and deduce that
\begin{align}
\ii\partial_t \alpha_{m}(t)&=\braket{\Psi_N(t)}{\left[\sum_{m=1}^N\mathcal{H}_R(x_m)-H_{N,R},\hat{m}(1/2)\right]\Psi_N(t)}\nonumber\\
&=\frac{N}{2}\braket{\Psi_N(t)}{\left[(v'(x_1)+v'(x_2)-(N-1)(v_{12}+v_{21})),\hat{m}(1/2)\right]\Psi_N(t)}\nonumber\\
&\quad +\frac{N}{3}\braket{\Psi_N(t)}{\left[(w'_1+w'_2+w'_3-(N-1)(N-2)(w_{123}+2w_{213})),\hat{m}(1/2)\right]\Psi_N(t)}\nonumber\\
&\quad+\frac{N}{2}(N-1)\braket{\Psi_N(t)}{\left[gN\mathbf{a}_N^2[\rho](x_1)+gN\mathbf{a}_N^2[\rho](x_2)-(gN+2\beta^2)|\mathbf{a}_N(x_1 -x_2)|^2,\hat{m}(1/2)\right]\Psi_N(t)}\nonumber\\
&=: M_V +M_W+M_X
\end{align}
by symmetry of $\hat{m}(1/2)$ and $\Psi_N$, where the operators $v'$, $w'$, $v_{12}$ and $w_{123}$ have been defined in Definition \ref{def:vw}. We will use the above splitting to prove the
\begin{lemma}[\textbf{Central Gronwall argument}]\mbox{}\\
\label{lem:dtam}
 Let $\Psi_N(t)$ be the solution of the Schr\"odinger equation \eqref{def:schro1} with initial data $\Psi_N(0)$, a Slater determinant built from the orthonormal family $\{u_j(0)\}_{j=1}^N$. Let $\omega(t)=\sum_{j=1}^N\ket{u_j}\bra{u_j}$ be a solution of $\mathrm{EQ}(N,R,s,m,\omega(0))$ defined in Definition \ref{def:css}. We then have that
\begin{align}
\left|\partial_t \alpha_m(t)\right|&\lesssim \sqrt{N\norm{|\mathbf{a}_N|^2\ast\rho}_{\infty}}\norm{\nabla_2 q_2\Psi_N}^2+\left(\alpha_m(t)+\frac{1}{\sqrt{N}}\right)\nonumber\\
&\quad \times
 \Big(1+\sqrt{N\norm{|\mathbf{a}_N|^2\ast\rho}_{\infty}}+\sqrt{N\norm{|\mathbf{a}_N|^2\ast\rho_{\nabla}}_{\infty}}+\norm{|\mathbf{a}_N\ast\rho}_{\infty}+\norm{|\mathbf{a}_N|\ast J}_{\infty}\nonumber\\
 &\quad + N^2\norm{|\mathbf{a}_N|^2\ast\rho}^2_{\infty}+\norm{|\mathbf{a}_N|\ast\rho}^{2}_{\infty}+N^{3/2}\norm{|\mathbf{a}_N|^4\ast\rho}^{1/2}_{\infty}+\sqrt{N\norm{\mathbf{a}_N[\rho]}^2_{\infty}\norm{|\mathbf{a}_N|^2\ast\rho}_{\infty}}\Big).\nonumber
 \end{align}
\end{lemma}
\begin{remark}
A direct application of Lemma \ref{lem:LRtrick} to bound the above quantities provides that
\begin{align}\label{eq:mainGron}
\left|\partial_t \alpha_m(t)\right|&\lesssim \norm{\nabla_2 q_2\Psi_N}^2+\left(\alpha_m(t)+\frac{1}{\sqrt{N}}\right),\\
&\text{for any}\begin{cases}
 0\leq s(1+2r)+1/2r<1/4 \text{ in the case of Assumption} \ref{ass:H2}\\
 0\leq s<1/2(1-1/m)\, \,\,\,\quad\quad\text{in the case of Assumption} \ref{ass:Hm}\\
  0\leq s<1/2\,\, \,\,\,\,\quad\quad\quad\quad\quad\quad\text{in the case of Assumption} \ref{ass:Hinfty}.\nonumber
\end{cases}
\end{align}
where $R=N^{-r}$, $r\geq 0$.
\end{remark}
\begin{proof}
The result of Lemma \ref{lem:dtam} directly comes from the inequalities \eqref{ine:MV}, \eqref{ine:MW} and \eqref{ine:MX} respectively controlling the terms $M_V$, $M_W$ and $M_X$. These inequality are derived in the three following subsections.
\end{proof}
In order to reduce the size of the expressions, we might use the notation
$$q_{\neq 1}^{\phi_j(x)}:=\one -\sum_{m=2}^Np_m^{\phi_j(x)}$$ for some $\phi\in L^2(\R^2)$, in many places of the next calculations.
\subsection{The singular term $M_X$, proof of \eqref{ine:MX}}
We start with $M_X$ for which we only use a mean-field cancelation for the $g$ term. The term containing $\beta^2 $ will be shown to be small. We treat it as an error-term by writing
\begin{align}
\left|M_X\right|&\leq \frac{g}{2}N^2(N-1) \left|\braket{\Psi_N(t)}{\left[\mathbf{a}_N^2[\rho](x_1)+\mathbf{a}_N^2[\rho](x_2)-|\mathbf{a}_N(x_1 -x_2)|^2,\hat{m}(1/2)\right]\Psi_N(t)}\right|\nonumber\\
&\quad +\beta^2N(N-1) \left|\braket{\Psi_N(t)}{\left[|\mathbf{a}_N(x_1 -x_2)|^2,\hat{m}(1/2)\right]\Psi_N(t)}\right|\nonumber\\
&=: M_X^{(1)}+M_X^{(2)}\nonumber
\end{align}
For the first term we obtain
\begin{align}
M_X^{(1)}&\leq \frac{g}{2}N^2(N-1)  \left|\braket{\widehat{m^{(1)}}^{1/2}\Psi_N(t)}{q_1p_2\left(\mathbf{a}_N^2[\rho](x_1)-|\mathbf{a}_N(x_1 -x_2)|^2\right)p_1p_2\widehat{m^{(-1)}}^{1/2}\Psi_N(t)}\right|\nonumber\\
&\quad +  \frac{g}{2}N^2(N-1) \left|\braket{\widehat{m^{(2)}}^{1/2}\Psi_N(t)}{q_1q_2|\mathbf{a}_N(x_1 -x_2)|^2p_1p_2\widehat{m^{(-2)}}^{1/2}\Psi_N(t)}\right|\nonumber\\
&\quad +  \frac{g}{2}N^2(N-1)  \left|\braket{\widehat{m^{(1)}}^{1/2}\Psi_N(t)}{q_1q_2\left(\mathbf{a}_N^2[\rho](x_1)-|\mathbf{a}_N(x_1 -x_2)|^2\right)p_1q_2\widehat{m^{(-1)}}^{1/2}\Psi_N(t)}\right|\nonumber\\
&=:M_X^{(1,1)}+M_X^{(1,2)}+M_X^{(1,3)}
\end{align} 
by a direct application of the formula \eqref{eq:W12} in which we took $\hat{f}=\hat{m}$ and the cancelations coming from $pq=0$. We start by studying the main term $M_X^{(1)}$.\newline
\textbf{Study of} $M^{(1,1)}_X$:
We denote $\phi:=\widehat{m^{(1)}}^{1/2}\Psi_N$ and $\tilde{\phi}:=\widehat{m^{(-1)}}^{1/2}\Psi_N$. 
\begin{align*}
\left|M^{(1,1)}_X\right|&\leq CN^2\left|\braket{q_1\phi}{\left(\mathbf{a}_N^2[\rho](x_1)-(N-1)p_2|\mathbf{a}_N(x_1 -x_2)|^2p_2\right)p_1\tilde{\phi}}\right|\\
&\leq CN^2\left|\braket{q_1\phi}{\left(\sum_{j=1}^N\lambda^{\mathbf{a^2}}_j(x_1)-\sum_{m=2}^Np_m|\mathbf{a}_N(x_1 -x_m)|^2p_m\right)p_1\tilde{\phi}}\right|\\
&\leq CN^2\left|\braket{q_1\phi}{\sum_{j=1}^N\lambda^{\mathbf{a^2}}_j(x_1)\left(\one-\sum_{m=2}^Np_m^{\eta_j(x_1)}\right)p_1\tilde{\phi}}\right|\\
\end{align*}
by using Lemma \ref{lem:diag2b}. A Cauchy-Schwarz inequality in the bracket and in the sum over $j$ provides
\begin{align*}
\left|M^{(1,1)}_X\right|^2&\leq CN^4\left|\braket{q_1\phi}{\sum_{j=1}^N\left|\lambda^{\mathbf{a^2}}_j(x_1)\right|^2q_1\phi}\right|\left|\braket{p_1\tilde{\phi}}{\sum_{j=1}^Nq_{\neq 1}^{\eta_j(x_1)}p_1\tilde{\phi}}\right|\\
&\leq CN^4\norm{|\mathbf{a}_N|^4\ast\rho}_{\infty}\norm{q_1\phi}^2\left((N-1)\norm{q_2 \tilde{\phi}}^2+\norm{\tilde{\phi}}^2\right)\\
&\leq CN^3\norm{|\mathbf{a}_N|^4\ast\rho}_{\infty}\alpha_m(t)(\alpha_m(t)+N^{-1/2})
\end{align*}
and the estimates \eqref{ine:m}, \eqref{ine:qm} and \eqref{ine:alphat}.
\textbf{Study of} $M^{(1,2)}_X$:
We denote $\phi:=\widehat{m^{(2)}}^{1/2}\Psi_N$ and $\tilde{\phi}:=\widehat{m^{(-2)}}^{1/2}\Psi_N$. 
\begin{align*}
\left|M^{(1,2)}_X\right|&\leq CN^3\left|\braket{q_2q_1\phi}{|\mathbf{a}_N(x_1 -x_2)|^2p_2p_1\tilde{\phi}}\right|\\
&\leq CN^2\left|\braket{q_1\phi}{\sum_{m=2}^N q_m|\mathbf{a}_N(x_1 -x_2)|^2p_mp_1\tilde{\phi}}\right|\\
&\leq CN^2\norm{q_1\phi}N\braket{\tilde{\phi}}{q_3p_1p_2|\mathbf{a}_N(x_1 -x_2)|^2|\mathbf{a}_N(x_1 -x_3)|^2p_1p_3q_2\tilde{\phi}}^{1/2}\\
&\quad +CN^2\norm{q_1\phi}\sqrt{N}\braket{\tilde{\phi}}{p_1p_2|\mathbf{a}_N(x_1 -x_2)|^2q_2|\mathbf{a}_N(x_1 -x_2)|^2p_1p_2\tilde{\phi}}^{1/2}\\
&\leq CN^3\norm{q_1\phi}\norm{|\mathbf{a}_N(x_1 -x_2)|^2q_3p_1p_2\tilde{\phi}}\\
&\quad +CN^2\sqrt{N}\norm{q_1\phi}\norm{|\mathbf{a}_N(x_1 -x_2)|^2p_1p_2\tilde{\phi}}\\
&\leq CN^{3/2}\norm{|\mathbf{a}_N|^4\ast\rho}^{1/2}_{\infty}\left(\alpha_m(t)+N^{-1/2}\right)
\end{align*}
where we applied \eqref{ine:m} and \eqref{ine:qm}. \newline
\textbf{Study of} $M^{(1,3)}_X$:
We denote $\phi:=\widehat{m^{(1)}}^{1/2}\Psi_N$ and $\tilde{\phi}:=\widehat{m^{(-1)}}^{1/2}\Psi_N$. 
\begin{align*}
\left|M^{(1,3)}_X\right|&\leq CN^3\left|\braket{q_2q_1\phi}{\mathbf{a}_N^2[\rho](x_1)+|\mathbf{a}_N(x_1 -x_2)|^2p_1q_2\tilde{\phi}}\right|\\
&\leq CN^3\left(\norm{\mathbf{a}_N^2[\rho]}^2_{\infty}\norm{q_1\phi}\norm{q_2\tilde{\phi}}+\norm{q_1\phi}\norm{|\mathbf{a}_N(x_1 -x_2)|^2p_1q_2\tilde{\phi}}\right)\\
&\leq C\left(N^2\norm{\mathbf{a}^2_N[\rho]}^2_{\infty}+N^{3/2}\norm{|\mathbf{a}_N|^4\ast\rho}^{1/2}_{\infty}\right)\alpha_m(t)
\end{align*}
where we applied \eqref{ine:qm}. 
To treat the rest $M_X^{(2)}$, we apply formula \eqref{eq:W12} providing the spliting $M_X^{(2)}:=M_X^{(2,1)}+M_X^{(2,2)}+M_X^{(2,3)}$ and bound the terms one by one using the same type of estimates as for $M_X^{(1)}$ but with a $N^{-1}$ factor more.\newline
\textbf{Study of} $M^{(2,1)}_X$:
We denote $\phi:=\widehat{m^{(1)}}^{1/2}\Psi_N$ and $\tilde{\phi}:=\widehat{m^{(-1)}}^{1/2}\Psi_N$. 
\begin{align}
\left|M^{(2,1)}_X\right|&\leq CN^2\left|\braket{p_2q_1\phi}{|\mathbf{a}_N(x_1 -x_2)|^2p_1p_2\tilde{\phi}}\right|\nonumber\\
&\leq CN^2\norm{q_1\phi}\norm{|\mathbf{a}_N(x_1 -x_2)|^2p_1p_2\tilde{\phi}}\label{ine:MX11}\\
&\leq CN^2 N^{-1/2}\sqrt{\alpha_m(t)}\norm{|\mathbf{a}_N(x_1 -x_2)|^2p_1p_2\tilde{\phi}}\nonumber\\
&\leq C N^{3/4}\norm{|\mathbf{a}_N|^4\ast\rho}^{1/2}_{\infty}\nonumber
\end{align}
by using that $\alpha_m \leq 1$, with the estimates \eqref{ine:m} and \eqref{ine:qm}. For the two other terms, we observe that $M_X^{(2,2)}\leq CN^{-1}M_X^{(1,2)}$ and $M_X^{(2,3)}\leq CN^{-1}M_X^{(1,3)}$.
Gathering the previous estimates, we conclude that
\begin{align}
\left|M_X\right|&\lesssim \left(\alpha_m(t)+N^{-1/2}\right) \left(N^2\norm{|\mathbf{a}_N|^2\ast\rho}^2_{\infty}+N^{3/2}\norm{|\mathbf{a}_N|^4\ast\rho}^{1/2}_{\infty}\right)\label{ine:MX}.
\end{align}

\subsection{The two-body term $M_V$, proof of  \eqref{ine:MV}}
Let $W_{12}$ denote $v'(x_1)+v'(x_2)-(N-1)(v_{12}+v_{21})$, note that $W_{12}=W_{21}$. The application of the identity \eqref{eq:W12} with $\hat{m}=\hat{f}$ in which we expand $P_{12}^{(a)}$ and  $P_{12}^{(b)}$ allows to obtain the three terms 
\begin{align}
M_V
&=2\ii N\Im \braket{\Psi_N}{\widehat{m^{(1)}}^{1/2}q_1p_2W_{12}p_1p_2\widehat{m^{(-1)}}^{1/2}\Psi_N}\nonumber\\
&\quad +\ii N\Im \braket{\Psi_N}{\widehat{m^{(2)}}^{1/2}q_1q_2W_{12}p_1p_2\widehat{m^{(-2)}}^{1/2}\Psi_N}\nonumber\\
&\quad +2\ii N\Im \braket{\Psi_N}{\widehat{m^{(1)}}^{1/2}q_1q_2W_{12}p_1q_2\widehat{m^{(-1)}}^{1/2}\Psi_N}\nonumber\\
&=:M^{(1)}_V+M^{(2)}_V+M^{(3)}_V.
\end{align}

\textbf{Study of} $M^{(1)}_V$:
We denote $\phi:=\widehat{m^{(1)}}^{1/2}\Psi_N$ and $\tilde{\phi}:=\widehat{m^{(-1)}}^{1/2}\Psi_N$. 
\begin{align}
M^{(1)}_V&=2\ii N\Im \braket{q_1p_2\phi}{W_{12},p_1p_2\tilde{\phi}}\nonumber\\
&=2\ii N\Im\braket{q_1\phi}{v'(x_1)-\sum_{m=2}^Np_m(v_{1m}+v_{m1})p_m, p_1\tilde{\phi}}\nonumber\\
&=2\ii N\Im\braket{q_1\phi}{\sum_{j=1}^N\lambda^{(\mathbf{a})}_j(x_1)\left(\one -\sum_{m=2}^Np_m^{\sigma_j(x_1)}\right)\cdot (-\ii\nabla_1)p_1\tilde{\phi}}\nonumber\\
&\quad +2\ii N\Im\braket{q_1\phi}{\sum_{j=1}^N\lambda^{(\mathbf{a}\cdot\nabla)}_j(x_1)\left(\one -\sum_{m=2}^Np_m^{\chi_j(x_1)}\right)p_1\tilde{\phi}}\nonumber\\
&=:M^{(1,1)}_V+M^{(1,2)}_V
\end{align}
by using the diagonalisation \eqref{def:diaglambda} and \eqref{def:diagmu}. Recall the notation $q_{\neq 1}^{\sigma_j(x_1)}=\left(\one -\sum_{m=2}^Np_m^{\sigma_j(x_1)}\right)$, we first compute
\begin{align}
\left|M^{(1,1)}_V\right|^2&\leq CN^2\left|\braket{q_1\phi}{\sum_{j=1}^Nq_{\neq 1}^{\sigma_j(x_1)}\lambda^{(\mathbf{a})}_j(x_1)\cdot(-\ii\nabla_1)p_1\tilde{\phi}}\right|^2\nonumber\\
&\leq CN^2\left|\sum_{j=1}^N\norm{\lambda^{(\mathbf{a})}_j(x_1)q_1\phi}\norm{q_{\neq 1}^{\sigma_j(x_1)}(-\ii\nabla_1)p_1\tilde{\phi}}\right|^2\nonumber\\
&\leq CN^2\left|\braket{(-\ii\nabla_1)p_1\tilde{\phi}}{\sum_{j=1}^Nq_{\neq 1}^{\sigma_j(x_1)}(-\ii\nabla_1)p_1\tilde{\phi}}\right|\left|\braket{q_1\phi}{\sum_{j=1}^N|\lambda^{(\mathbf{a})}_j(x_1)|^2q_1\phi}\right|
\end{align}
We applied twice the Cauchy-Schwarz inequality, first in the scalar product while keeping the sum outside and then in the index of the sum over $j$. We now use that $|\braket{\sigma_j}{\mathbf{a}_N \sigma_j}|^2\leq \norm{\sigma_j}_2^2\norm{\mathbf{a}_N \sigma_j}_2^2$ to get
\begin{align}
\left|M^{(1,1)}_V\right|^2&\leq N^2\sup_{x_1}\sum_{j=1}^N|\lambda^{(\mathbf{a})}_j(x_1)|^2\norm{q_1\phi}^2\left((N-1)\norm{q_2\tilde{\phi}}^2 +\norm{\tilde{\phi}}^2\right)\nonumber\\
&\leq N^2\norm{|\mathbf{a}_N|^2\ast\rho}_{\infty}\norm{q_2\tilde{\phi}}^2(\alpha_m(t)+N^{-1/2})\nonumber\\
&\leq N\norm{|\mathbf{a}_N|^2\ast\rho}_{\infty}\alpha_m(t)(\alpha_m(t)+N^{-1/2})
\end{align}
by using the estimates \eqref{ine:alphat}, and \eqref{ine:qm}, \eqref{ine:m} and \eqref{ine:Delta}.  We conclude that 
\begin{equation}
\left|M^{(1,1)}_V\right|\lesssim \sqrt{N\norm{|\mathbf{a}_N|^2\ast\rho}_{\infty}}\left(\alpha_m(t)+N^{-1/2}\right)\label{ine:MV11}
\end{equation}
We treat $M^{(1,2)}_V$. By using Cauchy-Schwarz's inequality, we have
\begin{align}
\left|M^{(1,2)}_V\right|^2&\leq C N^2\left|\braket{p_1\tilde{\phi}}{\sum_{j=1}^N\lambda^{(\mathbf{a}\cdot\nabla)}_j(x_1)\left(\one -\sum_{m=2}^Np_m^{\chi_j(x_1)}\right)q_1\phi}\right|^2\nonumber\\
&\leq C N^2\left|\braket{q_1\phi}{\sum_{j=1}^N\left|\lambda^{(\mathbf{a}\cdot\nabla)}_j(x_1)\right|^2q_1\phi}\right|\left|\braket{p_1\tilde{\phi}}{\sum_{j=1}^Nq_{\neq 1}^{\sigma_j(x_1)}p_1\tilde{\phi}}\right|\nonumber\\
&\leq CN^2\norm{\sum_{j=1}^N\left|\lambda^{(\mathbf{a}\cdot\nabla)}_j(x_1)\right|^2}_{\infty}\norm{q_1\phi}_2^2\left((N-1)\norm{q_2p_1\tilde{\phi}}_2^2+\norm{\tilde{\phi}}_2^2\right)\nonumber\\
&\leq CN^2\norm{\sum_{j=1}^N\left|\braket{u_j}{v_{21}^2 u_j}_2\right|}_{\infty}N^{-1}\alpha_m(t)\left(\alpha_m(t)+N^{-1/2}\right)\nonumber\\
&\lesssim  N\norm{|\mathbf{a}_N|^2\ast\rho_{\nabla}}_{\infty}\alpha_m(t)\left(\alpha_m(t)+\frac{1}{\sqrt{N}}\right) \label{ine:MV12}
\end{align}
by using Cauchy-Schwarz, \eqref{ine:alphat}, \eqref{ine:qm} and \eqref{ine:m}. \newline
\textbf{Study of} $M^{(2)}_V$:
This time we denote $\phi:=\widehat{m^{(2)}}^{1/2}\Psi_N$ and $\tilde{\phi}:=\widehat{m^{(-2)}}^{1/2}\Psi_N$. We use the symmetry of $\phi$ and $\tilde{\phi}$ to write
 \begin{align}
M^{(2)}_V&=\ii N\Im \braket{q_1q_2\phi}{W_{12},p_1p_2\tilde{\phi}}\nonumber\\
&=-\ii N\Im\braket{q_1\phi}{\sum_{m=2}^Nq_m(v_{1m}+v_{m1})p_m, p_1\tilde{\phi}}\nonumber
\end{align}
We then estimate
\begin{align}
\left|M^{(2)}_V\right|^2&\leq CN^2\norm{q_1\phi}^2\norm{\sum_{m=2}^Nq_m(v_{1m}+v_{m1})p_m, p_1\tilde{\phi}}^2.\nonumber\\
&\lesssim N\alpha_m(t)\sum_{m,k=2}^N\braket{q_m(v_{1m}+v_{m1})p_1p_m}{q_k(v_{1k}+v_{k1})p_1p_k}_{\tilde{\phi}}\nonumber\\
&\lesssim N\alpha_m(t)(N-1)\braket{p_1p_2}{(v_{12}+v_{21})q_2(v_{12}+v_{21})p_1p_2}_{\tilde{\phi}}\nonumber\\
&\quad +N\alpha_m(t)(N-1)(N-2)\braket{q_3p_1p_2}{(v_{12}+v_{21})(v_{13}+v_{31})q_2p_1p_2}_{\tilde{\phi}}\nonumber\\
&\lesssim N\alpha_m(t)(N-1)\braket{p_1p_2}{(v_{12}^2 + v_{21}^2)p_1p_2}_{\tilde{\phi}}\nonumber\\
&\quad +CN^3\alpha_m(t)\braket{q_3p_1p_2}{(v_{12}+v_{21})(v_{13}+v_{31})q_2p_1p_2}_{\tilde{\phi}}\nonumber\\
&\lesssim N\alpha_m(t)\left((N-1)\braket{p_1p_2}{v_{12}^2p_1p_2}_{\tilde{\phi}}+N^2\braket{p_1p_2q_3}{v_{12}^2p_1p_2q_3}_{\tilde{\phi}}\right)\nonumber\\
&\lesssim \alpha_m(t)N\norm{|\mathbf{a}_N|^2\ast\rho}_{\infty}\left(\norm{\tilde{\phi}}^2+N\norm{q_3\tilde{\phi}}^2\right)\nonumber\\
&\lesssim N\norm{|\mathbf{a}_N|^2\ast\rho}_{\infty}\alpha_m(t)\left(\alpha_m(t)+N^{-1/2}\right)\label{ine:MV2}
\end{align}
 by using \eqref{ine:qm}, \eqref{ine:m} together with the symmetry of $\tilde{\phi}$ and some Cauchy-Schwarz's inequalities.
\textbf{Study of} $M^{(3)}_V$:
we denote $\phi:=\widehat{m^{(1)}}^{1/2}\Psi_N$ and $\tilde{\phi}:=\widehat{m^{(-1)}}^{1/2}\Psi_N$. 
\begin{align*}
M^{(3)}_V
&=2\ii N\Im \braket{q_1q_2\phi}{W_{12}p_1q_2\tilde{\phi}}_{\Psi_N}\nonumber\\
&=2\ii N\Im\braket{q_1q_2\phi}{(v'(x_1)-(N-1)(v_{12}+v_{21})) p_1q_2\tilde{\phi}}\\
&=2\ii N(N-1)\Im\braket{q_1q_2\phi}{(v_{12}+v_{21}) p_1q_2\tilde{\phi}}\\
&\quad +2\ii N\Im\braket{q_1q_2\phi}{v'(x_1)p_1q_2\tilde{\phi}}\\
&=4\ii N(N-1)\Im\braket{-\ii(\nabla_1 -\nabla_2)q_1q_2\phi}{\mathbf{a}_N(x_1 -x_2)p_1q_2\tilde{\phi}}\\
&\quad +2\ii N\Im\braket{q_1q_2\phi}{v'(x_1)p_1q_2\tilde{\phi}}.
\end{align*}
Note that, in the above, $\nabla_1$ could be applied to $p_1$ and be controlled a priori. This is anyway not the case for $\nabla_2$ that has to hit a $q_2$.  
By Cauchy-Schwarz, we obtain
\begin{align}
\left|M^{(3)}_V\right|^2&\leq CN^4\norm{\mathbf{a}_N(x_1 -x_2)p_1q_2\tilde{\phi}}^2\left(\norm{\nabla_1q_1q_2\phi}^2+\norm{\nabla_2q_1q_2\phi}^2\right)\nonumber\\
&\quad +CN^2\norm{q_1q_2\phi}^2\norm{v'(x_1)p_1q_2\tilde{\phi}}^2\nonumber\\
&\lesssim N^3\norm{|\mathbf{a}_N|^2\ast\rho}_{\infty}\norm{q_2\tilde{\phi}}^2\norm{\nabla_1q_1q_2\phi}^2\nonumber \\
&\quad +CN^2\norm{q_2\tilde{\phi}}^2\norm{q_2\phi}^2\left(\norm{|\mathbf{a}_N\ast\rho}^2_{\infty}+\norm{|\mathbf{a}_N|\ast J}_{\infty}^2\right)
\label{ine:MV3}
\end{align}
by Lemma \ref{lem:pq} for the term with $v'$, estimates  \eqref{ine:qm}, \eqref{ine:nabm}. We conclude that
\begin{align}
\left|M^{(3)}_V\right|&\lesssim \left(\sqrt{N\norm{|\mathbf{a}_N|^2\ast\rho}_{\infty}}+\norm{|\mathbf{a}_N\ast\rho}_{\infty}+\norm{|\mathbf{a}_N|\ast J}_{\infty}\right)\alpha_m(t)\nonumber\\
&\quad +\sqrt{N\norm{|\mathbf{a}_N|^2\ast\rho}_{\infty}}\norm{\nabla_2 q_2\Psi_N}^2
\end{align}
and then, by gathering the estimates for $M_V^{(2)}$ of \eqref{ine:MV2}, $M_V^{(1,2)}$ of \eqref{ine:MV12} and $M_V^{(1,1)}$ of \eqref{ine:MV11}, that
\begin{align}
\left|M_V\right|&\lesssim \sqrt{N\norm{|\mathbf{a}_N|^2\ast\rho}_{\infty}}\norm{\nabla_2 q_2\Psi_N}^2+  \left(\alpha_m(t)+\frac{1}{\sqrt{N}}\right)\nonumber\\
&\quad \times\left(\sqrt{N\norm{|\mathbf{a}_N|^2\ast\rho}_{\infty}}+\sqrt{N\norm{|\mathbf{a}_N|^2\ast\rho_{\nabla}}_{\infty}}+\norm{|\mathbf{a}_N\ast\rho}_{\infty}+\norm{|\mathbf{a}_N|\ast J}_{\infty}\right)\label{ine:MV}
\end{align}

\subsection{The three-body term $M_W$, proof of  \eqref{ine:MW}}\label{sec:MW}
We denote
$$W_{123}:=w'_1 +w'_2 +w'_3-(N-1)(N-2)(w_{123}+w_{213}+w_{312})$$ 
as well as 
$$\tilde{w}(x_1,x_2,x_3):=w_{123}+w_{213}+w_{312}$$
which is symmetric under change of variables. We use the identity \eqref{eq:W123} in which we expand $P_{123}^{(a)}$ and  $P_{123}^{(b)}$ to get the following six terms
\begin{align}
M_W&=2\ii N  \Im \braket{\Psi_N}{\widehat{m^{(1)}}^{1/2}q_1p_2p_3W_{123}p_1p_2p_3\widehat{m^{(-1)}}^{1/2}\Psi_N}\nonumber\\
&\quad + 2\ii N  \Im \braket{\Psi_N}{\widehat{m^{(1)}}^{1/2}q_1q_2q_3W_{123}q_1q_2p_3\widehat{m^{(-1)}}^{1/2}\Psi_N}\nonumber\\
&\quad +2\ii N  \Im \braket{\Psi_N}{\widehat{m^{(2)}}^{1/2}q_1q_2p_3W_{123}p_1p_2p_3\widehat{m^{(-2)}}^{1/2}\Psi_N}\nonumber\\
&\quad +2\ii N  \Im \braket{\Psi_N}{\widehat{m^{(2)}}^{1/2}q_1q_2q_3W_{123}q_1p_2p_3\widehat{m^{(-2)}}^{1/2}\Psi_N}\nonumber\\
&\quad + 6\ii N  \Im \braket{\Psi_N}{\widehat{m^{(1)}}^{1/2}q_1q_2p_3W_{123}q_1p_2p_3\widehat{m^{(-1)}}^{1/2}\Psi_N}\nonumber\\
&\quad + \frac{2\ii}{3} N  \Im \braket{\Psi_N}{\widehat{m^{(3)}}^{1/2}q_1q_2q_3W_{123}p_1p_2p_3\widehat{m^{(-3)}}^{1/2}\Psi_N}\nonumber\\
&=:\sum_{i=1}^6 M_W^{(i)}
\end{align}
using symmetry and \eqref{eq:transdev}. We will now treat them one by one.

\textbf{Study of} $M_W^{(1)}$ :
We denote $\phi:=\widehat{m^{(1)}}^{1/2}\Psi_N$ and $\tilde{\phi}:=\widehat{m^{(-1)}}^{1/2}\Psi_N$. As it was previously the case for the two-body interaction, this first term uses the mean-field cancelation. 
\begin{align}
M_W^{(1)}&=2\ii N  \Im \braket{q_1\phi}{p_2p_3W_{123}p_2p_3 p_1\tilde{\phi}}\nonumber\\
&=2\ii N  \Im \braket{q_1\phi}{p_2p_3\left(w'_1-(N-1)(N-2)(w_{123}+2w_{213})\right)p_2p_3 p_1\tilde{\phi}}\nonumber\\
&=:M_W^{(1,1)}+M_W^{(1,2)}
\end{align}
The first term to treat is the non-intricate 
\begin{align}
M_W^{(1,1)}&:=2\ii N \Im \braket{q_1\phi}{\left(\beta^2\mathbf{a}_N[\rho]^2(x_1)-(N-1)(N-2) p_2p_3w_{123}p_2p_3\right) p_1\tilde{\phi}}.\nonumber
\end{align}
for which we apply the diagonalisation \eqref{eq:diagw123prod}. We get
\begin{align}
M^{(1,1)}_W
&=2\ii\Im\braket{q_1\phi}{\left(\beta^2\mathbf{a}_N\left [\rho_u\right ]^2(x_1)-\sum_{m=2}^N\sum_{n=3}^N p_m p_n w(x_1, x_m ,x_n)p_mp_n\right)p_1\tilde{\phi}}\label{def:MW11}\nonumber\\
&=2\ii\beta^2\Im\braket{q_1\phi}{\left(\left(\sum_{j=1}^N \lambda^{(\mathbf{a})}_j(x_1)\right)^2 -\sum_{i,j=1}^N \lambda^{(\mathbf{a})}_j(x_1) \cdot\lambda^{(\mathbf{a})}_i(x_1)\sum_{m=2}^N p^{\sigma_j(x_1)}_m \sum_{n=3}^Np^{\sigma_i(x_1)}_n \right)p_1\tilde{\phi}}\nonumber\\
&=2\ii\beta^2\Im\braket{q_1\phi}{\left(\sum_{i,j=1}^N \lambda^{(\mathbf{a})}_i(x_1)\cdot\lambda^{(\mathbf{a})}_j(x_1)\right)\left(\one -\sum_{m=2}^N p^{\sigma_j(x_1)}_m \sum_{n=3}^Np^{\sigma_i(x_1)}_n \right)p_1\tilde{\phi}}
\end{align}
we pull the two sums out of the bracket and apply a Cauchy-Schwarz in the bracket, one in $i$ followed by one in $j$ to obtain
\begin{align}
\left|M_W^{(1,1)}\right|^2&\leq CN^{2}\left|\braket{q_1\phi }{\left(\sum_{j=1}^N \left|\lambda^{(\mathbf{a})}_j(x_1)\right|^2\right)^2q_1\phi}\right|\left|\braket{p_1\tilde{\phi}}{\sum_{i,j=1}^N\left(\one -\sum_{m=2}^N p^{\sigma_j(x_1)}_m \sum_{n=3}^Np^{\sigma_i(x_1)}_n \right)p_1\tilde{\phi}}\right|\nonumber\\
&\lesssim N^{2}\norm{|\mathbf{a}_N|^2\ast \rho}^2_{\infty}\norm{q_1\phi}^2\left(N^2\norm{q_2p_1\tilde{\phi}}^2+N\norm{p_1\tilde{\phi}}^2\right)\\
&\lesssim N^{2}\norm{|\mathbf{a}_N|^2\ast \rho}^2_{\infty}\alpha_m(t)\left(\alpha_m(t)+N^{-1/2}\right)\label{ine:M11}
\end{align}
by $p_1\leq 1$ and the estimates \eqref{ine:m} and \eqref{ine:alphat2}.
The following term is the intricate product diagonalised in Equation \eqref{eq:diagw123} of Corollary \ref{cor:diag3bprod}.
\begin{align*}
M_W^{(1,2)}&=4\ii N \Im \braket{q_1\phi}{\left(-(\beta\mathbf{a}_N\ast\beta\mathbf{a}_N[\rho]\rho)(x_1)-(N-1)(N-2) p_2p_3 w_{213}p_2p_3\right) p_1\tilde{\phi}}
\end{align*}
We decompose the diagonalised part into three terms. The direct trace, to be canceled with $\mathbf{a}_N\ast\mathbf{a}_N[\rho]\rho$, and two error terms $e^{(1)}$ and $e^{(2)}$ as follows
\begin{align}
&(N-1)(N-2)p_2p_3w(x_2, x_1 ,x_3)p_2p_3=\sum_{m=2}^N\sum_{n=3}^N p_m p_n w(x_m, x_1 ,x_n)p_mp_n \label{diag:3bb}\\
&\quad\quad\quad\quad =\sum_{i,j=1}^N\sum_{m=2}^N\sum_{n=3}^N\braket{\omega^{(x_1)}_i}{\lambda^{(\mathbf{a}\cdot\mathbf{a})}_j(x_1 ,x_2)p^{\sigma_j(x_2)}_n\omega^{(x_1)}_i}_2p_m^{\omega_i(x_1)}\nonumber\\
&\quad\quad\quad\quad =\sum_{i,j=1}^N\braket{\omega^{(x_1)}_i}{\lambda^{(\mathbf{a}\cdot\mathbf{a})}_j(x_1 ,x_2)\omega^{(x_1)}_i}_2\nonumber\\
&\quad\quad\quad\quad\quad -\sum_{i,j=1}^N\braket{\omega^{(x_1)}_i}{\lambda^{(\mathbf{a}\cdot\mathbf{a})}_j(x_1 ,x_2)\omega^{(x_1)}_i}_2\left(\one -\sum_{m=2}^Np_m^{\omega_i(x_1)}\right)\nonumber\\
&\quad\quad\quad\quad\quad -\sum_{i,j=1}^N\braket{\omega^{(x_1)}_i}{\lambda^{(\mathbf{a}\cdot\mathbf{a})}_j(x_1 ,x_2)\left(\one -\sum_{n=3}^Np^{\sigma_j(x_1)}_n\right)\omega_i(x_1)}_2\sum_{m=2}^Np_m^{\omega^{(x_1)}_i}\nonumber\\
&\quad\quad\quad\quad =:\int_{\R^4}\mathbf{a}_N(x_2 -x_1)\cdot\mathbf{a}_N(x_2-x_3)\rho(x_2)\rho(x_3)\mathrm{d}x_2\mathrm{d}x_3+e^{(1)}+e^{(2)}\nonumber\\
&\quad\quad\quad\quad=-(\mathbf{a}_N\ast\mathbf{a}_N[\rho]\rho)(x_1)+e^{(1)}+e^{(2)}\label{diag:3be}
\end{align}
where we used the trace property \eqref{eq:trace3b} of Corollary \ref{cor:diag3bprod}. We then obtain

\begin{equation}
M_W^{(1,2)}=2\ii\beta^2 N \Im \braket{q_1\phi}{(e^{(1)}+e^{(2)} )p_1\tilde{\phi}}=:M_W^{(1,2,1)}+M_W^{(1,2,2)}.
\end{equation} 
For $M_W^{(1,2,1)}$, we find
\begin{align}
\left|M_W^{(1,2,1)}\right|^2&\lesssim N^{2}\left|\braket{p_1\tilde{\phi}}{ \sum_{i=1}^N\braket{\omega^{(x_1)}_i}{\mathbf{a}_N(x_2-x_1)\cdot\sum_{j=1}^N \lambda_j^{(\mathbf{a})}(x_2)\omega^{(x_1)}_i}_2\left(\one -\sum_{m=2}^Np_m^{\omega_i(x_1)}\right)q_1\phi}\right|^2\nonumber\\
&\lesssim  N^2\left|\braket{q_1\phi}{ \sum_{i=1}^N\left|\braket{\omega^{(x_1)}_i}{\mathbf{a}_N(x_2-x_1)\cdot\sum_{j=1}^N \lambda_j^{(\mathbf{a})}(x_2)\omega_i(x_1)}_2\right|^2q_1\phi}\right|\nonumber\\
&\quad\quad\quad\quad\quad\quad\times\left|\braket{p_1\tilde{\phi}}{ \sum_{i=1}^N\left(\one -\sum_{m=2}^Np_m^{\omega_i(x_1)}\right)p_1\tilde{\phi}}\right|\nonumber\\
&\lesssim  N^2\sup_{x_2}\left|\sum_{j=1}^N \lambda_j^{(\mathbf{a})}(x_2)\right|^2\left(N\norm{q_2p_1\tilde{\phi}}^2+\norm{p_1\tilde{\phi}}^2\right)\nonumber\\
&\quad\quad\quad\quad\quad\quad\times \left|\braket{q_1\phi}{ \sum_{i=1}^N\braket{\omega^{(x_1)}_i}{|\mathbf{a}_N(x_2-x_1)|^2\omega^{(x_1)}_i}_2q_1\phi}\right|\nonumber\\
&\lesssim N^2\norm{\mathbf{a}_N[\rho]}^2_{\infty}\norm{|\mathbf{a}_N|^2\ast\rho}_{\infty}\norm{q_1\phi}^2\left(\alpha_m(t)+N^{-1/2}\right)\nonumber\\
&\lesssim N\norm{\mathbf{a}_N[\rho]}^2_{\infty}\norm{|\mathbf{a}_N|^2\ast\rho}_{\infty}\alpha_m(t)\left(\alpha_m(t)+N^{-1/2}\right)\label{ine:M121}
\end{align}
where we used that
\begin{equation}
\sum_{i=1}^N\braket{\omega^{(x_1)}_i}{|\mathbf{a}_N(x_2-x_1)|^2\omega^{(x_1)}_i}_2=\sum_{i=1}^N\braket{u_i}{|\mathbf{a}_N(x_2-x_1)|^2u_i}_2=(|\mathbf{a}_N|^2\ast\rho)(x_1)\nonumber
\end{equation}
because $\mathrm{span} (u_1,\dots , u_N)=\mathrm{span} (\omega_1,\dots , \omega_N)$. We also applied
the estimates \eqref{ine:m},  \eqref{ine:alphat} together with Cauchy-Schwarz's inequalities. To treat 
\begin{align}
M^{(1,2,2)}_W:=4\ii \beta^2N\Im\braket{p_1\tilde{\phi}}{e^{(2)}q_1\phi},\label{def:EW122}
\end{align}
 we calculate
\begin{align}
&\left|M^{(1,2,2)}_W\right|^2\leq CN^2\left|\braket{p_1\tilde{\phi}}{\sum_{i=1}^N\braket{\omega^{(x_1)}_i}{\sum_{j=1}^N\lambda^{(\mathbf{a}\cdot\mathbf{a})}_j(x_1 ,x_2)\left(\one -\sum_{n=3}^Np^{\sigma_j(x_2)}_n\right)\omega^{(x_1)}_i}_2\sum_{m=2}^Np_m^{\omega_i(x_1)}q_1\phi}\right|^2\nonumber\\
&\leq CN^2\left|\braket{q_1\phi}{\sum_{i=1}^N\sum_{m=2}^N p_m^{\omega_i(x_1)}q_1\phi}_{\Psi_N(t)}\right|\left|\braket{p_1\tilde{\phi}}{\sum_{i=1}^N\left|\braket{\omega^{(x_1)}_i}{\sum_{j=1}^N\lambda^{(\mathbf{a}\cdot\mathbf{a})}_j\left(\one -\sum_{n=3}^Np^{\sigma_j(x_2)}_n\right)\omega^{(x_1)}_i}_2\right|^2p_1\tilde{\phi}}\right|\nonumber\\
&\leq CN^2\left|\braket{q_1\phi}{\sum_{m=2}^N p_mq_1\phi}\right|\left|\braket{p_1\tilde{\phi}}{\sum_{i=1}^N\braket{\omega^{(x_1)}_i}{\sum_{j=1}^N|\lambda^{(\mathbf{a}\cdot\mathbf{a})}_j(x_1 ,x_2)|^2\omega^{(x_1)}_i}_2\braket{\omega^{(x_1)}_i}{\sum_{j=1}^Nq_{\neq 2,3}^{\sigma_j(x_2)}\omega^{(x_1)}_i}_2p_1\tilde{\phi}}\right|\nonumber\\
&\leq CN^{2}\alpha_m(t)\norm{|\mathbf{a}_N|^2\ast\rho}_{\infty}\times\nonumber\\
&\quad\quad\left|\braket{p_1\tilde{\phi}}{\sqrt{\sum_{i=1}^N\left|\braket{\omega^{(x_1)}_i}{|\mathbf{a}_N(x_2-x_1)|^2\omega^{(x_1)}_i}_2\right|^2}\sqrt{\sum_{i=1}^N\left|\braket{\omega^{(x_1)}_i}{\sum_{j=1}^Nq_{\neq 2,3}^{\sigma_j(x_2)}\omega^{(x_1)}_i}_2\right|^2}p_1\tilde{\phi}}\right|\nonumber\\
&\leq CN^{2}\alpha_m(t)\norm{|\mathbf{a}_N|^2\ast\rho}^2_{\infty}\left|\braket{p_1\tilde{\phi}}{\sum_{i=1}^N\braket{\omega^{(x_1)}_i}{\sum_{j=1}^Nq_{\neq 2,3}^{\sigma_j(x_2)}\omega^{(x_1)}_i}_2p_1\tilde{\phi}}\right|\nonumber\\
&\leq CN^2\alpha_m(t)\norm{|\mathbf{a}_N|^2\ast \rho}^2_{\infty}\left|\braket{p_1\tilde{\phi}}{\sum_{i=1}^N\braket{\omega^{(x_1)}_i}{\left(N -\sum_{n=3}^Np_n\right)\omega^{(x_1)}_i}_2p_1\tilde{\phi}}\right|\nonumber\\
&\lesssim N^2\alpha_m(t)\norm{|\mathbf{a}_N|^2\ast \rho}^2_{\infty}\left(N\norm{q_3p_1\tilde{\phi}}^2+\norm{p_1\tilde{\phi}}\right)\nonumber\\
&\lesssim N^2\norm{|\mathbf{a}_N|^2\ast \rho}^2_{\infty}\alpha_m(t)\left(\alpha_m(t) +N^{-1/2}\right)\label{ine:M122}
\end{align}
where we used the estimates \eqref{ine:m} and \eqref{ine:alphat} to obtain the last line. If we combine the previous \eqref{ine:M122}, \eqref{ine:M121} and \eqref{ine:M11} we conclude that
\begin{equation}
\left|M_W^{(1)}\right|\lesssim \left(N\norm{|\mathbf{a}_N|^2\ast \rho}_{\infty}+\sqrt{N\norm{\mathbf{a}_N[\rho]}^2_{\infty}\norm{|\mathbf{a}_N|^2\ast\rho}_{\infty}}\right)\left( \alpha_m(t)+\frac{1}{\sqrt{N}}\right).
\end{equation}
\newline
\textbf{Study of} $M_W^{(2)}$ :
We denote $\phi:=\widehat{m^{(1)}}^{1/2}\Psi_N$ and $\tilde{\phi}:=\widehat{m^{(-1)}}^{1/2}\Psi_N$. 

\begin{align}
M^{(2)}_W
&=2\ii N\Im\braket{q_1q_2p_3\tilde{\phi}}{(w'(x_3)-(N-1)(N-2)\tilde{w}(x_1 , x_2, x_3))q_1q_2q_3\phi}\nonumber\\
&=-2\ii N(N-1)(N-2)\Im\braket{q_1q_2p_3\tilde{\phi}}{\tilde{w}(x_1 , x_2, x_3)q_1q_2q_3\phi}\nonumber\\
&\quad +2\ii N\Im\braket{q_1q_2p_3\tilde{\phi}}{w'(x_3)q_1q_2q_3\phi}\nonumber\\
\left|M^{(2)}_W\right|&\leq CN^3\left( \norm{\tilde{w}(x_1 , x_2, x_3)q_1q_2p_3 \tilde{\phi}}+N^{-2}\norm{w'(x_3)q_1q_2p_3 \tilde{\phi}}\right)\norm{q_1q_2q_3 \phi}\nonumber\\
&\leq CN^3\braket{q_1q_2\tilde{\phi}}{p_3|\mathbf{a}_N(x_1 -x_3)|^4p_3q_1q_2\tilde{\phi}}^{1/2}N^{-1/2}\sqrt{\alpha_m(t)}\nonumber\\
&\quad +C\norm{\mathbf{a}_N\left [\rho_u\right ]}^2_{\infty}\alpha_m(t)\nonumber\\
&\lesssim \alpha_m(t)\left(\norm{\mathbf{a}_N\left [\rho_u\right ]}^2_{\infty}+N^{3/2}\norm{|\mathbf{a}_N|^4\ast\rho}^{1/2}_{\infty}\right)\label{ine:M2W}
\end{align}
where we used  \eqref{ine:qm} and that 
$$\norm{w'}_{\infty}=\beta^2\norm{\mathbf{a}_N\left [\rho_u\right ]^2-2\mathbf{a}_N\ast(\mathbf{a}_N\left [\rho_u\right ]\rho_u)}_{\infty}\leq  \beta^2\norm{\mathbf{a}_N\left [\rho_u\right ]}^2_{\infty}.$$

\textbf{Study of} $M_W^{(3)}$ :
We denote $\phi:=\widehat{m^{(2)}}^{1/2}\Psi_N$ and $\tilde{\phi}:=\widehat{m^{(-2)}}^{1/2}\Psi_N$. 
We use the symmetry of $\phi$ and $\tilde{\phi}$ to write
\begin{align}
M^{(3)}_W
&=2\ii N(N-1)(N-2)\Im\braket{q_1q_2p_3\phi}{\tilde{w}(x_1 , x_2, x_3)p_1p_2p_3\tilde{\phi}}\nonumber\\
&=2\ii N\Im\braket{q_1\phi}{\sum_{m=2}^N\sum_{n=3}^Nq_m p_n\tilde{w}(x_1 , x_m, x_n)p_m p_n p_1\tilde{\phi}}\nonumber\\
\left|M^{(3)}_W\right|^2&\leq N^2\sum_{m,m'=2}^N\sum_{n,n'=3}^N\norm{q_1\phi}^2\braket{q_m p_n\tilde{w}(x_1 , x_m, x_n)p_m p_n p_1\tilde{\phi}}{q_{m'} p_{n'}\tilde{w}(x_1 , x_{m'}, x_{n'})p_{m'} p_{n'} p_1\tilde{\phi}}\nonumber\\
&\leq N^2(N-1)(N-2)\norm{q_1\phi}^2\braket{\tilde{w}(x_1 , x_2, x_3)p_2 p_3 p_1\tilde{\phi}}{q_{2} p_{3}\tilde{w}(x_1 , x_{2}, x_{3})p_{2} p_{3} p_1\tilde{\phi}}\label{ine:MW3}\\
 &\quad +N^2(N-1)(N-2)(N-3)(N-4)\norm{q_1\phi}^2\braket{q_2 p_3\tilde{w}(x_1 , x_2, x_3)p_2 p_3 p_1\tilde{\phi}}{q_{4} p_{5}\tilde{w}(x_1 , x_{4}, x_{5})p_{4} p_{5} p_1\tilde{\phi}}\nonumber
\end{align}
using the symmetry to separate the sum into the case $m=m'$, $n=n'$ for the first line of above and the rest for the second line with the property that $q_mp_n=0$ when $m=n$ and respectively when $m'=n'$.
\begin{align}
\left|M^{(3)}_W\right|^2
&\leq N^4N^{-1}\alpha_m(t)\braket{\tilde{\phi}}{p_1p_{2} p_{3}\tilde{w}^2(x_1 , x_{2}, x_{3})p_{2} p_{3} p_1\tilde{\phi}}\nonumber\\
 &\quad +N^6 N^{-1}\alpha_m(t)\braket{ p_3\tilde{w}(x_1 , x_2, x_3)q_4p_2 p_3 p_1\tilde{\phi}}{ p_{5}\tilde{w}(x_1 , x_{4}, x_{5})q_2p_{4} p_{5} p_1\tilde{\phi}}\nonumber\\
 &\leq N\norm{|\mathbf{a}_N|^2\ast\rho}_{\infty}^2\alpha_m(t)\norm{\tilde{\phi}}^2+N^5\alpha_m(t)\braket{q_4p_2 p_3 p_1\tilde{\phi}}{\tilde{w}^2(x_1 , x_2, x_3)q_4p_2 p_3 p_1\tilde{\phi}}\nonumber\\
 &\leq N^2 \norm{|\mathbf{a}_N|^2\ast\rho}^{2}_{\infty}\alpha_m(t)\left(\alpha_m(t)+N^{-1/2}\right)
 \end{align}
by symmetry and by applying \eqref{ine:ppw} of Lemma \ref{lem:wpp} and the estimates \eqref{ine:m} and \eqref{ine:qm}.

\textbf{Study of} $M_W^{(4)}$ :
We denote $\phi:=\widehat{m^{(2)}}^{1/2}\Psi_N$ and $\tilde{\phi}:=\widehat{m^{(-2)}}^{1/2}\Psi_N$. 

\begin{align}
M^{(4)}_W
&=2\ii N(N-1)(N-2)\Im\braket{q_1q_2q_3\phi}{\tilde{w}(x_1 , x_2, x_3)q_1p_2p_3\tilde{\phi}}\nonumber\\
\left|M^{(4)}_W\right|^2&\leq N^6\norm{q_2q_3q_1\phi}^2\norm{\tilde{w}(x_1 , x_2, x_3)q_1p_2p_3\tilde{\phi}}^2\nonumber\\
&\lesssim N^2 \norm{|\mathbf{a}_N|^2\ast\rho}^{2}_{\infty}\alpha_m(t)
\end{align}
by the application of the estimates \eqref{ine:ppw} and \eqref{ine:qm}.

\textbf{Study of} $M_W^{(5)}$ :
We denote $\phi:=\widehat{m^{(1)}}^{1/2}\Psi_N$ and $\tilde{\phi}:=\widehat{m^{(-1)}}^{1/2}\Psi_N$. 

\begin{align}
M^{(5)}_W
&=2\ii N\Im\braket{q_1q_2p_3\tilde{\phi}}{(w'(x_2)-(N-1)(N-2)\tilde{w}(x_1 , x_2, x_3))q_1p_2p_3\phi}\nonumber\\
&=-2\ii N(N-1)(N-2)\Im\braket{q_1q_2p_3\tilde{\phi}}{\tilde{w}(x_1 , x_2, x_3)q_1p_2p_3\phi}\nonumber\\
&\quad +2\ii N\Im\braket{q_1q_2p_3\tilde{\phi}}{w'(x_2)q_1p_2p_3\phi}\nonumber\\
\left|M^{(5)}_W\right|&\leq CN^3\left( \norm{\tilde{w}(x_1 , x_2, x_3)q_1p_2p_3 \phi}+N^{-2}\norm{w'(x_2)q_1p_2p_3 \phi}\right)\norm{q_1q_2p_3\tilde{\phi}}\nonumber\\
&\lesssim \left(N\norm{|\mathbf{a}_N|^2\ast\rho}_{\infty}+\norm{|\mathbf{a}_N|\ast\rho}^{2}_{\infty}\right)\alpha_m(t)
\end{align}
where we used that $\norm{w'}_{\infty}\leq \norm{|\mathbf{a}_N|\ast\rho}^{2}_{\infty}$ as in \eqref{ine:M2W} and the estimates \eqref{ine:ppw} and \eqref{ine:qm}.

\textbf{Study of} $M_W^{(6)}$ :
We denote $\phi:=\widehat{m^{(3)}}^{1/2}\Psi_N$ and $\tilde{\phi}:=\widehat{m^{(-3)}}^{1/2}\Psi_N$. 
We will proceed as for $M_W^{(3)}$ in \eqref{ine:MW3}, using the symmetry to write
\begin{align*}
M^{(6)}_W
&=2\ii N(N-1)(N-2)\Im\braket{q_1q_2q_3\phi}{\tilde{w}(x_1 , x_2, x_3)p_1p_2p_3\tilde{\phi}}\\
&=2\ii N\Im\braket{q_1\phi}{\sum_{m=2}^N\sum_{n=3}^Nq_m q_n\tilde{w}(x_1 , x_m, x_n)p_m p_n p_1\tilde{\phi}}
\end{align*}
and apply Cauchy-Schwarz's inequality to get
\begin{align*}
\left|M^{(6)}_W\right|^2&\leq N^2\sum_{\substack{m,m'=2\\n, n'=3}}^N\norm{q_1\phi}^2\braket{q_m q_n\tilde{w}(x_1 , x_m, x_n)p_m p_n p_1\tilde{\phi}}{q_{m'} q_{n'}\tilde{w}(x_1 , x_{m'}, x_{n'})p_{m'} p_{n'} p_1\tilde{\phi}}\\
&= N^2(N-1)(N-2)\norm{q_1\phi}^2\braket{\tilde{w}(x_1 , x_2, x_3)p_2 p_3 p_1\tilde{\phi}}{q_{2} q_{3}\tilde{w}(x_1 , x_{2}, x_{3})p_{2} p_{3} p_1\tilde{\phi}}\\
 &\quad +N^2(N-1)(N-2)(N-3)(N-4)\norm{q_1\phi}^2\braket{q_2 q_3\tilde{w}(x_1 , x_2, x_3)p_2 p_3 p_1\tilde{\phi}}{q_{4} q_{5}\tilde{w}(x_1 , x_{4}, x_{5})p_{4} p_{5} p_1\tilde{\phi}}\\
 &\leq N^3\alpha_m(t)\norm{\tilde{w}(x_1 , x_2, x_3)p_2 p_3 p_1\tilde{\phi}}^2 +N^5\alpha_m(t)\norm{\tilde{w}(x_1 , x_2, x_3)q_5q_4p_3 p_2 p_1\tilde{\phi}}^2\\
 &\lesssim N^2\norm{|\mathbf{a}_N|^2\ast\rho}^{2}_{\infty}\alpha_m(t)\left(\alpha_m(t)+N^{-1/2}\right)
\end{align*}
by \eqref{ine:ppw} and \eqref{ine:qm}. Gathering the estimates we obtain that
\begin{align}
\sum_{i=1}^6M_W^{(6)}&\lesssim \left(N\norm{|\mathbf{a}_N|^2\ast\rho}_{\infty}+\norm{|\mathbf{a}_N|\ast\rho}^{2}_{\infty}+N^{3/2}\norm{|\mathbf{a}_N|^4\ast\rho}^{1/2}_{\infty}+\sqrt{N\norm{\mathbf{a}_N[\rho]}^2_{\infty}\norm{|\mathbf{a}_N|^2\ast\rho}_{\infty}}\right)\nonumber\\
& \quad \times \left(\alpha_m(t)+\frac{1}{\sqrt{N}}\right)\label{ine:MW}
\end{align}
concluding the proof of Lemma \ref{lem:dtam}.

\section{Control of the kinetic energy}\label{sec:kin}
This section is dedicated to the control of the kinetic energy of the excited particles $\norm{\nabla_1 q_1\Psi_N(t)}$ in terms of the square root number functional $\alpha_m(t)$. Lemma \ref{lem:mf} allows, in a first time, to express the fact that almost all the energy must be contained in the vector $p_1p_2p_3\Psi_N$ using conservation of energy.
The energy at time $t$ is
\begin{equation}
	E_N(t):=\braket{\Psi_N(t)}{\HNR\Psi_N(t)}\quad\text{and}\quad\mathcal{E}_R^{\mathrm{af}}[\omega(t)]:=\tr\left[(-\im\nabla +\mathbf{a}_N[\rho])^2 \omega(t)\right]\nonumber
\end{equation}
and by conservation of the energy of Lemma \eqref{Thm:cons} and the Hartree approximation of Lemma \ref{lem:Ha}, we get
\begin{align*}
	E_N(t)=E_N(0)&=\mathcal{E}_R^{\mathrm{af}}[\omega(0)]+O(N^{s(1+2r)})=\mathcal{E}_R^{\mathrm{af}}[\omega(t)]+O(N^{s(1+2r)})
\end{align*}
and then
\begin{align}
\left|E_N(0)-\mathcal{E}_R^{\mathrm{af}}[\omega(0)]\right|=\left|E_N(t)-\mathcal{E}_R^{\mathrm{af}}[\omega(t)]\right|&\leq O(N^{s(1+2r)}).
\end{align}
\begin{lemma}\label{lem:mf}
	Let $\Phi:=(\one -p_1p_2p_2)\Psi_N(t)$  where $\Psi_N(t)$ is the solution of the Schr\"odinger equation \eqref{def:schro1} with initial data $\Psi_N(0)=\bigwedge_{j=1}^N u_j(0)$. Let $\omega(t)=\sum_{j=1}^N\ket{u_j}\bra{u_j}$ be a solution of $\mathrm{EQ}(N,R,s,m,\omega(0))$ defined in Definition \ref{def:css} and $R=N^{-r}$. We then have that
	\begin{align}
		 E_N(t)-\mathcal{E}_{R}^{\mathrm{af}}[\omega(t)]&=
		 N\braket{\Phi}{\left(\Delta_1+(N-1)v_{12}+(N-1)(N-2)w_{123}+\frac{g}{2}N(N-1)|\mathbf{a}_N|^2_{12}\right)\Phi}\nonumber\\
		 &\quad +\beta^2N(N-1)\braket{\Psi_N}{|\mathbf{a}_N|^2_{12}\Psi_N}+\mathcal{E}
	\end{align}
	where $\left|\mathcal{E}\right|\lesssim \alpha_m(t)+N^{-1/2}$. The above identity is valid
	\begin{align}
\text{for any}\quad\begin{cases}
 0\leq s(1+2r)+1/2r<1/4 \text{ in the case of Assumption} \ref{ass:H2}\\
 0\leq s<1/2(1-1/m)\, \,\,\,\quad\quad\text{in the case of Assumption} \ref{ass:Hm}\\
  0\leq s<1/2\,\, \,\,\,\,\quad\quad\quad\quad\quad\quad\text{in the case of Assumption} \ref{ass:Hinfty}.\nonumber
\end{cases}
\end{align}
	
\end{lemma}
\begin{proof}
First note that using the symmetry of $\Psi_N$, we can rewrite
	\begin{align*}
		\mathcal{E}_{R}^{\mathrm{af}}[\omega(t)] &= \sum_{j=1}^N\braket{u_j}{(-\Delta+ 2\beta(-\ii\nabla)\cdot \mathbf{a}_N[\rho]+  \beta^2|\mathbf{a}_N[\rho]|^2+\frac{g}{2}N\mathbf{a}_N^2[\rho])u_j}\\
		E_N(t) &= N\braket{\Psi_N}{\left(-\Delta_1+(N-1)v_{12}+(N-1)(N-2)w_{123}\right)\Psi_N}\\
		&\quad+N(N-1)\braket{\Psi_N}{(\frac{gN}{2}+\beta^2)|\mathbf{a}_N|^2(x_1-x_2)\Psi_N}
	\end{align*}
where we recall that
	\begin{equation}
	v_{12}= (-\ii\nabla_1)\cdot\mathbf{a}_N(x_1 -x_2)+\hc\quad\text{and}\quad w_{123}:=\mathbf{a}_N(x_1 -x_2)\cdot\mathbf{a}_N(x_1 -x_3).
	\end{equation}
	We now use $\one=(\one +p_1p_2p_3) -p_1p_2p_3 $ and introduce the function $\Phi:=(\one -p_1p_2p_3)\Psi_N(t)$ to obtain the identity
		\begin{align}
		E_N(t)-\mathcal{E}_{R}^{\mathrm{af}}[\omega(t)]&=
		 N\langle\Psi,p_{1}p_{2}p_{3}(-\Delta_1)p_{1}p_{2}p_{3}\Psi\rangle-\sum_{j=1}^N\braket{u_j}{(-\Delta)u_j}\label{eq:posterm}\\
		 & \quad+N(N-1)\langle\Psi,p_{1}p_{2}p_{3}v_{12}p_{1}p_{2}p_{3}\Psi\rangle-\sum_{j=1}^N\langle u_j,\tilde{v}u_j\rangle\nonumber\\
		 & \quad+N(N-1)(N-2)\langle\Psi,p_{1}p_{2}p_{3}w_{123}p_{1}p_{2}p_{3}\Psi\rangle-\sum_{j=1}^N\langle u_j,\tilde{w}u_j\rangle\nonumber\\
		 & \quad+\frac{g}{2}N^2(N-1)\langle\Psi,p_{1}p_{2}p_3|\mathbf{a}_N|^2_{12}p_{1}p_{2}p_{3}\Psi\rangle-\frac{g}{2}N\sum_{j=1}^N\langle u_j,\mathbf{a}_N^2[\rho]u_j\rangle\nonumber\\
& \quad+2N\Re\langle\Psi,p_{1}p_{2}p_{3}(-\Delta_1)\Phi\rangle \nonumber\\
	&\quad +2N(N-1)\Re\langle\Psi,p_{1}p_{2}p_{3}v_{12}\Phi\rangle\nonumber\\
		& \quad+2N(N-1)(N-2)\Re\langle\Psi,p_{1}p_{2}p_{3}w_{123}\Phi\rangle\nonumber\\
		&\quad +gN^2(N-1)\Re\langle\Psi,p_{1}p_{2}p_{3}|\mathbf{a}_N|^2_{12}\Phi\rangle\nonumber\\
		&\quad +\beta^2N(N-1)\braket{\Psi_N}{|\mathbf{a}_N|^2_{12}\Psi_N}\nonumber\\
		&\quad+N\braket{\Phi}{\left(-\Delta_1+(N-1)v_{12}+(N-1)(N-2)w_{123}+\frac{g}{2}N(N-1)|\mathbf{a}_N|^2_{12}\right)\Phi}.\nonumber
	\end{align}
The above can be re-written as
		\begin{align}
		E_N(t)-\mathcal{E}_{R}^{\mathrm{af}}[\omega]&=:E^{(\Delta)}+E^{(v)}+E^{(w)}+E^{(\mathbf{a})}\label{eq:posterm1}\\
		&\quad +R^{(\Delta)}+R^{(v)}+R^{(w)}+R^{(\mathbf{a})}+\beta^2N(N-1)\braket{\Psi_N}{|\mathbf{a}_N|^2_{12}\Psi_N}\nonumber\\
		&\quad +N\braket{\Phi}{\left(-\Delta_1+(N-1)v_{12}+(N-1)(N-2)w_{123}+\frac{g}{2}N(N-1)|\mathbf{a}_N|^2_{12}\right)\Phi}\nonumber
	\end{align}
	where the $E$ and $R$ terms are defined line by line with respect to the identity \eqref{eq:posterm}. At this point, we have obtained the identity of the lemma with $\mathcal{E}=E^{(\Delta)}+E^{(v)}+E^{(w)}+E^{(\mathbf{a})}+R^{(\Delta)}+R^{(v)}+R^{(w)}+R^{(\mathbf{a})}$. There is left to show that $\left|\mathcal{E}\right|\lesssim \alpha_m(t)+N^{-1/2}$. To this end, we will bound each of the energies $E^{(\cdot)}$ and the rests $R^{(\cdot)}$ in two distinct paragraphs. To do so, we will need to use the definition of  $$\hat{m}(\xi):=\sum_{k=0}^N\left(\frac{k}{N}\right)^{\xi}P^{(k)}_N$$ 
	and use that $\hat{m}(1/2)=\hat{m}(1/4)\hat{m}(1/4)$.\newline
	\textbf{Study} of the $E^{(\cdot)}$'s: we start with the Laplacian
	\begin{align*}
	\left|E^{(\Delta)}\right|&\leq \left|\langle\Psi,(N-2)p_{1}p_{2}p_{3}(-\Delta_1)p_{1}p_{2}p_{3}\Psi\rangle-\sum_{j=1}^N\braket{u_j}{(-\Delta)u_j}\right|\\
	&=\left|\langle\Psi,\sum_{j=1}^N\lambda_j^{(\Delta)}\sum_{\substack{m=1\\m\neq 2,3}}^Np_m^{\chi_j^{(\Delta)}}p_{2}p_{3}\Psi\rangle-\sum_{j=1}^N\lambda_j^{(\Delta)}\right|\\
	&=\left|\sum_{j=1}^N\lambda_j^{(\Delta)}\langle\Psi,\left(\one-\sum_{\substack{m=1\\m\neq 2,3}}^Np_m^{\chi_j^{(\Delta)}}p_{2}p_{3}\right)\Psi\rangle\right|\\
	&\leq\sqrt{\sum_{j=1}^{N}\left|\lambda_j^{(\Delta)}\right|^2\norm{\hat{m}(1/4)\Psi}^2}\sqrt{\langle \hat{m}(-1/2)\Psi,\left(N-\sum_{j=1}^N\sum_{\substack{m=1\\m\neq 2,3}}^Np_m^{\chi_j^{(\Delta)}}p_{2}p_{3}\right)\Psi\rangle}\\
	&\leq \sqrt{\alpha_m} \sqrt{\sum_{j=1}^N\norm{\Delta u_j}_2^2}\braket{\hat{m}(-1/2)\Psi}{\left(N-\sum_{\substack{m=1\\m\neq 2,3}}^Np_mp_{2}p_{3}\right)\Psi}^{1/2}\\
		&= \sqrt{\alpha_m}\sqrt{N}\braket{\hat{m}(-1/2)\Psi}{\left(N-(N-2)p_1p_{2}p_{3}\right)\Psi}^{1/2}\\
		&\leq \sqrt{\alpha_m}N\braket{\hat{m}(-1/2)\Psi}{\left(p_1p_{2}p_{3}-\one\right)\Psi}^{1/2}+2 \sqrt{\alpha_m}\sqrt{N}\braket{\Psi}{p_1p_{2}p_{3}\Psi}^{1/2}\\
		&\leq \sqrt{\alpha_m}N\braket{\Psi}{\hat{m}(-1/2)q_2\Psi}^{1/2}+2 \sqrt{\alpha_m}\sqrt{N}\braket{\Psi}{p_1p_{2}p_{3}\Psi}^{1/2}\\
		&\leq N\left(\alpha_m(t)+\frac{C}{\sqrt{N}}\right)
		\end{align*}
		using that $\one -p_1p_2p_3 =p_1p_2q_3+p_1q_2q_3+q_1p_2p_3 +p_1q_2p_3+q_1p_2q_3+q_1q_2p_3+q_1q_2q_3$, $\norm{\rho_{\Delta}}_1\lesssim N$ and that $$\hat{m}(-1/2)q_2=\sum_{k\in \mathbb{Z}}\frac{\sqrt{N}}{\sqrt{k}}\frac{k}{N}P_N^{(k)}=\hat{m}(1/2).$$
		We continue with the mixed two-body term
			\begin{align*}
	\left|E^{(v)}\right|&=\langle\Psi,p_{1}p_{2}N(N-1)v_{12}p_{1}p_{2}p_{3}\Psi\rangle -\sum_{j=1}^N\langle u_j,\tilde{v}u_j\rangle		\\
	&=\braket{\Psi}{\left(\sum_{m=1}^N\sum_{\substack{n=1\\n\neq m}}^Np_np_mv_{nm}p_mp_np_3-\sum_{j=1}^N\lambda_j^{(v)}\right)\Psi}\\
	&=\braket{\Psi}{\left(\sum_{m=1}^N\sum_{\substack{n=1\\n\neq m}}^N\sum_{j,k=1}^N\langle \sigma_k\braket{\chi_j}{v_{nm}\chi_j}_m\sigma_k\rangle_n p_n^{\sigma_k} p_m^{\chi_j}p_3-\sum_{j=1}^N\lambda_j^{(v)}\right)\Psi}
	\end{align*}
	where the basis $\sigma_i$ and $\chi_i$ have been introduced in Lemma \ref{lem:diag2b}. We then obtain
	\begin{align*}
	\left|E^{(v)}\right|&=\sum_{j,k=1}^N\braket{\langle \sigma_k\braket{\chi_j}{v_{nm}\chi_j}_m\sigma_k\rangle_n\hat{m}(1/4)\hat{m}(-1/4)\Psi}{\left(\one-\sum_{m=1}^N\sum_{\substack{n=1\\n\neq m}}^N p_n^{\sigma_k} p_m^{\chi_j}p_3\right)\Psi}\\
	&\leq \norm{\hat{m}(1/4)\Psi}\sqrt{\sum_{j,k=1}^N\left|\langle \sigma_k\braket{\chi_j}{v_{1m}\chi_j}_m\sigma_k\rangle_1\right|^2}\braket{\hat{m}(-1/2)\Psi}{\sum_{j,k=1}^N\left(\one-\sum_{m=1}^N\sum_{\substack{n=1\\n\neq m}}^N\ p_n^{\sigma_k} p_m^{\chi_j}p_3\right)\Psi}^{1/2}\\
	&\leq \sqrt{\alpha_m}\braket{\hat{m}(-1/2)\Psi}{\left(N^2-\sum_{m=1}^N\sum_{\substack{n=1\\n\neq m}}^N p_n p_mp_3\right)\Psi}^{1/2}\\
	&\leq \sqrt{\alpha_m}N\braket{\hat{m}(-1/2)\Psi}{ \left(p_1 p_2p_3-\one\right)\Psi}^{1/2}+\sqrt{\alpha_m}\sqrt{N}\braket{\hat{m}(-1/2)\Psi}{ p_1 p_2p_3\Psi}^{1/2}\\
	&\leq N \left(\alpha_m +CN^{-1/2}\right)
	\end{align*}
	by using that $p^{\phi}_mp_n^{\phi}=p_m^{\phi}\delta_{mn}$ and with the control of the trace
		\begin{align*}
	 \sum_{j,k}\langle \sigma_k\braket{\chi_j}{v_{1m}\chi_j}_m\sigma_k\rangle_1^2&= \sum_{j,k}\langle u_k\braket{u_j}{v_{1m}u_j}_mu_k\rangle_1^2\leq \sum_{j,k}\langle u_k\braket{u_j}{v^2_{12}u_j}_2 u_k\rangle_1\\
	 &\leq \sum_{k}\braket{\nabla u_k}{|\mathbf{a}_N|^2\ast \rho\nabla u_k}_2\leq \norm{|\mathbf{a}_N|^2\ast \rho}_{\infty}\sum_{j=1}^N\norm{\nabla u_j}_2^2\simeq 1
	\end{align*}
	by the application of the bound \eqref{ine:LRs}.
	We now treat the three-body term. there exist three basis $\chi_j$, $\sigma_j$ and $\omega_j$, such that
			\begin{align*}
	\left|E^{(w)}\right|&=\left|N(N-1)(N-2)\langle\Psi,p_{1}p_{2}p_{3}w_{123}p_{1}p_{2}p_{3}\Psi\rangle-\sum_{j=1}^N\langle u_j,\tilde{w}u_j\rangle\right|	\\
	&=\left|\braket{\Psi}{\left(\sum_{m=1}^N\sum_{\substack{o=1\\o\neq m}}^N\sum_{\substack{n=1\\n\neq m, n\neq o}}^Np_op_np_mw_{mno}p_mp_np_o-\sum_{j=1}^N\lambda_j^{(w)}\right)\Psi}\right|\\
	&=\left|\braket{\Psi}{\left(\sum_{m=1}^N\sum_{\substack{o=1\\o\neq m}}^N\sum_{\substack{n=1\\n\neq m, n\neq o}}^N\sum_{j,k,\ell=1}^N\braket{\chi_j\otimes \sigma_k\otimes \omega_{\ell}}{w_{mno}\chi_j\otimes \sigma_k\otimes \omega_{\ell} }p_n^{\sigma_k} p_m^{\chi_j}p_o^{\omega_{\ell}}-\sum_{j=1}^N\lambda_j^{(w)}\right)\Psi}\right|\\
	&=\left|\sum_{j,k,\ell=1}^N\braket{\braket{\chi_j\otimes \sigma_k\otimes \omega_{\ell}}{w_{mno}\chi_j\otimes \sigma_k\otimes \omega_{\ell} }\hat{m}(1/4)\hat{m}(-1/4)\Psi}{\left(\one-\sum_{m=1}^N\sum_{\substack{o=1\\o\neq m}}^N\sum_{\substack{n=1\\n\neq m, n\neq o}}^N p_n^{\sigma_k} p_m^{\chi_j}p_o^{\omega_{\ell}}\right)\Psi}\right|
	\end{align*}
	by using Cauchy-Schwarz's inequality. We then get that
	\begin{align*}
	\left|E^{(w)}\right|&\leq \norm{\hat{m}(1/4)\Psi}\sqrt{\sum_{j,k,\ell=1}^N\left|\braket{\chi_j\otimes \sigma_k\otimes \omega_{\ell}}{w_{mno}\chi_j\otimes \sigma_k\otimes \omega_{\ell} }_{m,n,o}\right|^2}\\
	&\quad \times\braket{\hat{m}(-1/2)\Psi}{\sum_{j,k,\ell=1}^N\left(\one-\sum_{m=1}^N\sum_{\substack{o=1\\o\neq m}}^N\sum_{\substack{n=1\\n\neq m, n\neq o}}^N p_n^{\sigma_k} p_m^{\chi_j}p_o^{\omega_{\ell}}\right)\Psi}^{1/2}\\
	&\leq \sqrt{\alpha_m}N^{-1/2}\braket{\hat{m}(-1/2)\Psi}{\left(N^3-\sum_{m=1}^N\sum_{\substack{o=1\\o\neq m}}^N\sum_{\substack{n=1\\n\neq m, n\neq o}}^N p_n p_mp_o\right)\Psi}^{1/2}\\
	&\leq \sqrt{\alpha_m}N\braket{\hat{m}(-1/2)\Psi}{ \left(p_1 p_2p_3-\one\right)\Psi}^{1/2}+ C\sqrt{\alpha_m}\sqrt{N}\braket{\hat{m}(-1/2)\Psi}{\Psi}^{1/2}\\
	&\leq N\left(\alpha_m +N^{-1/2}\right).
	\end{align*}
We continue with the singular two-body term
			\begin{align*}
	\left|E^{(\mathbf{a})}\right|&= \left|\frac{g}{2}N^2(N-1)\langle\Psi,p_{1}p_{2}p_3|\mathbf{a}_N|^2_{12}p_{1}p_{2}p_{3}\Psi\rangle-\frac{g}{2}N\sum_{j=1}^N\langle u_j,\mathbf{a}_N^2[\rho]u_j\rangle	\right|	\\
	&=\left|\frac{g}{2}N\braket{\Psi}{\left(\sum_{m=1}^N\sum_{\substack{n=1\\n\neq m}}^Np_np_m|\mathbf{a}_N|^2_{nm}p_mp_np_3-\sum_{j=1}^N\lambda_j^{(\mathbf{a}^2)}\right)\Psi}\right|\\
	&=\left|\frac{g}{2}N\braket{\Psi}{\left(\sum_{m=1}^N\sum_{\substack{n=1\\n\neq m}}^N\sum_{j,k=1}^N\langle \sigma_k\braket{\eta_j}{|\mathbf{a}_N|^2_{nm}\eta_j}_m\sigma_k\rangle_n p_n^{\sigma_k} p_m^{\eta_j}p_3-\sum_{j=1}^N\lambda_j^{(\mathbf{a}^2)}\right)\Psi}\right|\\
	&=\frac{g}{2}N\left|\sum_{j,k=1}^N\braket{\langle \sigma_k\braket{\eta_j}{|\mathbf{a}_N|^2_{nm}\eta_j}_m\sigma_k\rangle_n\hat{m}(1/4)\hat{m}(-1/4)\Psi}{\left(\one -\sum_{m=1}^N\sum_{\substack{n=1\\n\neq m}}^N p_n^{\sigma_k} p_m^{\eta_j}p_3\right)\Psi}\right|\\
	&\leq CN\norm{\hat{m}(1/4)\Psi}\sqrt{\sum_{j,k=1}^N\left|\langle \sigma_k\braket{\eta_j}{|\mathbf{a}_N|^2_{1m}\eta_j}_m\sigma_k\rangle_1\right|^2}\braket{\hat{m}(-1/2)\Psi}{\sum_{j,k=1}^N\left(\one -\sum_{m=1}^N\sum_{\substack{n=1\\n\neq m}}^N\ p_n^{\sigma_k} p_m^{\eta_j}p_3\right)\Psi}^{1/2}.\nonumber
	\end{align*}
	We now use that $p^{\phi}_mp_n^{\phi}=p_m^{\phi}\delta_{mn}$ and that
		\begin{align*}
	 \sum_{j,k}\langle \sigma_k\braket{\eta_j}{|\mathbf{a}_N|^2_{1m}\eta_j}_m\sigma_k\rangle_1^2&= \sum_{j,k}\langle u_k\braket{u_j}{|\mathbf{a}_N|^2_{1m}u_j}_mu_k\rangle_1^2\leq \sum_{j,k}\langle u_k\braket{u_j}{|\mathbf{a}_N|^4_{12}u_j}_2 u_k\rangle_1\\
	 &\leq \sum_{k}\braket{ u_k}{(|\mathbf{a}_N|^4\ast \rho) u_k}_2\\
	 &\leq \norm{|\mathbf{a}_N|^4\ast \rho}_{\infty}N\simeq N^{-2},
	\end{align*}
	to conclude
\begin{align*}
	&\leq C\sqrt{\alpha_m}\braket{\hat{m}(-1/2)\Psi}{\left(N^2 -\sum_{m=1}^N\sum_{\substack{n=1\\n\neq m}}^N p_n p_mp_3\one\right)\Psi}^{1/2}\nonumber\\
	&\leq CN\sqrt{\alpha_m}\braket{\hat{m}(-1/2)\Psi}{ \left(\one -p_1 p_2p_3\right)\Psi}^{1/2}+C\sqrt{N}\sqrt{\alpha_m}\braket{\hat{m}(-1/2)\Psi}{ \Psi}^{1/2}\\
	&\leq CN \left(\alpha_m +N^{-1/2}\right).
	\end{align*}	
\textbf{Study} of the $R^{(\cdot)}$'s, we start with
			\begin{align*}
	\left|R^{(\Delta)}\right|&\leq 2N\left|\langle\Psi,p_{1}p_{2}p_{3}(-\Delta_1)\Phi\rangle\right|	\\
	&=2N\left|\langle \Psi, p_{1}p_{2}p_{3}(-\Delta_1)\hat{m}(1/4)\hat{m}(-1/4)q_1p_2p_3\Psi\rangle\right|\\
	&\leq CN\norm{ \hat{m}(-1/4)q_1\Psi}\braket{\Psi}{ \hat{m}(1/2)p_{1}p_{2}p_{3}(-\Delta_1)p_1p_2p_3\Psi}^{1/2}\\
	&\leq CN\norm{ \hat{m}(1/4)\Psi}\braket{ \hat{m}(1/4)\Psi}{p_{1}(-\Delta_1)p_1 \hat{m}(1/4)\Psi}^{1/2}\\
	&\leq CN\norm{ \hat{m}(1/4)\Psi}^2\\
	&\leq CN\alpha_m(t)
		\end{align*}
		by using \eqref{ine:Delta}, $p_2p_3\leq \one$, $\hat{m}(-1/2)q_1=\hat{m}(1/2)$ and the Cauchy-Schwarz inequality.
		\begin{align*}
	\left|R^{(v)}\right|&\leq 2N(N-1)\left|\langle\Psi,p_{1}p_{2}p_{3}v_{12}\Phi\rangle\right|	\\
	&=2N(N-1)\left|\langle \Psi, p_{1}p_{2}p_{3}v_{12}\left(q_1p_2p_3+ p_1q_2p_3 +q_1q_2p_3 +q_1q_2q_3\right)\Psi\rangle\right|\\
&=2N(N-1)\left|\langle \hat{m}(1/4)v_{12}p_{1}p_{2}p_{3}\Psi, \hat{m}(-1/4) \left(q_1p_2p_3+ p_1q_2p_3 +q_1q_2p_3 +q_1q_2q_3\right)\Psi\rangle\right|\\
&\leq CN^2\braket{p_{3}\hat{m}(1/4)\Psi}{p_{1}p_{2}v^2_{12}p_{1}p_{2}p_{3}\hat{m}(1/4)\Psi}^{1/2}\norm{ \hat{m}(-1/4) \left(q_1p_2p_3+ p_1q_2p_3 +q_1q_2p_3 +q_1q_2q_3\right)\Psi}\\
&\leq CN\norm{\hat{m}(1/4)\Psi}\norm{\hat{m}(-1/4)q_1\Psi}\\
&\leq CN\alpha_m(t)
		\end{align*}	
		by using Cauchy-Schwarz, the estimate of Lemma \ref{ine:ppv2pp}, $p, q\leq \one$ and the fact that there is at least one $q$ in each factor of $\hat{m}(-1/4)$.
		The next rest is
			\begin{align*}
	\left|R^{(w)}\right|&\leq 2N(N-1)(N-2)\left|\langle\Psi,p_{1}p_{2}p_{3}w_{123}\Phi\rangle\right|	\\
	&=2N(N-1)(N-2)\left|\langle \Psi, p_{1}p_{2}p_{3}w_{123}\left(q_1p_2p_3+ p_1q_2p_3 +p_1p_2q_3+p_1q_2q_3+q_1p_2q_3+q_1q_2p_3 +q_1q_2q_3\right)\Psi\rangle\right|\\
&=2N(N-1)(N-2)\left|\langle \hat{m}(1/4)w_{123}p_{1}p_{2}p_{3}\Psi, \hat{m}(-1/4) \left(q_1p_2p_3+2p_1p_2q_3+p_1q_2q_3+2q_1q_2p_3 +q_1q_2q_3\right)\Psi\rangle\right|\\
&\leq CN^3\braket{\hat{m}(1/4)\Psi}{p_{1}p_{2}p_3w^2_{123}p_{1}p_{2}p_{3}\hat{m}(1/4)\Psi}^{1/2}\norm{ \hat{m}(-1/4)\left(q_1p_2p_3+2p_1p_2q_3+p_1q_2q_3+2q_1q_2p_3 +q_1q_2q_3\right)\Psi}\\
&\leq CN\norm{\hat{m}(1/4)\Psi}\norm{\hat{m}(-1/4)q_1\Psi}\\
&\leq CN\alpha_m(t)
		\end{align*}	
		by Cauchy-Schwarz's inequality, the estimate \eqref{ine:ppw} and the fact that there is at least one $q$ in each factor of $\hat{m}(-1/4)$. The last rest is
		\begin{align*}
	\left|R^{(\mathbf{a})}\right|&\leq CN^3\left|\langle\Psi,p_{1}p_{2}p_{3}|\mathbf{a}_N|^2_{12}\Phi\rangle\right|	\\
	&=CN^3\left|\langle \Psi, p_{1}p_{2}p_{3}|\mathbf{a}_N|^2_{12}\left(q_1p_2p_3+ p_1q_2p_3 +q_1q_2p_3 +q_1q_2q_3\right)\Psi\rangle\right|\\
&=CN^3\left|\langle \hat{m}(1/4)|\mathbf{a}_N|^2_{12}p_{1}p_{2}p_{3}\Psi, \hat{m}(-1/4) \left(q_1p_2p_3+ p_1q_2p_3 +q_1q_2p_3 +q_1q_2q_3\right)\Psi\rangle\right|\\
&\leq CN^3\braket{p_{3}\hat{m}(1/4)\Psi}{p_{1}p_{2}|\mathbf{a}_N|^4_{12}p_{1}p_{2}p_{3}\hat{m}(1/4)\Psi}^{1/2}\norm{ \hat{m}(-1/4) \left(q_1p_2p_3+ p_1q_2p_3 +q_1q_2p_3 +q_1q_2q_3\right)\Psi}\\
&\leq CN\norm{\hat{m}(1/4)\Psi}\norm{\hat{m}(-1/4)q_1\Psi}\\
&\leq CN\alpha_m(t)
		\end{align*}	
		by using Cauchy-Schwarz, the estimate \eqref{ine:pa4}, $p, q\leq \one$ and the fact that there is at least one $q$ in each factor of $\hat{m}(-1/4)$.
		\end{proof}
To be able to conclude our control on the kinetic energy of excited particles, we have to treat the second line of \eqref{eq:posterm1}. This is the content of the following lemma.
		\begin{lemma}[\textbf{Control of the kinetic energy of excited particles}]\mbox{}\\\label{lem:kinetic}
	Under the assumptions of Lemma \ref{lem:mf} we have that
	\begin{equation}
		\norm{\nabla_1 q_1\Psi_N}^2\lesssim \frac{\left| E_N(t)-\mathcal{E}_R^{\mathrm{af}}[\omega(t)] \right|}{N} +\alpha_m(t)+\frac{1}{\sqrt{N}}
	\end{equation}
	provided that $g\geq 12\beta^2$.
\end{lemma}
\begin{proof}
We start by noticing that
	\begin{align*}
		\|\nabla_{1}q_{1}\Psi_{N}(t)\| & =\|\nabla_{1}(1-p_{1}(p_{2}+q_{2})(p_{3}+q_{3}))\Psi_{N}(t)\|\\
		& \leq\|\nabla_{1}(1-p_{1}p_{2}p_{3})\Psi_{N}(t)\|+3\|\nabla_1 p_1q_{2}\Psi_{N}(t)\|\\
		& \leq\|\nabla_{1}\Phi\|+3\left\langle\Psi_{N}(t),\hat{m}(\tfrac12)\Psi_{N}(t)\right\rangle^{1/2}.
	\end{align*}
	by using the estimate \eqref{ine:Delta}. We will now control $\|\nabla_{1}\Phi\|$.
For any $0\leq \kappa\leq 1$, the application of Lemma \ref{lem:mf} provides
		\begin{align*}
		N\kappa\braket{\Phi}{\left(-\Delta_1\right)\Phi}&\leq (E_{\Psi_N}-\mathcal{E}_{R}^{\mathrm{af}}[\omega])+N\left(\alpha_m(t)+\frac{1}{\sqrt{N}}\right)\\
&\quad -N\braket{\Phi}{\left(-(1-\kappa)\Delta_1+(N-1)v_{12}+(N-1)(N-2)w_{123}+\frac{g}{2}N(N-1)|\mathbf{a}_N|^2_{12}\right)\Phi}
		\end{align*}
		where $\Phi=(\one -p_1p_2p_3)\Psi_N$ and where we used that $\braket{\Psi_N}{|\mathbf{a}_N|^2_{12}\Psi_N}\geq 0$ to drop it. We have to show that	
\begin{align}
\braket{\Phi}{\left(-(1-\kappa)\Delta_1+(N-1)v_{12}+(N-1)(N-2)w_{123}+\frac{g}{2}N(N-1)|\mathbf{a}_N|^2_{12}\right)\Phi}\geq 0\label{ine:pos}
\end{align}
and the proof will be complete. We fix $\kappa=1/4$. We use a weighted Cauchy-Schwarz inequality in order to bound our $v_{12}$ and the $w_{123}$ operators, we get
\begin{align}
\left|Nv_{12}\right|&\leq \frac{1}{4} (-\Delta_1) +4\beta^2N^2|\mathbf{a}_N|_{12}^2 \\
N^2\left|w_{123}\right|&\leq 2\beta^2N^2|\mathbf{a}_N|_{12}^2
\end{align}
by using the symmetry of $\Phi$ in the variables $x_1 ,x_2 ,x_3$. We use them to show that the operator of line \eqref{ine:pos} is larger than
\begin{equation}
 \frac{1}{4} (-\Delta_1)+\left(\frac{g}{2}-6\beta^2\right)N^2|\mathbf{a}_N|^2_{12}\geq \left(\frac{g}{2}-6\beta^2\right)
\end{equation}
which is positive under the constraint $g\geq 12\beta^2$. This concludes the proof.
\end{proof}


\appendix
\section[\qquad Appendix]{Appendix}\label{sec:appendix}
\subsection{Proofs of Lemmas \ref{lem:Ha} and Proposition \ref{pro:sla}}
Here we prove Proposition \ref{pro:sla}.
\begin{proof}\label{proof:slater}
Using the definitions we have
\begin{align*}
\omega_N^{(1)}(x_1;y_1)&=\frac{1}{N!}\sum_{\sigma\in S_N}\sum_{\sigma'\in S_N}{\mathrm{sgn}(\sigma)}{\mathrm{sgn}(\sigma')}\overline{\psi_{\sigma(1)}(x_1)}\psi_{\sigma'(1)}(y_1)\\
&\quad\int_{\R^{2(N-1)}}\overline{\psi_{\sigma(2)}(u_2)\psi_{\sigma'(2)}(u_2)\dots\psi_{\sigma(N)}(u_N)}\psi_{\sigma'(N)}(u_N)\mathrm{d}u_2\dots\d u_N\\
&=\frac{1}{N!}\sum_{\sigma , \sigma'\in S_N}\mathrm{sgn}(\sigma)\mathrm{sgn}(\sigma')\left(\prod_{j=2}^N\delta_{\sigma(j),\sigma'(j)}\right)\overline{\psi_{\sigma(1)}(x_1)}\psi_{\sigma'(1)}(y_1)\\
&=\frac{1}{N}\sum_{\sigma(1)=1}^N\overline{\psi_{\sigma(1)}(x_1)}\psi_{\sigma(1)}(y_1)
\end{align*}
where we have used that the delta's force $\sigma =\sigma'$. The same argument leads to
\begin{align*}
\omega_N^{(2)}(x_1,x_2;y_1,y_2)&=\frac{1}{N!}\sum_{\sigma\in S_N}\sum_{\sigma'\in S_N}{\mathrm{sgn}(\sigma)}{\mathrm{sgn}(\sigma')}\overline{\psi_{\sigma(1)}(x_1)}\psi_{\sigma'(1)}(y_1)\overline{\psi_{\sigma(1)}(x_2)}\psi_{\sigma'(1)}(y_2)\\
&\quad\int_{\R^{2(N-2)}}\overline{\psi_{\sigma(3)}(u_3)\psi_{\sigma'(3)}(u_3)\dots\psi_{\sigma(N)}(u_N)}\psi_{\sigma'(N)}(u_N)\mathrm{d}u_3\dots\d u_N\\
&=\frac{1}{N!}\sum_{\sigma , \sigma'\in S_N}\mathrm{sgn}(\sigma)\mathrm{sgn}(\sigma')\left(\prod_{j=3}^N\delta_{\sigma(j),\sigma'(j)}\right)\overline{\psi_{\sigma(1)}(x_1)}\psi_{\sigma'(1)}(y_1)\overline{\psi_{\sigma(2)}(x_2)}\psi_{\sigma'(2)}(y_2)\\
&=\frac{1}{N(N-1)}\sum_{j=1}^N\sum_{k=1}^N\left(\overline{\psi_{j}(x_1)}\psi_{j}(y_1)\overline{\psi_{k}(x_2)}\psi_{k}(y_2)-\overline{\psi_{j}(x_1)}\psi_{k}(y_1)\overline{\psi_{k}(x_2)}\psi_{j}(y_2)\right)\\
&=\frac{1}{N(N-1)}\sum_{j=1}^N\sum_{k=1}^N\overline{\psi_j(x_1)\psi_k(x_2)}(\one -U_{jk})\psi_j(y_1)\psi_k(y_2)
\end{align*}
where the delta's have fixed $\sigma(j)=\sigma'(j)$ for $j=3,\dots,N$ and left the cases $\sigma(1)=\sigma'(1)$ or $\sigma(1)=\sigma'(2)$ which only differ in one transmutation. Note that, in the above sums, the term $j=k$ cancels. The last computation involves in a similar way
\begin{equation}
\overline{\psi_{\sigma(1)}(x_1)}\psi_{\sigma'(1)}(y_1)\overline{\psi_{\sigma(2)}(x_2)}\psi_{\sigma'(2)}(y_2)\overline{\psi_{\sigma(3)}(x_3)}\psi_{\sigma'(3)}(y_3)
\end{equation}
in which we can have the six combinations leading to $(\one -U_{jk}-U_{jm}-U_{km}+U_{jk}U_{km}+U_{km}U_{jm})$.
\end{proof}

Here we prove Lemma \ref{lem:Ha}.
\begin{proof}\label{proof:Hartree}
We evaluate $\big <\Psi^{\mathrm{SL}} ,H_{N,R}\Psi^{\mathrm{SL}}\big >=\big <H_{N,R}\big >_{\mathrm{SL}}$ term by term like in $\eqref{expanded_H}$. Until the end of this section $\langle\cdot\rangle$ means $\langle\cdot\rangle_{\mathrm{SL}}$.
We will use the following notation:
\begin{align}
&W_{1}=\left (-\im\nabla_1\right )^{2}\nonumber\\
&W_{12}=(-\im\nabla_1)\cdot\mathbf{a}_{N}(x_1-x_2)+\mathbf{a}_{N}(x_1 -x_2)\cdot(-\im\nabla_1)\nonumber\\
&W_{123}=\mathbf{a}_{N}(x_{1}-x_{2})\cdot\mathbf{a}_{N}(x_{1}-x_{3})\nonumber
\end{align}

We use that $\omega$ has a kernel $\omega(x,y)$ such that $\norm{\rho_{\omega}}_1=\int_{\R^2}\omega(x,x)=N$ and $\int_{\R^4}\left|\omega(x,y)\right|^2 =N$ to extract the direct term corresponding to the Hartree energy plus the exchange terms that we bound as errors.
\begin{align}
\big <H_{N,R}\big >_{\mathrm{SL}}=\mathcal{E}^{\mathrm{af}}_{R}[\omega]&-\int_{\R^{4}}\omega (x_{1},x_{2})W_{12}\overline{\omega (x_{1},x_{2}})\d x_{1}\d x_{2}\nonumber\\
&+\int_{\R^{4}}\left|\mathbf{a}_{N}(x_{1}-x_{2})\right|^{2}\Big (\rho_{\omega }(x_{1})\rho_{\omega }(x_{2})
-\b \omega (x_{1},x_{2})\b^{2}\Big )\d x_{1} \d x_{2}\nonumber\\
&-N\frac{g}{2}\int_{\R^{4}}\left|\mathbf{a}_{N}(x_{1}-x_{2})\right|^{2}
\b \omega (x_{1},x_{2})|^2d x_{1} \d x_{2}\nonumber\\
&-\int_{\R^{6}}W_{123}\Big [\rho_{\omega} (x_{1})\b\omega(x_{2},x_{3})\b^{2}+\rho_{\omega} (x_{2})\b\omega(x_{1},x_{3})\b^{2}+\rho_{\omega} (x_{3})\b\omega(x_{1},x_{2})\b^{2}\Big ]\d x_{1}\d x_{2}\d x_{3}\nonumber\\
&+\int_{\R^{6}}W_{123}\Big [2\Re\Big (\omega(x_{1},x_{2})\omega(x_{2},x_{3})\omega(x_{3},x_{1})\Big )\Big ]\d x_{1}\d x_{2}\d x_{3}\label{eq:Hartreeerr}
\end{align}
We have four errors terms to bound.For the first of them, we use
\begin{equation}
\left|W_{12}\right|\leq 2\varepsilon N^{-1}(-\Delta_1)+2N\varepsilon^{-1} |\mathbf{a}_{N}|^2(x_1 -x_2)\nonumber
\end{equation}
to bound
\begin{align}
\left|\int_{\R^{4}}\omega (x_{1},x_{2})W_{12}\overline{\omega (x_{1},x_{2}})\d x_{1}\d x_{2}\right|&\lesssim N^{-1}\varepsilon\tr [-\Delta \omega] +\varepsilon^{-1}NN^{-2 +s}R^{-2s}N\nonumber\\
&\lesssim \varepsilon +\varepsilon^{-1}N^{s(1+2r)}\lesssim  N^{\frac{s}{2}(1+2r)}\nonumber
\end{align}
by picking $\varepsilon=N^{\frac{s}{2}(1+2r)}$.
The second error gives
\begin{align*}
&\left| \int_{\R^{4}}\left|\mathbf{a}_{N}(x_{1}-x_{2})\right|^{2}\Big (\rho_{\omega }(x_{1})\rho_{\omega }(x_{2})
-\b \omega (x_{1},x_{2})\b^{2}\Big )\d x_{1} \d x_{2}\right|\\
&\quad\quad\quad\quad\quad\quad\leq \norm{(|\mathbf{a}_{N}|^2\ast\rho)\rho}_{1}+N\norm{\mathbf{a}_{N}}^2_{\infty}\lesssim 1
\end{align*}
by \eqref{ine:LRs} and under $s\leq 1/2$ and $r\leq 1/2$. The third one is the worst and is bounded by
\begin{align}
N\norm{\mathbf{a}_{N}}_{\infty}^2\int_{\R^4}|\omega(x,y)|^2\leq N^{s(1+2r)}.
\end{align}
For the next term we use
\begin{align}
\left|\int_{\R^{6}}W_{123}\rho_{\omega} (x_{1})\b\omega(x_{2},x_{3})\b^{2}\d x_{1}\d x_{2}\d x_{3}\right|\leq N\norm{|\mathbf{a}_N|^2\ast\rho}_{\infty}\lesssim 1\nonumber
\end{align}
by \eqref{ine:LRs} and similarly that
\begin{align*}
\left|\int_{\R^{6}}W_{123}\rho_{\omega} (x_{2})\b\omega(x_{1},x_{3})\b^{2}\d x_{1}\d x_{2}\d x_{3}\right|&\leq N\sup_{x_1 , x_3}\int_{\R^6}\left(\varepsilon^{-1} |\mathbf{a}_{N}|^2(x_1 -x_2) +\varepsilon|\mathbf{a}_{N}|^2(x_1 -x_3)\right)\rho(x_2)\d x_{2}\nonumber\\
&\lesssim N\varepsilon^{-1}\norm{|\mathbf{a}_N|^2\ast\rho}_{\infty}+ N^2\varepsilon N^{-2+s}R^{2s}\nonumber\\
&\lesssim \varepsilon^{-1}+\varepsilon N^{s(1+2r)}\lesssim  N^{\frac{s}{2}(1+2r)}.\nonumber
\end{align*}
The last term of \eqref{eq:Hartreeerr} is treated as above, using that 
\begin{equation}
\left|\omega(x_{1},x_{2})\omega(x_{2},x_{3})\omega(x_{3},x_{1})\right|\lesssim \left|\omega(x_{1},x_{2})\right|^2\rho(x_3) +\left|\omega(x_{2},x_{3})\right|^2\rho(x_1).\nonumber
\end{equation}
\end{proof}

Here we prove Lemma \ref{lem:comw}
\begin{proof}\label{pro:comm}
We use that
\begin{equation}
P_N^{(k)}=\sum_{d=1}^nP^{(d)}_{1\dots n}P_{n+1\dots N}^{(k-d)}
\end{equation}
in which $P^{(d)}_{1\dots n}$ contains $d$ projectors $q$ and $P_{n+1\dots N}^{(k-d)}$ contains $k-d$ of them. We can then use that $P^{(b)}_{1\dots n}P^{(d)}_{1\dots n}=\delta_{db}P^{(b)}_{1\dots n}$ and $P^{(a)}_{1\dots n}P^{(d)}_{1\dots n}=\delta_{da}P^{(a)}_{1\dots n}$ to get
\begin{align*}
\left(P^{(a)}_{1\dots n}h_{1\dots n}P^{(b)}_{1\dots n}\right)\hat{f}&=\sum_{k\in\mathbb{Z}}^N f(k)P^{(a)}_{1\dots n}h_{1\dots n}P^{(b)}_{1\dots n}\sum_{d=1}^nP^{(d)}_{1\dots n}P_{n+1\dots N}^{(k-d)}\\
&=\sum_{k\in\mathbb{Z}}^N f(k)P^{(a)}_{1\dots n}h_{1\dots n}P^{(b)}_{1\dots n}P_{n+1\dots N}^{(k-b)}\\
&=\sum_{k\in\mathbb{Z}}^N f(k)P_{n+1\dots N}^{(k-b)}P^{(a)}_{1\dots n}h_{1\dots n}P^{(b)}_{1\dots n}\\
&=\sum_{k\in\mathbb{Z}}^N f(k)\sum_{d=0}^nP_{1\dots n}^{d}P_{n+1\dots N}^{(k+a-b-d)}P^{(a)}_{1\dots n}h_{1\dots n}P^{(b)}_{1\dots n}\\
&=\sum_{k\in\mathbb{Z}}^N f(k)P_{N}^{(k+a-b)}\left(P^{(a)}_{1\dots n}h_{1\dots n}P^{(b)}_{1\dots n}\right)\\
&=\sum_{m\in\mathbb{Z}}^N f(m+b-a)P_{N}^{(m)}\left(P^{(a)}_{1\dots n}h_{1\dots n}P^{(b)}_{1\dots n}\right)\\
&=\widehat{\tau_{b-a}f}\left(P^{(a)}_{1\dots n}h_{1\dots n}P^{(b)}_{1\dots n}\right).
\end{align*}
The seconde relation of the lemma follows the exact same way.
\end{proof}

\bibliographystyle{siam}
\bibliography{biblio_June2025.bib}
\end{document}